\documentclass{article}
\usepackage[sort&compress,square,sort,comma,numbers]{natbib}

\usepackage[inline]{enumitem}
\usepackage[normalem]{ulem}
\usepackage{algorithm,algorithmic,amsmath,amssymb,amsthm,tikz,url}
\usepackage{mathrsfs} 
\usetikzlibrary{arrows.meta,decorations.pathmorphing,positioning,backgrounds}

\newtheorem{theorem}{Theorem}

\newtheorem{proposition}[theorem]{Proposition}
\newtheorem{corollary}[theorem]{Corollary}
\newtheorem{lemma}[theorem]{Lemma}
\theoremstyle{definition}
\newtheorem{definition}[theorem]{Definition}
\theoremstyle{remark}
\newtheorem{example}[theorem]{Example}

\newcommand{\Graph}{\mathcal{G}}
\newcommand{\Tree}{\mathcal{T}}
\newcommand{\Chain}{\mathcal{C}}
\newcommand{\Nodes}{\mathbf{N}}
\newcommand{\ELabels}{\mathbf{E}}
\newcommand{\Alphabet}{\Sigma}
\newcommand{\Get}[2]{#2\langle#1\rangle}

\newcommand{\Label}{\ell}
\newcommand{\Relation}{\mathrm{R}}
\newcommand{\Edges}{\mathcal{E}}
\DeclareMathOperator{\depth}{depth}
\DeclareMathOperator{\distance}{distance}

\newcommand{\convop}{{{}^{-1}}}
\newcommand{\proop}[1][]{{\pi_{#1}}}
\newcommand{\coproop}[1][]{\overline{\pi}_{#1}}
\newcommand{\transop}{{{}^{+}}}
\newcommand{\transstarop}{{{}^{*}}}

\newcommand{\all}{\mathrm{all}}
\newcommand{\diversity}{\mathrm{di}}
\newcommand{\identity}{\mathrm{id}}
\newcommand{\compose}{\circ}
\newcommand{\union}{\cup}
\newcommand{\intersect}{\cap}
\newcommand{\difference}{-}
\newcommand{\converse}[1]{\smash{\left[{#1}\right]^{-1}}}
\newcommand{\transitive}[1]{\smash{\left[{#1}\right]^{+}}}
\newcommand{\transitivestar}[1]{\smash{\left[{#1}\right]^{*}}}
\newcommand{\project}[2]{\proop[#1]\mathord{\left[{#2}\right]}}
\newcommand{\coproject}[2]{\coproop[#1]\mathord{\left[{#2}\right]}}

\newcommand{\cond}{c}
\newcommand{\expr}{e}
\newcommand{\cexpr}[1]{\bullet(#1)}
\newcommand{\Apply}[2]{{#1}\mathord{\left\langle#2\right\rangle}}

\newcommand{\esize}[1]{\lvert#1\rvert}

\newcommand{\Conditions}{C}
\newcommand{\src}{\mathsf{s}}
\newcommand{\tgt}{\mathsf{t}}

\newcommand{\Automaton}{\mathcal{A}}
\newcommand{\States}{S}
\newcommand{\Initials}{I}
\newcommand{\Finals}{F}
\newcommand{\Transitions}{\delta}
\newcommand{\StateConditions}{\gamma}
\newcommand{\nIdentity}{\text{\sout{id}}}

\newcommand{\Lang}{\mathcal{N}}
\newcommand{\Fragment}{\mathfrak{F}}
\newcommand{\LeExpr}{\prec}
\newcommand{\LeqExpr}{\preceq}
\newcommand{\nLeqExpr}{\npreceq}

\renewcommand{\path}[1]{\mathrel{{#1}_{\textit{p}}}}
\newcommand{\bool}[1]{\mathrel{{#1}_{\textit{b}}}}
\newcommand{\BASE}[1]{\underline{#1}}

\DeclareMathOperator{\depthcond}{cdepth}
\newcommand{\notstate}{\rho}
\DeclareMathOperator{\weightcond}{cweight}

\newcommand{\condc}[1]{\operatorname{ccompl}(#1)}

\newcommand{\dcompl}[1]{\overline{#1}_{\mathord{\downarrow}}}

\newcommand{\GETS}{:=}
\newcommand{\VAR}[1]{\textit{#1}}

\newcommand{\PowerSet}[1]{\mathcal{P}(#1)}
\newcommand{\abs}[1]{\lvert#1\rvert}

\newcommand{\NatPlus}{\mathbb{N}^+}
\newcommand{\DOT}{\vphantom{o}\cdot}
\newcommand{\true}{\mathop{\texttt{true}}}
\newcommand{\false}{\mathop{\texttt{false}}}

\tikzset{
        >=Stealth,
        hasse_style/.style={
            xscale=1.3,
            yscale=1.7,
            baseline=(current bounding box.center),
            dot/.style={circle,scale=0.35,draw=black,fill=black}
        },
        tree_style/.style={
            baseline=(current bounding box.center),
            dot/.style={circle,scale=0.35,draw=black,fill=black}
        },
        automata_style/.style={
            baseline=(current bounding box.center),
            state/.append style={circle,draw,minimum size=0.725cm]},
            accepting/.append style={double},
            label/.append style={above,font=\small},
            idset/.append style={font=\footnotesize},
            conditions/.append style={font=\scriptsize},
            node distance=0.6cm
        }
}

\begin{document}

\title{Comparing Downward Fragments of the Relational Calculus with Transitive Closure on Trees\footnote{This is a revised and extended version of the conference paper `Relative Expressive Power of Downward Fragments of Navigational Query Languages on Trees and Chains', presented at the 15th International Symposium on Database Programming Languages, Pittsburgh, Pennsylvania, United States (DBPL 2015). Yuqing Wu carried out part of her work during a sabbatical visit to Hasselt University with a Senior Visiting Postdoctoral Fellowship of the Research Foundation Flanders (FWO).}}

\author{Jelle Hellings
            \footnote{Hasselt University, Martelarenlaan 42, Hasselt, Belgium.}
    \and Marc Gyssens\footnotemark[2]
    \and Yuqing Wu
            \footnote{Pomona College, 185 E 6th St., Claremont, CA, USA.}
    \and Dirk Van Gucht
            \footnote{Indiana University, Lindley Hall 215, 150 S.~Woodlawn Ave., Bloomington, IN, USA.}
    \and Jan Van den Bussche\footnotemark[2]
    \and Stijn Vansummeren
            \footnote{Universit\'e Libre de Bruxelles, Avenue Franklin Roosevelt 50, Brussels, Belgium.}
    \and George H. L. Fletcher
            \footnote{Eindhoven University of Technology, Den Dolech 2, Eindhoven, The Netherlands.}}

\date{}

\maketitle

\begin{abstract}
Motivated by the continuing interest in the tree data model, we study the expressive power of downward navigational query languages on trees and chains.  Basic navigational queries are built from the identity relation and edge relations using composition and union. We study the effects on relative expressiveness when we add transitive closure, projections, coprojections, intersection, and difference; this for boolean queries and path queries on labeled and unlabeled structures. In all cases, we present the complete Hasse diagram. In particular, we establish, for each query language fragment that we study on trees, whether it is closed under difference and intersection.

\end{abstract}

\section{Introduction}

Many relations between data can be described in a hierarchical way, including taxonomies such as the taxonomy of species studied by biologists, corporate hierarchies, and file and directory structures. A logical step is to represent these data using a tree-based data model. It is therefore not surprising that tree-based data models were among the first used in commercial database applications, the prime example being the hierarchical data model introduced in the 1960s~\cite{hierarchy}.  Since the 1970s, other data models, such as the relational data model~\cite{codd}, replaced the hierarchical data model almost completely. Interest in tree-based data models revived in the 1990s by the introduction of XML~\cite{xml}, which allowed for unstructured and semi-structured tree data, and, more recently, by JSON, as used by several NoSQL and relational database products~\cite{json}.

Observe that tree-based data models are special cases of graph-based data models. In practice, query languages for trees and graphs usually rely on navigating the structure to find the data of interest. Examples of the focus on navigation can be found in XPath~\cite{xpath,condxpath,xpath_tc,xpath_leashed}, SPARQL~\cite{rdf,sparql}, and the regular path queries (RPQs)~\cite{rpq}. The core navigational power of these query languages can be captured by fragments of the calculus of relations, popularized by Tarski, extended with transitive closure~\cite{tarski,givant}. In the form of the navigational query languages of Fletcher et al.~\cite{graph_icdt}, the relative expressive power of these fragments have been studied in full detail on graph-structured data~\cite{graph_amai,graph_navjr,graphjournal_tc}. Much less is known for the more restrictive tree data model, however. Notice in particular that the separation results on graphs of Fletcher et al.~do not necessarily also apply to trees. In addition, the expressiveness results for several XPath fragments~\cite{struct_xpath, condxpath,nav_xpath,xpath_tc,nav_xpath_calcalg,xpathalgebrajr,pospathtree} in the context of XML do not provide a complete picture of the relative expressive power of the navigational query languages we consider here.  As a first step towards a complete picture of the relative expressive power, we study the expressive power of downward fragments: these are navigational query languages that only allow downward navigation in the tree via parent-child relations. Downward navigation plays a big role in practical data retrieval from tree data. In the JSON data model, for example, most data retrieval is done by explicit top-down traversal of a data structure representation of the JSON data. Even in more declarative settings, such as within the PostgreSQL relational database system, the JSON query facilities primary aim at downward navigation.\footnote{For details on what PostgreSQL provides, we refer to \url{https://www.postgresql.org/docs/9.6/static/functions-json.html}. Observe that all basic arrow operators provided by PostgreSQL perform, in essence, downward navigation.} This focus on downward navigation is also found outside the setting of tree data. As an example, we mention nested relational database models that use downward navigation as an important tool to query the data (see, e.g.~\cite{colby}).

All downward fragments we consider in this paper can express queries by building binary relations from the edge relations and the identity relation ($\identity$), using composition ($\compose$) and union ($\union$). We study the effect on the expressive power of the presence of transitive closure ($\transop$); projections ($\proop$), which can be used to express conditions similar to the node-expressions in XPath~\cite{condxpath} and the branching operator in nested RPQs~\cite{rpq}; coprojections  ($\coproop$), which can be used to express negated conditions; intersection ($\intersect$); and difference  ($\difference$). In other words, we consider all query languages having at least the features $\identity$, $\compose$, and $\union$  of the downward navigational query language features we mentioned above. For these fragments, we study relative expressiveness for both path queries, which  evaluate to a set of node pairs, and boolean queries, which evaluate to true or false. We consider not only labeled trees, but also unlabeled trees and labeled and unlabeled chains, the reason being that most query languages are easier to analyze on these simpler structures and inexpressiveness results obtained on them can then be bootstrapped to the more general case. 

For all the cases we consider, we are able to present the complete Hasse diagram of relative expressiveness; these Hasse diagrams are summarized in Figure~\ref{fig:main_results}. In several cases, we are able to argue that pairs of downward fragments of the navigational query languages that are not equivalent in expressive power when used to query graphs, are already not equivalent in expressive power on the simplest of graphs: labeled or unlabeled chains. Hence, for these languages, we actually strengthen the results of Fletcher et al.~\cite{graph_icdt}.

In the cases where graphs and trees yield different expressiveness results, we are able to prove collapse results. In particular, we are able to establish, for each fragment of the navigational query languages that we study, whether it is closed under difference and intersection when applied on trees: adding intersection to a downward fragment of the navigational query languages never changes the expressive power, and adding difference only adds expressive power when $\proop$ is present and $\coproop$ is not present, in which case difference only adds the ability to express $\coproop$. To prove these closure results, we develop a novel technique based on finite automata~\cite{automata}, which we adapt to a setting with conditions. We use these condition automata to represent and manipulate navigational queries, with the goal to replace $\intersect$ and $\difference$ operations. We also use these condition automata to show that, in the boolean case, $\proop$ never adds expressive power when querying labeled chains. Finally, using homomorphism-based techniques, we show that, in the boolean case on unlabeled trees and unlabeled chains, only fragments with the non-monotone operator $\coproop$ can express queries that are not equivalent to queries of the form \textit{the height of the tree is at least $k$}.

Our study of the relative expressive power of the downward fragments of the navigational query languages on trees also has practical ramifications. If, for example, two language fragments are equivalent, then this leads to a choice in query language design. On the one hand, one can choose a smaller set of operators that, due to its simplicity, is easier to implement and optimize, even when dealing with big data in a distributed setting or when using specialized hardware. On the other hand, a bigger set of operators allows for easier query writing by the end users. Indeed, if one is only interested in boolean queries on unlabeled trees, then RPQs are much harder to evaluate than queries of the form \textit{the height of the tree is at least $k$}, although our results indicate that these query languages are, in this case, equivalent. Moreover, all our collapse results are constructive: we present ways to rewrite queries using operators such as $\intersect$ and $\difference$ into queries that do not rely on these operators. Hence, our results can be used as a starting point for automatic query rewriting and optimization techniques that, depending on the hardware, the data size, and the data type, choose an appropriate query evaluation approach.

This is a revised and extended version of Hellings et al.~\cite{dbpl}, to which we added full proofs of all the expressivity results. In addition, we generalize the boolean collapse to queries of the form \textit{the height of the tree is at least $k$} to also cover non-downward operations. 

\paragraph{Organization}
In Section~\ref{sec:prelim}, we introduce the basic notions and terminology used throughout this paper. In Section~\ref{sec:expressive_power}, we present our results on the relative expressive power of the downward navigational query languages, as well as some generalizations of these. The results and their generalizations are visualized in the Hasse diagrams of relative expressiveness shown in Figure~\ref{fig:main_results}. Observe that these diagrams include collapses involving diversity and converse, which are non-downward. In Section~\ref{sec:related}, we discuss related work. In Section~\ref{sec:conclusion}, we summarize our findings and propose directions for future work.

\begin{figure}[htb!]
    \centering
    \makebox[\textwidth][c]{
    \begin{tabular}{c|rrr}
                    &\multicolumn{2}{c}{\textbf{Boolean queries}} &\multicolumn{1}{c}{\textbf{Path queries}}\\
                    &\multicolumn{1}{c}{\textbf{Chains}} &\multicolumn{1}{c}{\textbf{Trees}}&\multicolumn{1}{c}{\textbf{Chains and Trees}}\\
             \hline
             \hline
    {\rotatebox[origin=c]{90}{\textbf{Labeled}}}
        &
                \begin{tikzpicture}[hasse_style]
                    \node[dot] (n1) at (0, 0) {};
                    \node[dot] (n2) at (1, 0.25) {} edge[<-] (n1);
                    \node[dot] (n5) at (0, 2) {} edge[<-] (n1);
                    \node[dot] (n6) at (1, 2.25) {} edge[<-] (n2) edge [<-] (n5);

                    \node (n1a) [left,align=right] at (n1) {$\Lang()$\\$\Lang(\BASE{\intersect})$\\$\Lang(\BASE{\difference})$\\$\Lang(\BASE{\proop})$\\$\Lang(\BASE{\proop, \intersect})$};
                    \node (n5a) [left,align=right] at (n5) {$\Lang(\BASE{\coproop})$\\$\Lang(\BASE{\coproop, \intersect})$\\$\Lang(\BASE{\coproop, \difference})$};
                    \node (n2a) [right,align=left] at (n2) {$\Lang(\BASE{\transop})$\\$\Lang(\BASE{\transop, \intersect})$\\$\Lang(\BASE{\transop, \difference})$\\$\Lang(\BASE{\transop, \proop})$\\$\Lang(\BASE{\transop, \proop, \intersect})$};
                    \node (n6a) [right,align=left] at (n6) {$\Lang(\BASE{\transop, \coproop})$\\$\Lang(\BASE{\transop, \coproop, \difference})$};
                \end{tikzpicture}
        &
                \begin{tikzpicture}[hasse_style]
                    \node[dot] (n1) at (0, 0) {};
                    \node[dot] (n2) at (1, 0.25) {} edge[<-] (n1);
                    \node[dot] (n3) at (0, 1) {} edge[<-] (n1);
                    \node[dot] (n4) at (1, 1.25) {} edge[<-] (n2) edge [<-] (n3);
                    \node[dot] (n5) at (0, 2) {} edge[<-] (n3);
                    \node[dot] (n6) at (1, 2.25) {} edge[<-] (n4) edge [<-] (n5);

                    \node (n1a) [left,align=right] at (n1) {$\Lang()$\\$\Lang(\BASE{\intersect})$\\$\Lang(\BASE{\difference})$};
                    \node (n3a) [left,align=right] at (n3) {$\Lang(\BASE{\proop})$\\$\Lang(\BASE{\proop, \intersect})$};
                    \node (n5a) [left,align=right] at (n5) {$\Lang(\BASE{\coproop})$\\$\Lang(\BASE{\coproop, \intersect})$\\$\Lang(\BASE{\coproop, \difference})$};
                    \node (n2a) [right,align=left] at (n2) {$\Lang(\BASE{\transop})$\\$\Lang(\BASE{\transop, \intersect})$\\$\Lang(\BASE{\transop, \difference})$};
                    \node (n4a) [right,align=left] at (n4) {$\Lang(\BASE{\transop, \proop})$\\$\Lang(\BASE{\transop, \proop, \intersect})$};
                    \node (n6a) [right,align=left] at (n6) {$\Lang(\BASE{\transop, \coproop})$\\$\Lang(\BASE{\transop, \coproop, \difference})$};
                \end{tikzpicture}
        &
                \begin{tikzpicture}[hasse_style]
                    \node[dot] (n1) at (0, 0) {};
                    \node[dot] (n2) at (1, 0.25) {} edge[<-] (n1);
                    \node[dot] (n3) at (0, 1) {} edge[<-] (n1);
                    \node[dot] (n4) at (1, 1.25) {} edge[<-] (n2) edge [<-] (n3);
                    \node[dot] (n5) at (0, 2) {} edge[<-] (n3);
                    \node[dot] (n6) at (1, 2.25) {} edge[<-] (n4) edge [<-] (n5);

                    \node (n1a) [left,align=right] at (n1) {$\Lang()$\\$\Lang(\BASE{\intersect})$\\$\Lang(\BASE{\difference})$};
                    \node (n3a) [left,align=right] at (n3) {$\Lang(\BASE{\proop})$\\$\Lang(\BASE{\proop, \intersect})$};
                    \node (n5a) [left,align=right] at (n5) {$\Lang(\BASE{\coproop})$\\$\Lang(\BASE{\coproop, \intersect})$\\$\Lang(\BASE{\coproop, \difference})$};
                    \node (n2a) [right,align=left] at (n2) {$\Lang(\BASE{\transop})$\\$\Lang(\BASE{\transop, \intersect})$\\$\Lang(\BASE{\transop, \difference})$};
                    \node (n4a) [right,align=left] at (n4) {$\Lang(\BASE{\transop, \proop})$\\$\Lang(\BASE{\transop, \proop, \intersect})$};
                    \node (n6a) [right,align=left] at (n6) {$\Lang(\BASE{\transop, \coproop})$\\$\Lang(\BASE{\transop, \coproop, \difference})$};
                \end{tikzpicture}
        \\
    \hline
    {\rotatebox[origin=c]{90}{\textbf{Unlabeled}}}
        &
                \begin{tikzpicture}[hasse_style]
                    \node[dot] (n1) at (0, 0) {};
                    \node[dot] (n2) at (0, 2) {} edge[<-] (n1);
                    \node[dot] (n3) at (1, 2.25) {} edge[<-] (n2);

                    \node (n1a) [below,align=center] at (n1) {$\Lang(\BASE{\Fragment}), \Fragment \subseteq \{\diversity, \convop,\transop,\proop,\intersect\}$\\$\Lang(\BASE{\Fragment}), \Fragment \subseteq \{\transop,\intersect,\difference\}$};
                    \node (n2a) [left,align=right] at (n2) {$\Lang(\BASE{\coproop})$\\$\Lang(\BASE{\coproop, \intersect})$\\$\Lang(\BASE{\coproop, \difference})$};
                    \node (n3a) [right,align=left] at (n3) {$\Lang(\BASE{\transop, \coproop})$\\$\Lang(\BASE{\transop, \coproop, \difference})$};
                \end{tikzpicture}
        &
                \begin{tikzpicture}[hasse_style]
                    \node[dot] (n1) at (0, 0) {};
                    \node[dot] (n2) at (0, 2) {} edge[<-] (n1);
                    \node[dot] (n3) at (1, 2.25) {} edge[<-] (n2);

                    \node (n1a) [below,align=center] at (n1) {$\Lang(\BASE{\Fragment}), \Fragment \subseteq \{\convop,\transop,\proop,\intersect\}$\\$\Lang(\BASE{\Fragment}), \Fragment \subseteq \{\transop,\intersect,\difference\}$};
                    \node (n2a) [left,align=right] at (n2) {$\Lang(\BASE{\coproop})$\\$\Lang(\BASE{\coproop, \intersect})$\\$\Lang(\BASE{\coproop, \difference})$};
                    \node (n3a) [right,align=left] at (n3) {$\Lang(\BASE{\transop, \coproop})$\\$\Lang(\BASE{\transop, \coproop, \difference})$};
                \end{tikzpicture}
        &
                \begin{tikzpicture}[hasse_style]
                    \node[dot] (n1) at (0, 0) {};
                    \node[dot] (n2) at (1, 0.25) {} edge[<-] (n1);
                    \node[dot] (n3) at (0, 1) {} edge[<-] (n1);
                    \node[dot] (n4) at (1, 1.25) {} edge[<-] (n2) edge [<-] (n3);
                    \node[dot] (n5) at (0, 2) {} edge[<-] (n3);
                    \node[dot] (n6) at (1, 2.25) {} edge[<-] (n4) edge [<-] (n5);

                    \node (n1a) [left,align=right] at (n1) {$\Lang()$\\$\Lang(\BASE{\intersect})$\\$\Lang(\BASE{\difference})$};
                    \node (n3a) [left,align=right] at (n3) {$\Lang(\BASE{\proop})$\\$\Lang(\BASE{\proop, \intersect})$};
                    \node (n5a) [left,align=right] at (n5) {$\Lang(\BASE{\coproop})$\\$\Lang(\BASE{\coproop, \intersect})$\\$\Lang(\BASE{\coproop, \difference})$};
                    \node (n2a) [right,align=left] at (n2) {$\Lang(\BASE{\transop})$\\$\Lang(\BASE{\transop, \intersect})$\\$\Lang(\BASE{\transop, \difference})$};
                    \node (n4a) [right,align=left] at (n4) {$\Lang(\BASE{\transop, \proop})$\\$\Lang(\BASE{\transop, \proop, \intersect})$};
                    \node (n6a) [right,align=left] at (n6) {$\Lang(\BASE{\transop, \coproop})$\\$\Lang(\BASE{\transop, \coproop, \difference})$};
                \end{tikzpicture}
        \\
    \end{tabular}}
    \caption[The full Hasse diagrams describing the relations between the expressive power of the various fragments of $\Lang(\transop, \coproop, \proop, \difference, \intersect)$.]{The full Hasse diagrams describing the relations between the expressive power of the various fragments of $\Lang(\transop, \coproop, \proop, \difference, \intersect)$. An edge $A$\tikz[baseline=-0.5ex]{\path (0,0) edge[->] (0.5,0);}$B$ indicates $A \LeqExpr B$ and $B\nLeqExpr A$. For boolean queries, we have included the cases for the non-downward operations diversity ($\diversity$) and converse ($\convop$) that follow from the homomorphism results. Notice, for path queries, adding $\diversity$ or $\convop$ to a downward fragment always adds expressive power.}\label{fig:main_results}
\end{figure}
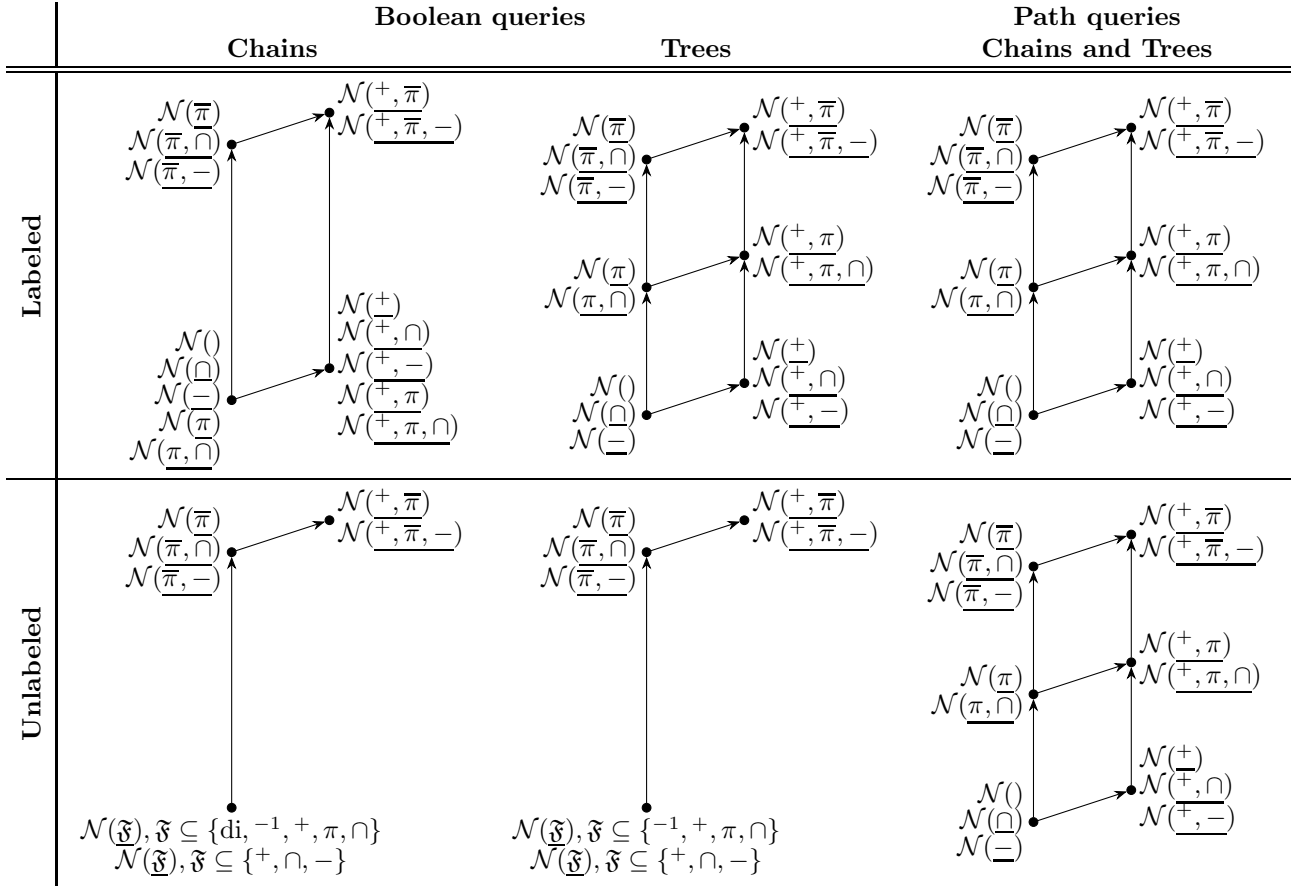

\section{Preliminaries}\label{sec:prelim}
A \emph{graph} is a triple $\Graph = (\Nodes, \Alphabet, \ELabels)$, with $\Nodes$ a finite set of nodes, $\Alphabet$ a finite set of edge labels, and $\ELabels : \Alphabet \rightarrow 2^{\Nodes \times \Nodes}$ a function mapping edge labels to edge relations. A graph is \emph{unlabeled} if $\abs{\Alphabet} = 1$. We use $\Edges$ to refer to the union of all edge relations, and, when the graph is unlabeled, to the single edge relation.

A \emph{tree} $\Tree = (\Nodes, \Alphabet, \ELabels)$ is an acyclic graph in which exactly one node, the \emph{root}, has no incoming edges, and all other nodes have exactly one incoming edge. In an edge $(m, n)$, node $m$ is the \emph{parent} of node $n$ and node $n$ is a \emph{child} of node $m$. A \emph{chain} is a tree in which all nodes have at most one child.  A \emph{path} in a graph $\Graph = (\Nodes, \Alphabet, \ELabels)$ is a sequence $n_1\Label_1n_2\dots \Label_{i-1}n_i$ with $n_1, \dots, n_i \in \Nodes$, $\Label_1, \dots, \Label_{i-1} \in \Alphabet$, and, for all $1 \leq j < i$, $(n_j, n_{j+1}) \in \Get{\Label_j}{\ELabels}$.

\begin{definition}\label{def:nav_grammar}
The \emph{navigational expressions} over graphs are defined by the grammar \begin{multline*}\expr := \emptyset \mid \identity \mid \diversity \mid \Label \text{ (for $\Label$ an edge label)} \mid \converse{\expr} \mid{}\\ \transitive{\expr} \mid \project{1}{\expr} \mid \project{2}{\expr} \mid \coproject{1}{\expr} \mid \coproject{2}{\expr} \mid \expr \compose \expr \mid \expr \union \expr \mid \expr \intersect \expr \mid \expr \difference \expr.\end{multline*}
\end{definition}
We also use the shorthand notations $\all = \identity \union \diversity$, $\transitivestar{\expr} = \identity \union \transitive{\expr}$, \[
\Edges = \bigcup_{\Label \in \Alphabet} \Label,\text{ and }\expr^k = \begin{cases}
    \identity &\text{if $k = 0$};\\
    \expr \compose \expr^{k-1} &\text{if $k > 0$}.\end{cases}\]
We generalize the usage of $\expr^k$, with $1 \leq k$, to arbitrary binary relations $\Relation$, such that $\Relation^1 = \Relation$ and, for all $1 < i$, $\Relation^i = \Relation \compose \Relation^{i-1}$.

We define the size of a navigational expression $\expr$, denoted by $\esize{\expr}$, as follows
\[\esize{\expr} = \begin{cases}
0                                          & \text{if $\expr \in \{ \emptyset, \identity, \diversity \}$};\\
0                                          & \text{if $\expr = \Label$, with $\Label$ an edge label};\\
\esize{\expr'} + 1                         & \text{if $\expr \in \{ \converse{\expr'}, \transitive{\expr}, \project{1}{\expr'}, \project{2}{\expr'},  \coproject{1}{\expr'}, \coproject{2}{\expr'}  \} $};\\
\esize{\expr_1} + \esize{\expr_2} + 1 & \text{if $\expr \in \{ \expr_1 \compose \expr_2, \expr_1 \union \expr_2, \expr_1 \intersect \expr_2,  \expr_1 \difference \expr_2 \}$}.
\end{cases}\]

The basic language we study, denoted by $\Lang()$, is the language that allows the operators $\emptyset$, $\identity$, $\Label$ (for $\Label$ an edge label), $\compose$, and $\union$. If $\Fragment \subseteq \{ \diversity,  \convop,\allowbreak \transop, \proop[1], \proop[2], \coproop[1], \coproop[2], \intersect, \difference \}$, then $\Lang(\Fragment)$ denotes the language that allows all basic operators and, additionally, the operators in $\Fragment$. We usually only consider fragments without $\proop[1]$ and $\proop[2]$ or with both included, and we simply write $\proop$. Likewise, we only consider fragments without $\coproop[1]$ and $\coproop[2]$ or with both included, and we simply write $\coproop$.

\begin{definition}
Let $\Graph = (\Nodes, \Alphabet, \ELabels)$ be a graph and let $\expr$ be a navigational expression. We write $\Apply{\expr}{\Graph}$ to denote the \emph{evaluation} of expression $\expr$ on graph $\Graph$, and the semantics of evaluation is defined as follows:
\begin{align*}
\Apply{\emptyset}{\Graph}           &= \emptyset;\\\displaybreak[0]
\Apply{\identity}{\Graph}           &= \{ (m, m) \mid m \in \Nodes \};\\\displaybreak[0]
\Apply{\diversity}{\Graph}           &= \{ (m, n) \mid m, n \in \Nodes \land m \neq n \};\\\displaybreak[0]
\Apply{\Label}{\Graph}                   &= \Get{\Label}{\ELabels};\\\displaybreak[0]
\Apply{\converse{\expr}}{\Graph}              &= \converse{\Apply{\expr}{\Graph}};\\\displaybreak[0]
\Apply{\transitive{\expr}}{\Graph}              &= \transitive{\Apply{\expr}{\Graph}};\\\displaybreak[0]
\Apply{\project{1}{\expr}}{\Graph}               &= \{ (m, m) \mid \exists n\ (m, n) \in \Apply{\expr}{\Graph} \};\\\displaybreak[0]
\Apply{\project{2}{\expr}}{\Graph}               &= \{ (m, m) \mid \exists n\ (n, m) \in \Apply{\expr}{\Graph} \};\\\displaybreak[0]
\Apply{\coproject{1}{\expr}}{\Graph}             &= \{ (m, m) \mid (m \in \Nodes) \land (\lnot\exists n\ (m, n) \in \Apply{\expr}{\Graph}) \};\\\displaybreak[0]
\Apply{\coproject{2}{\expr}}{\Graph}             &= \{ (m, m) \mid (m \in \Nodes) \land (\lnot\exists n\ (n, m) \in \Apply{\expr}{\Graph}) \};\\\displaybreak[0]
\Apply{\expr_1 \compose \expr_2}{\Graph}    &= \Apply{\expr_1}{\Graph} \compose \Apply{\expr_2}{\Graph};\\\displaybreak[0]
\Apply{\expr_1 \union \expr_2}{\Graph}      &= \Apply{\expr_1}{\Graph} \union \Apply{\expr_2}{\Graph};\\\displaybreak[0]
\Apply{\expr_1 \intersect \expr_2}{\Graph}  &= \Apply{\expr_1}{\Graph} \intersect \Apply{\expr_2}{\Graph};\\\displaybreak[0]
\Apply{\expr_1 \difference \expr_2}{\Graph} &= \Apply{\expr_1}{\Graph} \difference \Apply{\expr_2}{\Graph}.
\end{align*}
In the above, $\converse{\Relation}$, for a binary relation $\Relation \subseteq \Nodes \times \Nodes$, is defined by $\converse{\Relation} = \{ (m, n) \mid (n, m) \in \Relation \}$, $\transitive{\Relation}$, for a binary relation $\Relation \subseteq \Nodes \times \Nodes$, is defined by $\transitive{\Relation} = \bigcup_{1 \leq k} \Relation^k$, and $\Relation_1 \compose \Relation_2$, for binary relations $\Relation_1, \Relation_2 \subseteq \Nodes \times \Nodes$, is defined by $\Relation_1 \compose \Relation_2 = \{ (m, n) \mid \exists z\ ((m, z) \in \Relation_1 \land (z, n) \in \Relation_2) \}$.
\end{definition}

The following example illustrates the usage of navigational expressions:

\begin{example}
Consider the class-structure of a program described by relations \textit{subclass} and \textit{method}, as visualized by the tree in Figure~\ref{fig:example_class}. 
In this setting, the expression $\transitive{\textit{subclass\/}}$ returns the relation between classes and their descendant classes, the expression $\coproject{1}{\textit{method\/}}$ returns all classes that do not define their own methods, and the expression $\project{1}{\textit{method\/}} \difference \project{1}{\transitive{\textit{subclass\/}}\compose\textit{method\/}}$ returns all classes that define methods, while having no descendants that also define methods.
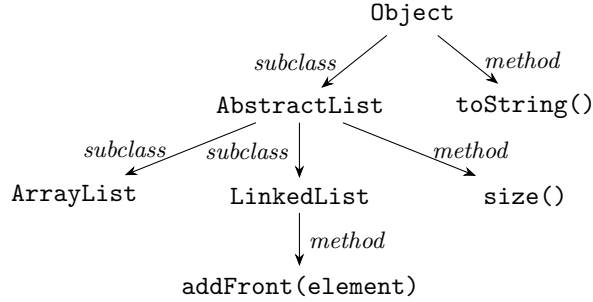
\begin{figure}[ht!]
    \centering
    \begin{tikzpicture}[tree_style,yscale=1.2]
        \node (o) at (4.5, 3) {\texttt{Object}};
            \node (ots) at (6, 2) {\texttt{toString()}} edge[<-] node[right=0.25em] {\small\textit{method}} (o);

        \node (al) at (3, 2) {\texttt{AbstractList}} edge[<-] node[left=0.25em] {\small\textit{subclass}} (o);
            \node (alsize) at (6, 1) {\texttt{size()}} edge[<-] node[right=0.5em] {\small\textit{method}} (al);

        \node (arl) at (0, 1) {\texttt{ArrayList}} edge[<-] node[left=0.5em] {\small\textit{subclass}} (al);
        \node (ll) at (3, 1) {\texttt{LinkedList}} edge[<-] node[left] {\small\textit{subclass}} (al);
            \node  (llaf) at (3, 0) {\texttt{addFront(element)}} edge[<-] node[right] {\small\textit{method}} (ll);

    \end{tikzpicture}
    \caption{The hierarchical relations within typical list-classes in a Java-like object-oriented programming language.}\label{fig:example_class}
\end{figure}
\end{example}

We say that queries $\expr_1$ and $\expr_2$ are \emph{path-equivalent} if, for all graphs $\Graph$, we have $\Apply{\expr_1}{\Graph} = \Apply{\expr_2}{\Graph}$, and \emph{boolean-equivalent} if, for all graphs $\Graph$, we have $\Apply{\expr_1}{\Graph} \neq \emptyset$ if and only if $\Apply{\expr_2}{\Graph} \neq \emptyset$. If $\Lang_1$ and $\Lang_2$ are query languages, then $\Lang_2$ \emph{path-subsumes} $\Lang_1$, denoted by $\Lang_1 \path\LeqExpr \Lang_2$, if every query in $\Lang_1$ is path-equivalent to a query in $\Lang_2$. Likewise, $\Lang_2$ \emph{boolean-subsumes} $\Lang_1$, denoted by $\Lang_1 \bool\LeqExpr \Lang_2$, if every query in $\Lang_1$ is boolean-equivalent to a query in $\Lang_2$. We write $\Lang_1 \path\LeExpr \Lang_2$ if $\Lang_1 \path\LeqExpr \Lang_2$ and $\Lang_2 \path\nLeqExpr \Lang_1$, and we write $\Lang_1 \bool\LeExpr \Lang_2$ if $\Lang_1 \bool\LeqExpr \Lang_2$ and $\Lang_2 \bool\nLeqExpr \Lang_1$.

Several operators can be expressed in terms of the other operators~\cite{graph_navjr}:
\begin{align*}
\project{1}{\expr}&\equiv \coproject{1}{\coproject{1}{\expr}} \equiv \coproject{2}{\coproject{1}{\expr}} \equiv \expr \compose \converse{\expr} \intersect \identity \equiv \expr \compose \all \intersect \identity \equiv \project{2}{\converse{\expr}};\\
\project{2}{\expr}&\equiv \coproject{1}{\coproject{2}{\expr}} \equiv \coproject{2}{\coproject{2}{\expr}} \equiv  \converse{\expr} \compose \expr \intersect \identity \equiv \all \compose \expr \intersect \identity \equiv \project{1}{\converse{\expr}};\\
\coproject{1}{\expr}&\equiv\identity \difference \project{1}{\expr} \equiv \coproject{2}{\converse{\expr}};\\
\coproject{2}{\expr}&\equiv\identity \difference \project{2}{\expr} \equiv \coproject{1}{\converse{\expr}};\\
\expr_1 \intersect \expr_2&\equiv \expr_1 \difference (\expr_1 \difference \expr_2).
\end{align*}
Let $\Fragment \subseteq \{ \diversity, \convop, \transop, \proop[1], \proop[2], \coproop[1], \coproop[2], \intersect, \difference \}$. We define $\BASE{\Fragment}$ (underlined $\Fragment$) to be the superset of $\Fragment$ obtained by adding all operators that can be expressed indirectly in $\Lang(\Fragment)$ by using the  above identities. For example, $\{ \BASE{\convop, \difference} \} = \{ \convop, \proop[1], \proop[2], \coproop[1], \coproop[2], \intersect, \difference \}$.

We also observe the following straightforward result:
\begin{lemma}\label{lem:remove_emptyset}
Let $\Fragment \subseteq \{ \diversity, \convop, \transop, \proop[1], \proop[2], \coproop[1], \coproop[2], \intersect, \difference \}$ and let $\expr$ be a navigational expression in $\Lang(\Fragment)$. If there exists a graph $\Graph$ such that $\Apply{\expr}{\Graph} \neq \emptyset$, then there exists a path-equivalent navigational expression in $\Lang(\Fragment)$ that does not utilize the operator $\emptyset$.
\end{lemma}

Hence, we may ignore the operator $\emptyset$ unless we need to express the query that returns the empty set on every input graph. We conclude these preliminaries with some established results that will be used throughout this work:
\begin{proposition}[Fletcher et al.~\cite{graph_navjr}]\label{prop:pathbool_carry}
Let $\Lang_1$ and $\Lang_2$ be query languages.
\begin{enumerate}
\item If $\Lang _1 \path\LeqExpr \Lang_2$, then $\Lang_1 \bool\LeqExpr \Lang_2$;
\item If $\Lang_1 \bool\nLeqExpr \Lang_2$, then $\Lang_1 \path\nLeqExpr \Lang_2$.
\end{enumerate}
\end{proposition}

Besides carrying over results between boolean and path queries, we can also carry over results between types of graphs.
\begin{proposition}\label{prop:class_carry}
Let $\mathord{\leq} \in \{ \bool\LeqExpr, \path\LeqExpr \}$, let $\Lang_1$ and $\Lang_2$ be query languages, and let $\mathcal{C}_1$ and $\mathcal{C}_2$ be classes of graphs such that $\mathcal{C}_1$ is a subclass of $\mathcal{C}_2$.
\begin{enumerate}
    \item If $\Lang_1 \leq \Lang_2$ on $\mathcal{C}_2$, then $\Lang_1 \leq \Lang_2$ on $\mathcal{C}_1$;
    \item If $\Lang_1 \nleq \Lang_2$ on $\mathcal{C}_1$, then $\Lang_1 \nleq \Lang_2$ on $\mathcal{C}_2$.
\end{enumerate}
\begin{figure}[htb!]
    \centering
    {\it
        \begin{tikzpicture}[xscale=3]
            \node (n1) at (0, 0) {unlabeled chain};
            \node (n2) at (0.5, 1) {labeled chain} edge[<-] (n1);
            \node (n3) at (1, 0) {unlabeled tree}  edge[<-] (n1);
            \node (n4) at (1.5, 1) {labeled tree} edge[<-] (n3) edge[<-] (n2);
            \node (n5) at (2, 0) {unlabeled graph} edge[<-] (n3);
            \node (n6) at (2.5, 1) {labeled graph} edge[<-] (n5) edge[<-] (n4);
        \end{tikzpicture}
    }
    \caption{The various classes of graphs on which relationships in the expressive power of navigational query languages is studied. The arrows between classes define is-a relationships.}\label{fig:graph_classes}
\end{figure}
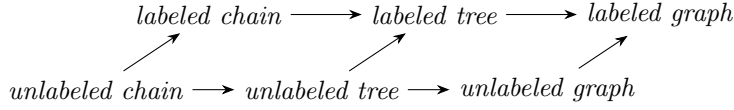
\end{proposition}

We often use Proposition~\ref{prop:pathbool_carry}, Proposition~\ref{prop:class_carry}, and the subclass relations of Figure~\ref{fig:graph_classes} implicitly to carry over results between various cases.

\begin{proposition}[Fletcher et al.~\cite{graph_icdt}]\label{prop:subleqexpr}
Let $\Fragment_1, \Fragment_2 \subseteq \{ \diversity, \convop, \transop, \proop[1], \proop[2], \coproop[1], \coproop[2], \intersect, \allowbreak \difference \}$. If $\Fragment_1 \subseteq \BASE{\Fragment_2}$, then $\Lang(\Fragment_1) \path\LeqExpr \Lang(\Fragment_2)$.
\end{proposition}

\section{Results on the Expressive Power}\label{sec:expressive_power}

In this section, we present all relevant results on the relative expressive power of the downward navigational query languages, resulting in the Hasse diagrams of relative expressiveness visualized in Figure~\ref{fig:main_results}. We divide our results into four categories, based on the techniques used to prove them. Section~\ref{ss:locality} provides all the results obtained using locality-based arguments. Section~\ref{ss:homo} provides all the results obtained using homomorphisms. These include a major collapse result for boolean queries on unlabeled trees and unlabeled chains. Section~\ref{ss:auto} introduces condition automata and uses condition automata to obtain closure results for difference and intersection. Furthermore, condition automata are used to show that projection does not add expressive power for boolean queries on labeled chains. Section~\ref{ss:fo} uses the relation between the navigational query languages and first-order logic to prove some expressiveness results involving transitive closure.

Combined, all  these results prove the following:

\begin{theorem}\label{thm:main_result}
Let $\Fragment_1, \Fragment_2 \subseteq \{ \transop, \coproop, \proop, \intersect, \difference\}$. We have $\Lang(\Fragment_1) \bool\LeqExpr \Lang(\Fragment_2)$, respectively $\Lang(\Fragment_1) \path\LeqExpr \Lang(\Fragment_2)$, on unlabeled chains, respectively labeled chains, unlabeled trees, or labeled trees if and only if there exists a directed path from $\BASE{\Fragment_1}$ to $\BASE{\Fragment_2}$ in the corresponding Hasse diagram of Figure~\ref{fig:main_results}.
\end{theorem}

\subsection{Results using direction and locality}\label{ss:locality}

Navigational expressions that do not utilize diversity $(\diversity)$ or converse $(\convop)$ will always navigate a tree from a parent to a child, which we formalize as \emph{downward}:

\begin{definition}
Let $\Tree = (\Nodes, \Alphabet, \ELabels)$ be a tree and let $m \in \Nodes$ be a node. A node $n \in \Nodes$ is an ancestor of node $m$ if there exists a directed path from node $n$ to node $m$. We say that a navigational expression is \emph{downward} if, for any tree $\Tree$, we have $(m, n) \in \Apply{\expr}{\Tree}$ implies $m$ is an ancestor of $n$.
\end{definition}

As the navigational expressions $\converse{\Edges}$ and $\diversity$ are not downward, we may immediately conclude the following:

\begin{proposition}\label{prop:chain_downward}
Let $\Fragment \subseteq \{ \transop, \proop, \coproop, \intersect, \difference \}$. On unlabeled chains we have $\Lang(\convop) \allowbreak\path\nLeqExpr \Lang(\Fragment)$ and $\Lang(\diversity) \path\nLeqExpr \Lang(\Fragment)$.
\end{proposition}

Observe that the query result of $\project{1}{\expr}$ and $\project{2}{\expr}$ will always be a subset of $\identity$. For downward languages that cannot express $\proop$ directly or via straightforward rewriting, we can easily show that each navigational expression, when evaluated on a chain $\Chain$, will either include all of $\Apply{\identity}{\Chain}$ or will have no overlap with $\Apply{\identity}{\Chain}$:

\begin{proposition}\label{prop:emptyorid}
Let $\Chain = (\Nodes, \Alphabet, \ELabels)$ be an unlabeled chain, let $\Fragment \subseteq \{ \transop, \intersect, \difference \}$, and let $\expr$ be a navigational expression in $\Lang(\Fragment)$. Then either $\Apply{\identity}{\Chain} \subseteq \Apply{\expr}{\Chain}$ or $\Apply{\identity}{\Chain} \intersect \Apply{\expr}{\Chain} = \emptyset$.
\end{proposition}
\begin{proof}
The proof is by induction on the structure of $\expr$. The base cases are $\expr$ with $\esize{\expr} = 0$ ($\expr \in \{ \emptyset, \identity, \Label \}$, with $\Label$ an edge label) and are straightforward to verify. Assume that, for every expression $\expr'$ with $\esize{\expr'} < i$ and every chain $\Chain$, we have $\Apply{\identity}{\Chain} \subseteq \Apply{\expr'}{\Chain}$ or we have $\Apply{\identity}{\Chain} \intersect \Apply{\expr'}{\Chain} = \emptyset$.  Let $\expr$ be an expression with $\esize{\expr} = i$. We distinguish the following cases:
\begin{enumerate}
\item $\expr = \transitive{\expr''}$. We have $\esize{\expr''} = i - 1$, hence, we apply the induction hypothesis to conclude that $\Apply{\identity}{\Chain} \subseteq \Apply{\expr''}{\Chain}$ or $\Apply{\identity}{\Chain} \intersect \Apply{\expr''}{\Chain} = \emptyset$. As $\expr$ is downward, we can use the semantics of $\transop$ in a straightforward manner to conclude that $\Apply{\identity}{\Chain} \subseteq \Apply{\expr}{\Chain}$ if and only if $\Apply{\identity}{\Chain} \subseteq \Apply{\expr''}{\Chain}$ and $\Apply{\identity}{\Chain} \intersect \Apply{\expr}{\Chain} = \emptyset$ otherwise.

\item $\expr = \expr_1 \compose \expr_2$. We have $\esize{\expr_1} < i$ and $\esize{\expr_2} < i$, and we apply the induction hypothesis to conclude that $\Apply{\identity}{\Chain} \subseteq \Apply{\expr_1}{\Chain}$ or $\Apply{\identity}{\Chain} \intersect \Apply{\expr_1}{\Chain} = \emptyset$, and that $\Apply{\identity}{\Chain} \subseteq \Apply{\expr_2}{\Chain}$ or $\Apply{\identity}{\Chain} \intersect \Apply{\expr_2}{\Chain} = \emptyset$. Using the semantics of $\compose$, we conclude that $\Apply{\identity}{\Chain} \subseteq \Apply{\expr}{\Chain}$ if and only if $\Apply{\identity}{\Chain} \subseteq \Apply{\expr_1}{\Chain}$ and $\Apply{\identity}{\Chain} \subseteq \Apply{\expr_2}{\Chain}$, and $\Apply{\identity}{\Chain} \intersect \Apply{\expr}{\Chain} = \emptyset$ otherwise.

\item $\expr = \expr_1 \union \expr_2$. We have $\esize{\expr_1} < i$ and $\esize{\expr_2} < i$, and we apply the induction hypothesis to conclude that $\Apply{\identity}{\Chain} \subseteq \Apply{\expr_1}{\Chain}$ or $\Apply{\identity}{\Chain} \intersect \Apply{\expr_1}{\Chain} = \emptyset$, and that $\Apply{\identity}{\Chain} \subseteq \Apply{\expr_2}{\Chain}$ or $\Apply{\identity}{\Chain} \intersect \Apply{\expr_2}{\Chain} = \emptyset$. Using the semantics of $\union$, we conclude that $\Apply{\identity}{\Chain} \subseteq \Apply{\expr}{\Chain}$ if and only if $\Apply{\identity}{\Chain} \subseteq \Apply{\expr_1}{\Chain}$ or $\Apply{\identity}{\Chain} \subseteq \Apply{\expr_2}{\Chain}$, and $\Apply{\identity}{\Chain} \intersect \Apply{\expr}{\Chain} = \emptyset$ otherwise.

\item \label{prop:emptyorid:diff} $\expr = \expr_1 \difference \expr_2$. We have $\esize{\expr_1} < i$ and $\esize{\expr_2} < i$, and we apply the induction hypothesis to conclude that $\Apply{\identity}{\Chain} \subseteq \Apply{\expr_1}{\Chain}$ or $\Apply{\identity}{\Chain} \intersect \Apply{\expr_1}{\Chain} = \emptyset$, and that $\Apply{\identity}{\Chain} \subseteq \Apply{\expr_2}{\Chain}$ or $\Apply{\identity}{\Chain} \intersect \Apply{\expr_2}{\Chain} = \emptyset$. Using the semantics of $\difference$, we conclude that $\Apply{\identity}{\Chain} \subseteq \Apply{\expr}{\Chain}$ if and only if $\Apply{\identity}{\Chain} \subseteq \Apply{\expr_1}{\Chain}$ and $\Apply{\identity}{\Chain} \intersect \Apply{\expr_2}{\Chain} = \emptyset$, and $\Apply{\identity}{\Chain} \intersect \Apply{\expr}{\Chain} = \emptyset$ otherwise.

\item $\expr = \expr_1 \intersect \expr_2$. Follows from Case~\ref{prop:emptyorid:diff}, as $\expr_1 \intersect \expr_2 = \expr_1 \difference (\expr_1 \difference \expr_2)$. \qedhere
\end{enumerate}
\end{proof}

Proposition~\ref{prop:emptyorid} yields the following:

\begin{corollary}\label{cor:prim_proop_downchain}
Let $\Fragment \subseteq \{ \transop, \intersect, \difference \}$. On unlabeled chains we have $\Lang(\Fragment \union \{\proop \}) \path\nLeqExpr \Lang(\Fragment)$ and  $\Lang(\Fragment \union \{\coproop \}) \path\nLeqExpr \Lang(\Fragment)$.
\end{corollary}
\begin{proof}
Let $\Chain = (\Nodes, \Alphabet, \ELabels)$ be a chain with $\abs{\Nodes} \leq 2$. We have $\emptyset \subsetneq \Apply{\project{1}{\Edges}}{\Chain} = \Apply{\project{1}{\Edges}}{\Chain} \intersect \Apply{\identity}{\Chain} \subsetneq \Apply{\identity}{\Chain}$. By Proposition~\ref{prop:emptyorid}, no expression in $\Lang(\Fragment)$ can be path-equivalent to $\project{1}{\Edges}$.
\end{proof}

\subsection{Results using homomorphisms}\label{ss:homo}

Many of the language fragments we consider are closed under homomorphisms. Using these closure results, we can prove several expressivity results for boolean queries.

\begin{definition}
Let $\Graph_1 = (\Nodes_1, \Alphabet, \ELabels_1)$ and $\Graph_2 = (\Nodes_2, \Alphabet, \ELabels_2)$ be graphs. We say that a mapping $h : \Nodes_1 \rightarrow \Nodes_2$ is a \emph{homomorphism} from $\Graph_1$ to $\Graph_2$ if, for every pair of nodes $m, n \in \Nodes_1$ and every edge label $\Label \in \Alphabet$, we have that $(m, n) \in \Get{\Label}{\ELabels_1}$ implies $(h(m), h(n)) \in \Get{\Label}{\ELabels_2}$. A homomorphism $h : \Nodes_1 \rightarrow \Nodes_2$ is called \emph{injective} if, for all $m, n \in \Nodes_1$, $n \neq m$ implies $h(n) \neq h(m)$.
\end{definition}

\begin{definition}
Let $F$ be a class of functions of $\Nodes_1 \rightarrow \Nodes_2$ and let $\Fragment \subseteq \{ \diversity, \convop, \transop, \proop, \coproop, \intersect, \difference \}$. We say that  $\Lang(\Fragment)$ is \emph{closed} under $F$ if, for every navigational expression $\expr$ in $\Lang(\Fragment)$ and every $f \in F$, we have, $(m, n) \in \Apply{\expr}{\Graph_1}$ implies $(f(m), f(n)) \in \Apply{\expr}{\Graph_2}$.
\end{definition}

Via straightforward induction proofs, we show that navigational query languages without the operators diversity, coprojection, and difference are closed under homomorphisms. Observe that diversity is a form of inequality, hence, when we do allow diversity, we  show that these languages are closed under injective homomorphisms:
\begin{lemma}\label{lem:closed_homo}
Let $\Fragment \subseteq \{ \convop, \transop, \proop, \intersect \}$. The language $\Lang(\Fragment)$ is closed under homomorphisms and the language $\Lang(\Fragment  \union \{ \diversity \})$ is closed under injective homomorphisms.
\end{lemma}

We shall now show that there always exist homomorphisms from unlabeled trees to long unlabeled chains, and injective homomorphisms from unlabeled chains to deep unlabeled trees and longer chains. We shall use these (injective) homomorphisms to show that the languages closed under (injective) homomorphisms are unable to recognize complex structures in trees.

Let $\Tree = (\Nodes, \Alphabet, \ELabels)$ be a tree and let $m,n \in \Nodes$ be nodes such that $n$ is an ancestor of $m$. We define the distance from $n$ to $m$, denoted by $\distance(n, m)$, as the length of the path (in edges) from $n$ to$m$. The \emph{depth} of tree $\Tree$ is defined by $\depth(\Tree) = \max_{m \in \Nodes} \distance(r, m)$. If $\Tree$ is a chain, then $\abs{\Nodes} = \depth(\Tree) + 1$.

\begin{lemma}\label{lem:simple_homo}

\begin{enumerate}
\item Let $\Chain_1 = (\Nodes_1, \Alphabet, \ELabels_1)$ and $\Chain_2 = (\Nodes_2, \Alphabet, \ELabels_2)$ be unlabeled chains. If $\abs{\Nodes_1} \leq \abs{\Nodes_2}$, then there exists an injective homomorphism from $\Chain_1$ to $\Chain_2$.
\item Let $\Tree = (\Nodes_\Tree, \Alphabet, \ELabels_\Tree)$ be an unlabeled tree and let $\Chain = (\Nodes_\Chain, \Alphabet, \ELabels_\Chain)$ be an unlabeled chain.
\begin{enumerate*}
\item If $\abs{\Nodes_\Chain} \geq \depth(\Tree) + 1$, then there exists a homomorphism from $\Tree$ to $\Chain$.
\item If $\abs{\Nodes_\Chain} \leq \depth(\Tree) + 1$, then there exists an injective homomorphism from $\Chain$ to $\Tree$.
\end{enumerate*}
\end{enumerate}
\end{lemma}
\begin{proof}[Proof (sketch).]
\begin{enumerate}
\item Map nodes $n \in \Nodes_1$ to $n' \in \Nodes_2$ such that $n$ and $n'$ have equal distance from the root of $\Chain_1$ and $\Chain_2$, respectively.
\item
\begin{enumerate}
\item Map nodes $n \in \Nodes_\Tree$ to $n' \in \Nodes_\Chain$  such that $n$ and $n'$ have equal distance from the root of $\Tree$ and $\Chain$, respectively.
\item Choose a path in the tree $\Tree$ from the root node $r$ to leaf node $m$ with $\distance(r, m) = \depth(\Tree)$. Map nodes $n \in \Nodes_\Chain$ to nodes $n'$ on this path such that $n$ and $n'$ have equal distance from the root of $\Chain$ and $\Tree$, respectively.\qedhere
\end{enumerate}
\end{enumerate}
\end{proof}

We use Lemma~\ref{lem:simple_homo} to show that, on unlabeled trees, all navigational expressions that are closed under homomorphisms are equivalent to $\emptyset$ or to the boolean query ``\emph{the height of the tree is at least $k$}'', for a fixed value of $k$.

\begin{proposition}\label{pro:depthatk}
Let $\Fragment \subseteq \{ \convop, \transop, \proop, \intersect \}$ and let $\expr$ be a navigational expression in $\Lang(\Fragment)$ ($\Lang(\Fragment \union \{ \diversity \})$). If, on unlabeled trees (unlabeled chains), $\expr$ is not boolean-equivalent to $\emptyset$, then there exists a $k$, $0 \leq k$, such that $\expr$ is boolean-equivalent to $\Edges^k$.
\end{proposition}
\begin{proof}
Let $\Tree = (\Nodes, \Alphabet, \ELabels)$ be a tree and let $\Chain' = (\Nodes', \Alphabet, \ELabels')$ be a chain such that $\depth(\Tree) + 1 = \abs{\Nodes'}$. By Lemma~\ref{lem:simple_homo}, there exists a homomorphism $h_1 : \Nodes \rightarrow \Nodes'$ from $\Tree$ to $\Chain'$ and a homomorphism $h_2 : \Nodes' \rightarrow \Nodes$ from $\Chain'$ to $\Tree$. Hence, by Lemma~\ref{lem:closed_homo}, we have, for every navigational expression $\expr'$ in $\Lang(\Fragment)$, $\Apply{\expr}{\Tree} \neq \emptyset$ if and only if $\Apply{\expr}{\Chain'} \neq \emptyset$. As a consequence, we can conclude that no navigational expression in $\Lang(\Fragment)$ can distinguish between trees and chains of equal depth.

Let $\Chain = (\Nodes, \Alphabet, \ELabels)$ and $\Chain' = (\Nodes', \Alphabet, \ELabels')$ be chains such that $\abs{\Nodes} < \abs{\Nodes'}$. By Lemma~\ref{lem:simple_homo}, there exists a homomorphism $h : \Nodes \rightarrow \Nodes'$ from $\Chain$ to $\Chain'$. Hence, by Lemma~\ref{lem:closed_homo}, we have, for every navigational expression $\expr'$ in $\Lang(\Fragment)$, $\Apply{\expr}{\Chain} \neq \emptyset$ implies  $\Apply{\expr}{\Chain'} \neq \emptyset$. As a consequence, we can conclude that, if a navigational expression holds on a chain, it also holds on every chain of greater depth.

Let $\expr$ be a navigational expression in $\Lang(\Fragment)$. We choose $\Chain = (\Nodes, \Alphabet, \ELabels)$ to be the chain with minimum depth such that $\Apply{\expr}{\Chain} \neq \emptyset$. If no such chain exists, then, by the two properties shown above, $\expr$ is boolean-equivalent to $\emptyset$. Since $\Chain$ is also the chain with minimum depth such that $\Apply{\Edges^k}{\Chain} \neq \emptyset$, with $k = \abs{\Nodes} - 1$, we may conclude, by the two properties shown above, that $\Edges^k$ and $\expr$ are boolean-equivalent.

The case for $\Lang(\Fragment \union \{ \diversity \})$ on unlabeled chains is analogous, taking into account that on chains we can always construct injective homomorphisms.
\end{proof}

As a consequence, we have the following collapses.

\begin{corollary}\label{cor:collapse_lang}
Let $\Fragment \subseteq \{ \convop, \transop, \proop, \intersect \}$. On unlabeled trees, we have $\Lang(\Fragment) \bool\LeqExpr \Lang()$, and, on unlabeled chains, we have $\Lang(\Fragment \union \{ \diversity \}) \bool\LeqExpr \Lang()$.
\end{corollary}
\begin{proof}
By Proposition~\ref{pro:depthatk}, we must only be able to express $\emptyset$ and $\Edges^k$, for $0 \leq k$, which are both already expressible in $\Lang()$.
\end{proof}

Using $\coproop$, we can easily distinguish chains from longer chains. Hence, we conclude the following:

\begin{proposition}\label{prop:unl_coproop}
Let $\Fragment \subseteq \{ \diversity, \convop, \transop, \proop, \intersect \}$. On unlabeled chains, we have $\Lang(\coproop) \bool\nLeqExpr \Lang(\Fragment)$.
\end{proposition}
\begin{proof}
Let $\expr$ be the expression $\coproject{2}{\Edges} \compose \Edges \compose \coproject{1}{\Edges}$. If a chain $\Chain$ has depth $2$, then we have $\Apply{\expr}{\Chain} \neq \emptyset$. For all chains $\Chain'$ with a depth other than $2$, we have $\Apply{\expr}{\Chain'} = \emptyset$. Hence, Proposition~\ref{pro:depthatk} shows that no navigational expression in $\Lang(\Fragment)$ is boolean-equivalent to $\expr$. 
\end{proof}

Besides the above results, we use homomorphisms to show that $\Lang()$ and $\Lang(\transop)$ cannot properly distinguish between labeled trees and labeled chains. This result is then used to show that, on labeled trees, $\Lang()$ and $\Lang(\transop)$ cannot express all boolean queries expressible by $\Lang(\proop)$:

\begin{lemma}\label{lem:nonempty_chain_l}
Let $\Fragment \subseteq \{ \transop \}$. Let $\expr$ be a navigational expression in $\Lang(\Fragment)$. If, on labeled trees, $\expr$ is not boolean-equivalent to $\emptyset$, then there exists a labeled chain $\Chain$ such that $\Apply{\expr}{\Chain} \neq \emptyset$. 
\end{lemma}
\begin{proof}
We prove a slightly stronger claim, namely that there exists a labeled chain $\Chain$ with root $r$ and leaf $l$ such that $(r, l) \in \Apply{\expr}{\Chain}$. The proof is by induction on the structure of $\expr$. We assume that $\expr$ is not boolean-equivalent to $\emptyset$, and, by Lemma~\ref{lem:remove_emptyset}, is $\emptyset$-free. The base cases are $\expr$ with $\esize{\expr} = 0$ ($\expr \in \{ \identity, \Label \}$, with $\Label$ an edge label). For $\expr = \identity$, we choose $\Chain$ to be a single node and for $\expr = \Label$ a single edge labeled $\Label$.

Now assume that, for every expression $\expr'$ with $\esize{\expr'} < i$, there exists a labeled chain $\Chain'$ with root $r'$ and leaf $l'$ such that $(r', l') \in \Apply{\expr'}{\Chain'}$. Let $\expr$ be an expression with $\esize{\expr} = i$. We distinguish the following cases:
\begin{enumerate}
\item $\expr = \transitive{\expr''}$. By the semantics of $\transop$, we have $\Apply{\expr''}{\Chain} \subseteq \Apply{\transitive{\expr''}}{\Chain}$. As $\esize{\expr''} = i -1$, we can use the induction hypothesis to conclude that there exists a labeled chain $\Chain''$ with root $r''$ and leaf $l''$ such that $(r'',l'') \in \Apply{\expr''}{\Chain''} \subseteq \Apply{\transitive{\expr''}}{\Chain''}$.

\item $\expr = \expr_1 \compose \expr_2$. We have $\esize{\expr_1} < i$ and $\esize{\expr_2} < i$, and we apply the induction hypothesis to conclude that there exists labeled chains $\Chain_1 = (\Nodes_1, \Alphabet, \ELabels_1)$, $\Chain_2 = (\Nodes_2, \Alphabet, \ELabels_2)$ with roots $r_1$, $r_2$ and leafs $l_1$, $l_2$ such that $(r_1, l_1) \in \Apply{\expr_1}{\Chain_1}$ and  $(r_2, l_2) \in \Apply{\expr_2}{\Chain_2}$. Now consider the chain $\Chain = (\Nodes, \Alphabet, \ELabels)$ obtained by concatenating $\Chain_1$ and $\Chain_2$ (by merging $l_1$ and $r_2$ to a single node). Let $h_1 :\Nodes_1 \rightarrow \Nodes$ and $h_2 :\Nodes_2 \rightarrow \Nodes$ be the functions mapping nodes from $\Chain_1$ and $\Chain_2$ to the corresponding node in $\Chain$, hence, with $h_1(r_1)$ being the root of $\Chain$, $h_1(l_1) = h_2(r_2)$, and $h_2(l_2)$ being the leaf of $\Chain$. The functions $h_1$ and $h_2$ are homomorphisms from $\Chain_1$ to $\Chain$ and from $\Chain_2$ to $\Chain$, respectively. Hence, by Lemma~\ref{lem:closed_homo}, we have $(h_1(r_1), h_1(l_1)) \in \Apply{\expr_1}{\Chain}$ and $(h_2(r_2), h_2(l_2)) \in \Apply{\expr_2}{\Chain}$. By the semantics of $\compose$, we conclude $(h_1(r_1), h_2(l_2)) \in \Apply{\expr}{\Chain}$.

\item $\expr = \expr_1 \union \expr_2$. By the semantics of $\union$, we have $\Apply{\expr}{\Tree} = \Apply{\expr_1}{\Tree_1} \union \Apply{\expr_2}{\Tree_1}$. We have $\esize{\expr_1} < i$ and $\esize{\expr_2} < i$, and we apply the induction hypothesis to conclude that there exists labeled chains $\Chain_1$, $\Chain_2$ with roots $r_1$, $r_2$ and leaf $l_1$, $l_2$ such that $(r_1, l_1) \in \Apply{\expr_1}{\Chain_1}$ and  $(r_2, l_2) \in \Apply{\expr_2}{\Chain_2}$. Hence, using the semantics of $\union$, both $\Chain_1$ and $\Chain_2$ meet the required conditions. \qedhere
\end{enumerate}
\end{proof}

\begin{proposition}\label{prop:proop_branching}
Let $\Fragment \subseteq \{ \transop \}$. On labeled trees we have $\Lang(\proop) \bool\nLeqExpr \Lang(\Fragment)$ and $\Lang(\convop) \bool\nLeqExpr \Lang(\Fragment)$.
\end{proposition}
\begin{proof}
Let $\Label_1$ and $\Label_2$ be two edge labels and let $\expr = \project{1}{\Label_1}\compose\project{1}{\Label_2}$ be a navigational expression in $\Lang(\proop)$. On labeled trees this expression evaluates to non-empty if and only if a node has two distinct outgoing edges labeled with $\Label_1$ and $\Label_2$, respectively. Hence, for all chains $\Chain = (\Nodes, \Alphabet, \ELabels)$, we have $\Apply{\expr}{\Chain} = \emptyset$. As $\expr$ never evaluates to non-empty on chains, we use Lemma~\ref{lem:nonempty_chain_l} to conclude that no navigational expression in $\Lang(\Fragment)$ is boolean-equivalent to $\expr$. For showing $\Lang(\convop) \bool\nLeqExpr \Lang(\Fragment)$, we use the above reasoning on the expression $\converse{\Label_1} \compose \Label_2$.
\end{proof}

\subsection{Automaton-based Results}\label{ss:auto}

Observe that $\Lang(\transop)$ and the regular path queries~\cite{rpq} are equivalent: queries in these query languages select pairs of nodes $m, n$ such that there exists a directed path from $m$ to $n$ whose labeling satisfies some regular expression. In the case of trees, this directed path is unique, which yields a strong relation between $\Lang(\transop)$ and the closure results under intersection and difference for regular languages~\cite{automata}. As a consequence, we can show, in a relative straightforward way, that $\Lang(\transop, \intersect, \difference) \path\LeqExpr \Lang(\transop)$.

\begin{example}\label{exam:intro_loop_int}
We can, for example, rewrite the navigational expressions $\transitive{\Label^3} \intersect \transitive{\Label^7}$ and $\transitive{\Label^3} \difference \transitive{\Label^7}$ to path-equivalent navigational expressions use neither intersection ($\intersect$) nor difference ($\difference$):
\begin{align*}
\transitive{\Label^3} \intersect \transitive{\Label^7} &\equiv \transitive{\Label^{21}};\\
\transitive{\Label^3} \difference \transitive{\Label^7} &\equiv \left(\Label^3 \union \Label^6 \union \Label^9 \union \Label^{12} \union \Label^{15} \union \Label^{18}\right) \compose \transitivestar{\Label^{21}}.
\end{align*}
Notice that this rewriting does not work on arbitrary graphs. Indeed, on the graph $\Graph$ in Figure~\ref{fig:3int7graph}, we have $\Apply{\transitive{\Label^3} \intersect \transitive{\Label^7}}{\Graph} \neq \emptyset$, whereas $\Apply{\transitive{\Label^{21}}}{\Graph} = \emptyset$.
\begin{figure}[htb!]
    \centering
    \begin{tikzpicture}[tree_style]
        \node[dot] (src) at (1, 1) {};
        \node[dot] (n1) at (3, 1) {} edge[<-] (src);
        \node[dot] (n2) at (5, 1) {} edge[<-] (n1);
        \node[dot] (m1) at (1.5, 2) {} edge[<-] (src);
        \node[dot] (m2) at (2.5, 2) {} edge[<-] (m1);
        \node[dot] (m3) at (3.5, 2) {} edge[<-] (m2);
        \node[dot] (m4) at (4.5, 2) {} edge[<-] (m3);
        \node[dot] (m5) at (5.5, 2) {} edge[<-] (m4);
        \node[dot] (m6) at (6.5, 2) {} edge[<-] (m5);
        \node[dot] (tgt) at (7, 1) {}  edge[<-] (n2) edge[<-] (m6);
        
        \node[left] at (src) {$\src\vphantom{\tgt}$};
        \node[right] at (tgt) {$\tgt\vphantom{\src}$};
    \end{tikzpicture}
    \caption{An acyclic directed graph, which is not a chain, a tree, or a forest. Observe that $\Label^3 \intersect \Label^7$ will return the node pair $(\src, \tgt)$.}\label{fig:3int7graph}
\end{figure}
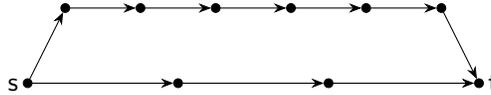
\end{example}

\subsubsection{Adapting closure results under intersection and difference}
For regular expressions, the closure results under intersection and difference are usually proven by first proving that regular expressions have the same expressive power as finite automata, and then proving that finite automata are closed under intersection and difference.  We extend these automata-based techniques to the languages $\Lang(\Fragment)$ with $\Fragment \subseteq \{ \transop, \proop, \coproop \}$ by introducing conditions on automaton states.

\begin{definition}
A navigational expression $\expr$ is a \emph{condition} if, for every graph $\Graph$, we have $\Apply{\expr}{\Graph} \subseteq \Apply{\identity}{\Graph}$.
\end{definition}

The conditions we consider in the following are expressions of the form $\emptyset$, $\identity$, $\project{1}{\expr}$, $\project{2}{\expr}$, $\coproject{1}{\expr}$, or $\coproject{2}{\expr}$.

We use these extended automata to prove that, on trees, $\Lang(\Fragment)$ is closed under intersection and $\Lang(\BASE{\Fragment})$ is closed under difference.

\begin{definition}
A \emph{condition automaton} is a $7$-tuple $\Automaton = (\States, \Alphabet, \Conditions, \Initials, \Finals, \Transitions, \StateConditions)$, where $\States$ is a set of states, $\Alphabet$ a set of transition labels, $\Conditions$  a set of condition expressions, $\Initials \subseteq \States$  a set of initial states, $\Finals \subseteq \States$ a set of final states, $\Transitions \subseteq \States \times (\Alphabet \cup \{ \identity \}) \times \States$ the transition relation, and $\StateConditions \subseteq \States \times \Conditions$ the state-condition relation. For a state $q \in \States$, we denote $\StateConditions(q) = \{ \cond \mid (q, \cond) \in \StateConditions \}$.

Let $\Fragment \subseteq \{ \transop, \proop, \coproop \}$. We say that $\Automaton$ is \emph{$\Fragment$-free} if every condition in $\Conditions$ is a navigational expression in $\Lang(\{ \transop, \proop, \coproop \} \difference \Fragment)$, we say that $\Automaton$ is \emph{acyclic} if the transition relation $\Transitions$ of $\Automaton$ is acyclic (viewed as a labeled graph over $\States \times \States$), and we say that $\Automaton$ is \emph{$\identity$-transition free} if $\Transitions \subseteq \States \times \Alphabet \times \States$.
\end{definition}

\begin{example}\label{exam:cta}
Consider the condition automaton  $\Automaton = (\States, \Alphabet, \Conditions, \Initials, \Finals,\allowbreak \Transitions, \StateConditions)$ with \begin{align*}
        \States         &= \{ q_1, q_2, q_3, q_4 \};\\
        \Alphabet         &= \{ \Label_1, \Label_2, \Label_3 \};\\
        \Conditions     &= \{ \identity, \project{2}{{\Label_1}^2}, \project{1}{{\Label_2}^3} \};\\
        \Initials       &= \{ q_1, q_4 \};\\
        \Finals         &= \{ q_3, q_4 \};\\
        \Transitions    &= \{ (q_1, \Label_1, q_2),  (q_1, \Label_3, q_4), (q_2, \Label_1, q_2), (q_2, \Label_2, q_3) \};\text{ and}\\
        \StateConditions&= \{ (q_1, \identity), (q_2, \project{1}{{\Label_1}^2}), (q_2, \project{2}{{\Label_2}^3}) \}.\end{align*}
This automaton is visualized in Figure~\ref{fig:example_automaton}. Using this visualization, it is easy to verify that the condition automaton is not acyclic (due to the $\Label_1$ labeled self-loop), is $\{ \coproop, \transop \}$-free, and is $\identity$-transition free.
\begin{figure}[htb!]
    \centering
    \begin{tikzpicture}[automata_style]
        \node        (s1)  at (-1, 0) {};
        \node[state] (n1) at (0, 0) {$q_1$};
        \node[state] (n2) at (2, 0) {$q_2$};
        \node[state,accepting] (n3) at (4, 0) {$q_3$};  
        \node        (s4)  at (-1, 2) {};
        \node[state,accepting] (n4) at (0, 2) {$q_4$};
        
        \node[conditions,below of=n1] {$\{ \identity \}$};
        \node[conditions,below of=n2] {$\{ \project{2}{{\Label_1}^2}, \project{1}{{\Label_2}^3} \}$};
        \node[conditions,below of=n3] {$\{ \}$};
        \node[conditions,above of=n4] {$\{ \}$};

        \begin{scope}[on background layer]
            \path[->] (s1) edge (n1)
                      (s4) edge (n4)
                      (n1) edge node[label] {$\Label_1$} (n2)
                      (n1) edge node[label,left] {$\Label_3$} (n4)
                      (n2) edge[loop] node[label] {$\Label_1$} (n2)
                      (n2) edge node[label] {$\Label_2$} (n3);
        \end{scope}
    \end{tikzpicture}
    \caption{An example of a condition automaton.}\label{fig:example_automaton}
\end{figure}
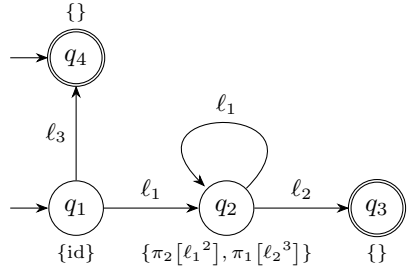
\end{example}

Observe that  condition automata are strongly related to finite automata, the main difference being that states in the automata have a set of conditions. In the evaluation of condition automata on trees, this set of conditions determines in which tree nodes a state can hold, which we define next.

\begin{definition}
Let $\Graph = (\Nodes, \Alphabet, \ELabels)$ be a graph and let $\Automaton = (\States, \Alphabet, \Conditions, \Initials, \Finals, \Transitions, \StateConditions)$ be a condition automaton. If $q\in \States$, then $\cexpr{q}$ denotes the expression $\cexpr{q} = \cond_1 \compose \dots \compose \cond_k$, with $\StateConditions(q) = \{ \cond_1, \dots, \cond_k \}$, unless $\StateConditions(q) = \emptyset$, in which case $\cexpr{q} = \identity$.\footnote{Since each term in $\StateConditions(q)$ is a condition, the compositions used in $\cexpr{q}$ are commutative, as such the ordering of the terms in $\StateConditions(q)$ is not relevant. The expression $\cexpr{q}$ is path-equivalent to the expression $\smash{\bigcap_{\cond \in \StateConditions(q) } \cond}$ and is also used to express the intersection of a set of conditions (without actually using intersection).} We say that a node $n \in \Nodes$ \emph{satisfies} state $q \in \States$ if $(n,n) \in \Apply{\cexpr{q}}{\Graph}$.

A \emph{run} of $\Automaton$ on $\Graph$ is a sequence $$(q_0, n_0) \Label_0 (q_1, n_1) \Label_1 \dots (q_{i-1}, n_{i-1}) \Label_{i-1} (q_i, n_i),$$ where  $q_0, \dots, q_i \in \States$, $n_0, \dots, n_i \in \Nodes$, $\Label_0, \dots, \Label_{i-1} \in \Alphabet \cup \{ \identity \}$, and the following conditions hold: \begin{enumerate}
\item for all $0 \leq j \leq i$, $n_j$ satisfies $q_j$;
\item for all $0 \leq j < i$, $(q_j, \Label_j, q_{j+1}) \in \Transitions$; and
\item for all $0 \leq j < i$, $(n_j, n_{j+1}) \in \Apply{\Label_j}{\Graph}$.
\end{enumerate}

We say that $\Automaton$ \emph{accepts} node pair $(m, n) \in \Nodes \times \Nodes$ if there exists a run $(q_0, m)\Label_0\dots(q_i, n)$ of $\Automaton$ on $\Graph$ with $q_0 \in \Initials$ and $q_i \in \Finals$. We define the \emph{evaluation} of $\Automaton$ on $\Graph$, denoted by $\Apply{\Automaton}{\Graph}$, as $\Apply{\Automaton}{\Graph} = \{ (m, n) \mid \text{$\Automaton$ accepts $(m, n)$} \}$. Using \emph{path query semantics}, $\Automaton$ on $\Graph$ evaluates to $\Apply{\Automaton}{\Graph}$, and using \emph{boolean query semantics}, $\Automaton$ on $\Graph$ evaluates to the truth value of $\Apply{\Automaton}{\Graph} \neq \emptyset$.
\end{definition}

\begin{example}
Consider the condition automaton of Example~\ref{exam:cta}, shown in Figure~\ref{fig:example_automaton}, and the graph shown in Figure~\ref{fig:example_auto_graph}. For this combination of a condition automaton and a graph, we can construct several accepting runs. Examples are the run $(q_1, r) \Label_3 (q_4, m)$, which semantically implies \[ (r, m) \in \Apply{\cexpr{q_1} \compose \Label_3 \compose \cexpr{q_4}}{\Graph} = \Apply{\identity \compose \Label_3 \compose \identity}{\Graph}, \] and $(q_1, n_1) \Label_1 (q_2, n_2) \Label_1 (q_2, n_3) \Label_2 (q_3, n_4)$, which semantically implies \begin{multline*}(n_1, n_2) \in \Apply{\cexpr{q_1} \compose \Label_1 \compose \cexpr{q_2} \compose \Label_1 \compose \cexpr{q_2} \compose \Label_2 \compose \cexpr{q_3}}{\Graph} ={}\\ \Apply{\identity \compose \Label_1 \compose \project{2}{{\Label_1}^2} \compose \project{1}{{\Label_2}^3}\compose \Label_1 \compose \project{2}{{\Label_1}^2} \compose \project{1}{{\Label_2}^3} \compose \Label_2 \compose \identity}{\Graph}.\end{multline*}

\begin{figure}[ht!]
    \centering
    \begin{tikzpicture}[tree_style,scale=0.75,xscale=0.5]
        \node[dot] (root) at (0  , 5)  {};
        \node[dot] (n1)   at (-1 , 4)  {} edge[<-] node[left,label] {$\Label_1$} (root);
        \node[dot] (n2)   at (-2 , 3)  {} edge[<-] node[left,label] {$\Label_1$} (n1);
        \node[dot] (n3)   at (-3 , 2)  {} edge[<-] node[left,label] {$\Label_1$} (n2);
        \node[dot] (n4)   at (-4 , 1)  {} edge[<-] node[left,label] {$\Label_2$} (n3);
        \node[dot] (m21)  at (-1 , 2)  {} edge[<-] node[right,label] {$\Label_2$} (n2);
        \node[dot] (m22)  at (0  , 1)  {} edge[<-] node[right,label] {$\Label_2$} (m21);
        \node[dot] (m23)  at (1  , 0)  {} edge[<-] node[right,label] {$\Label_2$} (m22);
        \node[dot] (m41)  at (-3 , 0)  {} edge[<-] node[right,label] {$\Label_2$} (n4);
        \node[dot] (m42)  at (-2 ,-1)  {} edge[<-] node[right,label] {$\Label_2$} (m41);
        \node[dot] (rr)   at (1, 4)    {} edge[<-] node[right,label] {$\Label_3$} (root);

        \node[right] at (root) {$r$};
        \node[right] at (rr) {$m$};
        \node[right] at (n1) {$n_1$};
        \node[right] at (n2) {$n_2$};
        \node[right] at (n3) {$n_3$};
        \node[right] at (n4) {$n_4$};
    \end{tikzpicture}
    \caption{A labeled tree in which six distinct nodes are named.}\label{fig:example_auto_graph}
\end{figure}
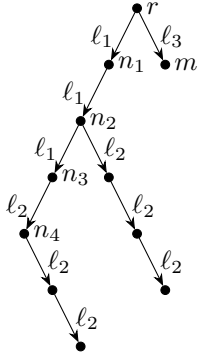
\end{example}

Our first goal is to show the path-equivalence of $\Lang(\Fragment)$, $\Fragment \subseteq \{ \transop, \proop, \coproop \}$, with a restricted class of condition automata, as summarized in Table~\ref{tbl:expr_cta}.

\begin{example}
Consider the condition automaton of Example~\ref{exam:cta}, shown in Figure~\ref{fig:example_automaton}. By carefully examining the automaton, one can conclude that it is path-equivalent to the navigational expression $$\Label_1 \compose \project{2}{{\Label_1}^2} \compose \project{1}{{\Label_2}^3} \compose \transitivestar{\Label_1 \compose \project{2}{{\Label_1}^2} \compose \project{1}{{\Label_2}^3}} \compose \Label_2 \union{}\\ \Label_3 \union \identity.$$
\end{example}

\begin{table}[htb!]
    \centering
    \begin{tabular}{l|l}
        Navigational language&Class of condition automata\\
        \hline
        \hline
        $\Lang()$& $\{\transop, \proop, \coproop\}$-free and acyclic.\\
        $\Lang(\proop)$& $\{\transop, \coproop\}$-free and acyclic.\\
        $\Lang(\proop, \coproop)$& $\{\transop\}$-free and acyclic.\\
        $\Lang(\transop)$& $\{\proop, \coproop \}$-free.\\
        $\Lang(\transop, \proop)$& $\{\coproop\}$-free.\\
        $\Lang(\transop, \proop, \coproop )$& no restrictions.\\
    \end{tabular}
    \caption{Navigational languages and the corresponding class of condition automata.}\label{tbl:expr_cta}
\end{table}

To show the path-equivalence of Table~\ref{tbl:expr_cta},  we first adapt standard closure properties for finite automata under composition, union, and Kleene plus to the setting of condition automata:

\begin{proposition}\label{prop:cta_close_basics}
Let $\Fragment \in \{ \transop, \proop, \coproop \}$ and let $\Automaton_1$ and $\Automaton_2$ be $\Fragment$-free condition automata. There exists $\Fragment$-free condition automata $\Automaton_\compose$, $\Automaton_\union$, and $\Automaton_\transop$ such that, for every graph $\Graph$,  $\Apply{\Automaton_\compose}{\Graph} = \Apply{\Automaton_1}{\Graph} \compose \Apply{\Automaton_2}{\Graph}$, $\Apply{\Automaton_\union}{\Graph} = \Apply{\Automaton_1}{\Graph} \union \Apply{\Automaton_2}{\Graph}$, and $\Apply{\Automaton_\transop}{\Graph} = \transitive{\Apply{\Automaton_1}{\Graph}}$. The condition automata $\Automaton_\compose$ and $\Automaton_\union$ are acyclic whenever $\Automaton_1$ and $\Automaton_2$ are acyclic.
\end{proposition}
\begin{proof}[Proof (sketch).]
Let $\Automaton_1 = (\States_1, \Alphabet_1, \Conditions_1, \Initials_1, \Finals_1, \Transitions_1, \StateConditions_1)$ and $\Automaton_2 = (\States_2, \Alphabet_2, \Conditions_2, \Initials_2, \Finals_2,\allowbreak \Transitions_2, \StateConditions_2)$ be $\Fragment$-free condition automata.  Without loss of generality, we may assume that $\States_1 \intersect \States_2 = \emptyset$. We define $\Automaton_\compose$, $\Automaton_\union$, and $\Automaton_\transop$ as follows:

\begin{enumerate}
\item $\Automaton_\compose = (\States_1 \union \States_2, \Alphabet_1 \union \Alphabet_2, \Conditions_1 \union \Conditions_2, \Initials_1, \Finals_2, \Transitions_1 \union \Transitions_2 \union \Transitions_\compose, \StateConditions_1 \union \StateConditions_2)$, in which $\Transitions_\compose = \{ (q_1, \identity, q_2) \mid (q_1 \in \Finals_1) \land (q_2 \in \Initials_2) \}$.
\item $\Automaton_\union = (\States_1 \union \States_2, \Alphabet_1 \union \Alphabet_2, \Conditions_1 \union \Conditions_2, \Initials_1 \union \Initials_2, \Finals_1 \union \Finals_2, \Transitions_1 \union \Transitions_2, \StateConditions_1 \union \StateConditions_2)$.
\item $\Automaton_\transop = (\States_1 \union \{ v, w \}, \Alphabet_1, \Conditions_1, \{ v \},  \{ w \}, \Transitions_1 \union \Transitions_\transop, \StateConditions_1)$, in which $v, w \notin \States_1$ are two distinct fresh states and $\Transitions_\transop = \{ (v, \identity, q) \mid q \in \Initials_1 \} \union \{ (q, \identity, w) \mid q \in \Finals_1 \} \union \{ (w, \identity, v) \}$.
\end{enumerate}

Observe that we did not add new condition expressions to the set of condition expressions in the proposed constructions. Hence, we conclude that $\Automaton_\compose$, $\Automaton_\union$, and $\Automaton_\transop$ are $\Fragment$-free whenever $\Automaton_1$ and $\Automaton_2$ are $\Fragment$-free. In  $\Automaton_\compose$ and $\Automaton_\union$, no new loops have been introduced, and, hence, $\Automaton_\compose$ and $\Automaton_\union$ are acyclic whenever $\Automaton_1$ and $\Automaton_2$ are acyclic.
\end{proof}

\begin{proposition}\label{prop:ne_to_cta}
Let $\Fragment \subseteq \{\transop, \proop, \coproop \}$. On labeled graphs, each navigational expression in $\Lang(\Fragment)$ is path-equivalent to some condition automaton in the class specified for $\Lang(\Fragment)$ in Table~\ref{tbl:expr_cta}.
\end{proposition}
\begin{proof}
 Let $\expr$ be a navigational expression in $\Lang(\Fragment)$ and let $\Alphabet$ be the set of all edge labels used in $\expr$. We translate $\expr$ to a condition automaton using structural induction. The base cases are described in Table~\ref{tbl:expr_cta_base}. The inductive cases are expressions of the form  $\expr = \expr_1 \compose \expr_2$, $\expr = \expr_1 \union \expr_2$, or $\expr =\transitive{\expr_1}$ with $\expr_1$ and $\expr_2$ navigational sub-expressions. For each of the inductive cases, we use the constructions needed to prove Proposition~\ref{prop:cta_close_basics}.\qedhere

\begin{table}[htb!]
    \centering
    \begin{tabular}{l|l}
        $\expr$&Condition automaton\\
        \hline
        \hline
        $\emptyset$&$\Automaton = (\{v, w\}, \Sigma, \emptyset, \{v\}, \{w\},\emptyset, \emptyset)$\\
        $\identity$&$\Automaton = (\{v, w\}, \Sigma, \emptyset, \{v\}, \{w\},\{ (v, \identity, w) \}, \emptyset)$\\
        $\Label$& $\Automaton = (\{v, w\}, \Sigma, \emptyset, \{v\}, \{w\},\{ (v, \Label, w) \}, \emptyset)$\\
        $\project{1}{\expr'}$&$\Automaton = (\{v\}, \Sigma, \{ \project{1}{\expr'} \}, \{v\}, \{v\},\emptyset, \{(v, \project{1}{\expr'})\})$\\
        $\project{2}{\expr'}$&$\Automaton = (\{v\}, \Sigma, \{ \project{2}{\expr'} \}, \{v\}, \{v\},\emptyset, \{(v, \project{2}{\expr'})\})$\\
        $\coproject{1}{\expr'}$&$\Automaton = (\{v\}, \Sigma, \{ \coproject{1}{\expr'} \}, \{v\}, \{v\},\emptyset, \{(v, \coproject{1}{\expr'})\})$\\
        $\coproject{2}{\expr'}$&$\Automaton = (\{v\}, \Sigma, \{ \coproject{2}{\expr'} \}, \{v\}, \{v\},\emptyset, \{(v, \coproject{2}{\expr'})\})$\\
    \end{tabular}
    \caption{Basic building blocks used by the translation from navigational expressions to condition automata. In the table, $\Label$ is an edge label.}\label{tbl:expr_cta_base}
\end{table}
\end{proof}

\begin{proposition}\label{prop:cta_to_ne}
Let $\Fragment \subseteq \{\transop, \proop, \coproop \}$. On labeled graphs, each condition automaton in the class specified for $\Lang(\Fragment)$ in Table~\ref{tbl:expr_cta} is path-equivalent to some expression in $\Lang(\Fragment)$. 
\end{proposition}
\begin{proof}
Let $\Automaton = (\States, \Alphabet, \Conditions, \Initials, \Finals, \Transitions, \StateConditions)$ be a condition automaton. Let $v, w \notin \States$ be two distinct fresh states. Let $\Automaton'  = (\States \union \{v, w\}, \Alphabet, \Conditions,\allowbreak \{v\}, \{w\}, \Transitions \union \Transitions_{v,w}, \StateConditions)$ with $\Transitions_{v,w} = \{ (v, \identity, q) \mid q \in \Initials \} \union \{ (q, \identity, w) \mid q \in \Finals \}$ be a condition automaton that is path-equivalent to $\Automaton$ and having only a single initial state and a single final state. We translate $\Automaton'$ into a navigational expression using Algorithm~\ref{alg:ctatone}.

\begin{algorithm}[ht!]
\caption{From condition automaton to navigational expression}\label{alg:ctatone}
\begin{algorithmic}[1]
\STATE We mark each state $q \in \States \union \{v, w\}$: $M[q] \GETS \true$
\STATE We construct navigational expressions $\expr_{q,r}$ between state $q \in \States \union \{v, w\}$ and $r \in \States \union \{v, w\}$ and initialize $$\expr_{q,r} \GETS \bigcup_{(q, \Label, r) \in \Transitions \union \Transitions_{v,w}} \cexpr{q} \compose \Label \compose \cexpr{r},$$ with $\expr_{q,r} \GETS \emptyset$ if there are no transitions between $q$ and $r$\label{line:inital_expr}
\WHILE{$\exists q\ (q \in  \States) \land (M[q] = \true)$}\label{line:step_while_base}
    \STATE Choose $q$ with $(q \in  \States) \land (M[q] = \true)$\label{line:while_start}
    \FOR{$p_1,p_2 \in  \States \union \{v, w\}$ with $q \notin \{ p_1,p_2 \}$}
        \STATE $\expr_{p_1,p_2} \GETS \expr_{p_1,p_2} \union \expr_{p_1,q}\compose\transitivestar{\expr_{q,q}}\compose\expr_{q, p_2}$\label{line:step_expr}
        \STATE If applicable, remove $\emptyset$ from $\expr_{p_1,p_2}$ or reduce $\expr_{p_1,p_2}$ to $\emptyset$
    \ENDFOR
    \STATE Unmark state $q$: $M[q] \GETS \false$\label{line:unmark_q}
\ENDWHILE
\RETURN $\expr_{v, w}$
\end{algorithmic}
\end{algorithm}

Let $\Graph = (\Nodes, \Alphabet, \ELabels)$ be a graph. We prove that the final navigational expression $\expr_{v,w}$ is path-equivalent to $\Automaton'$. We do so by proving the following invariants of Algorithm~\ref{alg:ctatone}:
\begin{enumerate}
\item \label{inv:emptypathfree} \emph{If no path exists from state $q_1$ to state $q_2$ of at least a single transition, with $q_1,q_2\in \States \cup \{ v, w\}$, then $\expr_{q_1, q_2} = \emptyset$.}

If there exists no path from $q_1$ to $q_2$, then also no transition exists from $q_1$ to $q_2$. Hence, we initialize $\expr_{q_1,q_2} = \emptyset$. After initialization, the value of $\expr_{q_1,q_2}$ only changes at line~\ref{line:step_expr}. As there exists no path from $q_1$ to $q_2$, we have one of the following three cases:
\begin{enumerate}
\item No path from $q_1$ to $q$ exists and there exists a path from $q$ to $q_2$. In this case we have $\expr_{q_1,q_2} = \emptyset \union \emptyset\compose\transitivestar{\expr_{q,q}}\compose\expr_{q, q_2}$ after line~\ref{line:step_expr}.
\item There exists a path from $q_1$ to $q$ and no path from $q$ to $q_2$ exists. In this case we have $\expr_{q_1,q_2} = \emptyset \union \expr_{q_1, q}\compose\transitivestar{\expr_{q,q}}\compose\emptyset$ after line~\ref{line:step_expr}.
\item No path from $q_1$ to $q$ exists and no path from $q$ to $q_2$ exists. In this case we have $\expr_{q_1,q_2} = \emptyset \union \emptyset\compose\transitivestar{\expr_{q,q}}\compose\emptyset$ after line~\ref{line:step_expr}.
\end{enumerate}
In all three cases, $\expr_{q_1,q_2}$ can be simplified to $\emptyset$ using Lemma~\ref{lem:remove_emptyset}.

\item \label{inv:loopfree} \emph{If $\Automaton'$ is acyclic, then $\expr_{q,q} = \emptyset$ for all $q \in \States \union \{ v, w \}$.}

Observe that $\Automaton'$ is acyclic if there exists no path from a state to itself of at least a single transition. Hence, by Invariant~\ref{inv:emptypathfree}, we have $\expr_{q,q} = \emptyset$.

\item \label{inv:exprinf} \emph{Every expression $\expr_{q_1, q_2}$, with $q_1,q_2\in \States \cup \{ v, w\}$, is a navigational expression in $\Lang(\Fragment)$.}

We initialize $\expr_{q_1, q_2}$ as either $\emptyset$ or a union of navigational expressions of the form $\cexpr{q_1} \compose \Label \compose \cexpr{q_2}$, with $\Label$ an edge label. Clearly, these navigational expressions are in $\Lang(\Fragment)$ if all condition expressions in $\Conditions$ are in $\Lang(\Fragment)$. After initialization, the value of $\expr_{q_1,q_2}$ only changes at line~\ref{line:step_expr}. Line~\ref{line:step_expr} does not introduce the operators $\proop$ and $\coproop$.

Line~\ref{line:step_expr} does introduce the operator $\transop$ via the operator $\transstarop$. The operator $\transstarop$ is only introduced for subexpressions of the form $\transitivestar{\expr_{q,q}}$. If $\Automaton'$ is acyclic, which must be the case when $\transop \notin \Fragment$, then, by Invariant~\ref{inv:loopfree}, we have $\expr_{q,q} = \emptyset$ and, using Lemma~\ref{lem:remove_emptyset}, we have $\transitivestar{\expr_{q,q}} = \transitivestar{\emptyset}$, which is path-equivalent to $\identity$.

\item \label{inv:exprtocta} \emph{If $(m,n) \in \Apply{\expr_{q_1,q_2}}{\Graph}$, with $q_1,q_2\in \States \cup \{ v, w\}$, then there exists a run $(t_1, m) \dots (t_i, n)$ of $\Automaton'$ on $\Graph$ with $t_1 = q_1$ and $t_i  = q_2$ that performs at least one transition.}

If at line~\ref{line:inital_expr} we have $(m,n) \in \Apply{\expr_{q_1,q_2}}{\Graph}$, then, by the initial construction of $\expr_{q_1,q_2}$ at line~\ref{line:inital_expr}, and, by the semantics of $\union$, there exists a subexpression $\cexpr{q_1}\compose \Label \compose \cexpr{q_2}$ in $\expr_{q_1,q_2}$ such that $(m,n) \in \Apply{\cexpr{q_1} \compose \Label \compose \cexpr{q_2}}{\Graph}$ and $(q_1, \Label, q_2) \in \Transitions \union \Transitions_{v,w}$. By the semantics of $\compose$, this implies $(m,m) \in \Apply{\cexpr{q_1}}{\Graph}$,  $(m,n) \in \Apply{\Label}{\Graph}$, and  $(n,n) \in \Apply{\cexpr{q_2}}{\Graph}$. Hence, we construct the run $(q_1, m) \Label (q_2,n)$ of $\Automaton'$ on $\Graph$.

Assume the Invariant holds before execution of line~\ref{line:step_expr}. Now consider the change made to $\expr_{q_1,q_2}$ when executing line~\ref{line:step_expr}. We denote the new value of $\expr_{q_1,q_2}$ by $\expr_{q_1,q_2}'$ for distinction. If $(m,n) \in \Apply{\expr_{q_1,q_2}'}{\Graph}$, then, by the construction of $\expr_{q_1,q_2}'$, there are two possible cases:
\begin{enumerate}
\item $(m,n) \in \Apply{\expr_{q_1,q_2}}{\Graph}$, in which case the Invariant can be applied to $\expr_{q_1,q_2}$ to provide the required run.
\item $(m,n) \notin \Apply{\expr_{q_1,q_2}}{\Graph}$ and $(m, n) \in \Apply{\expr_{q_1, q}\compose\transitivestar{\expr_{q,q}}\compose\expr_{q, q_2}}{\Graph}$. By the semantics of $\compose$, there exists nodes $m', n' \in \Nodes$ such that $(m, m') \in \Apply{\expr_{q_1, q}}{\Graph}$, $(m', n') \in \Apply{\transitivestar{\expr_{q,q}}}{\Graph}$, and $(n', n) \in \Apply{\expr_{q, q_2}}{\Graph}$. By applying the Invariant on $(m, m') \in \Apply{\expr_{q_1, q}}{\Graph}$ and $(n', n) \in \Apply{\expr_{q, q_2}}{\Graph}$, we conclude that there exists runs $(q_1, m) \dots (q, m')$ and $(q, n') \dots (q_2, n)$ of $\Automaton$ on $\Graph$. Since $\transitivestar{\expr_{q,q}} = \transitive{\expr_{q,q}} \union \identity$, there are two possible cases:
        \begin{enumerate}
            \item $(m', n') \in \Apply{\transitive{\expr_{q,q}}}{\Graph}$. By the semantics of $\transop$, there exists $k\in \NatPlus$ such that $(m',n') \in \Apply{\expr_{q,q}^k}{\Graph}$. By the semantics of $\compose$, there exists nodes $n_1, \dots, n_{k+1}$ with $m'= n_1$ and $n' = n_{k+1}$ such that, for $1 \leq i \leq k$, $(n_i, n_{i+1}) \in \Apply{\expr_{q,q}}{\Graph}$. By applying the Invariant on each $(n_i, n_{i+1}) \in \Apply{\expr_{q,q}}{\Graph}$, we conclude that there exists runs $(q, n_i) \dots (q, n_{i+1})$ of $\Automaton$ on $\Graph$. To construct the required run, we concatenate the runs $(q_1, m) \dots (q, m')$, $(q, n_1) \dots (q, n_2)$, $\dots$, $(q, n_k) \dots (q, n_{k+1})$, $(q, n') \dots (q_2, n)$.
            \item $(m', n') \in \Apply{\identity}{\Graph}$. Hence, $m' = n'$. To construct the required run, we concatenate the runs $(q_1, m) \dots (q, m')$ and $(q, n') \dots (q_2, n)$.
        \end{enumerate}
\end{enumerate}

\item \label{inv:increase_expr} \emph{If , at some point during the execution of the algorithm, $(m,n) \in \Apply{\expr_{q_1,q_2}}{\Graph}$, with $q_1,q_2\in \States \cup \{ v, w\}$, then $(m,n) \in \Apply{\expr_{q_1,q_2}}{\Graph}$ at all later steps.}

Follows immediately from inspecting line~\ref{line:step_expr} of the algorithm.

\item \label{inv:ctatoexpr} \emph{If $(q_0, n_0) \Label_0 (q_1, n_1) \Label_1 \dots (q_{i-1}, n_{i-1}) \Label_{i-1} (q_i, n_i)$ is a run of $\Automaton'$ on $\Graph$, with $M[q_1] = M[q_2] = \dots = M[q_{i-1}] = \false$, which performs at least one transition, then $(n_0, n_i) \in \Apply{\expr_{q_0,q_i}}{\Graph}$.}

First, we consider the runs that meet the conditions of the Invariant before execution of the while-loop starting at line~\ref{line:step_while_base}. Initially, at line~\ref{line:inital_expr}, all states are marked, and, hence, only runs of the form $(q_0, n_0) \Label (q_1, n_1)$ of $\Automaton'$ on $\Graph$ satisfy the restrictions. By the definition of a run, we have $(q_0, \Label, q_1) \in \Transitions \union \Transitions_{v,w}$. By the initial construction of $\expr_{q_0, q_1}$ at line~\ref{line:inital_expr}, $\expr_{q_0,q_1}$ is a union that includes the subexpression $\cexpr{q_0} \compose \Label \compose \cexpr{q_1}$. By the definition of a run, we have $(n_0, n_0) \in \Apply{\cexpr{q_0}}{\Graph}$, $(n_0, n_1) \in \Apply{\Label}{\Graph}$, and $(n_1, n_1) \in \Apply{\cexpr{q_1}}{\Graph}$. Hence, $(n_0, n_1) \in \Apply{\cexpr{q_0} \compose \Label \compose \cexpr{q_1}}{\Graph}$, yielding $(n_0, n_1) \in \Apply{\expr_{q_0,q_1}}{\Graph}$.

Next we consider the runs that meet the conditions of the Invariant due to changes made during execution of the while-loop starting at line~\ref{line:step_while_base}. Assume the Invariant holds before execution of line~\ref{line:while_start}. Now consider that we unmark state $q$ by executing line~\ref{line:unmark_q}. Let $$(q_0, n_0) \Label_0 (q_1, n_1) \Label_1 \dots (q_{i-1}, n_{i-1}) \Label_{i-1}(q_i, n_i)$$ be a run of $\Automaton'$ on $\Graph$ that performs at least one transition with $M[q_1] = M[q_2] = \dots = M[q_{i-1}] = \false$. We have two cases.
    \begin{enumerate}
        \item $q \notin \{ q_1, \dots, q_{i-1} \}$. Hence, by the Invariant, we had $(n_0, n_i) \in \Apply{\expr_{q_0,q_i}}{\Graph}$ at line~\ref{line:while_start}. By Invariant~\ref{inv:increase_expr}, we conclude that, when executing line~\ref{line:unmark_q}, we have $(n_0, n_i) \in \Apply{\expr_{q_0,q_i}}{\Graph}$.
        \item $q \in \{ q_1, \dots, q_{i-1} \}$.  During an iteration of the while loop at line~\ref{line:step_while_base}, the value of $\expr_{q_0, q_i}$ is changed by the execution of line~\ref{line:step_expr}. We denote the new value of $\expr_{q_0,q_i}$ by $\expr_{q_0,q_i}'$ for distinction. We have
            \begin{align*}
                \expr_{q_0,q_i}' &= \expr_{q_0,q_i} \union \expr_{q_0,q}\compose\transitivestar{\expr_{q,q}}\compose\expr_{q,q_i}\\
                                &= \expr_{q_0,q_i} \union \expr_{q_0,q}\compose\left(\transitive{\expr_{q,q}} \union \identity\right)\compose\expr_{q,q_i}\\
                                &= \expr_{q_0,q_i} \union \expr_{q_0,q}\compose\transitive{\expr_{q,q}}\compose\expr_{q,q_i} \union \expr_{q_0,q}\compose\identity\compose\expr_{q,q_i}
                \end{align*}
        
        Based on the number of occurrences of $q$ in $\{ q_1, \dots, q_{i-1} \}$, we again distinguish two cases::
        \begin{enumerate}
            \item There exists exactly one $j$, $1\leq j\leq i-1$, with $q_j = q$. We split the run into $(q_0, n_0)\dots(q_j, n_j)$ and $(q_j, n_j)\dots(q_i, n_i)$. At line~\ref{line:while_start}, we had $M[q_1] = \dots = M[q_{j-1}] = M[q_{j+1}] = \dots = M[q_{i-1}] = \false$. Hence, by the Invariant, we had $(n_0, n_j) \in \Apply{\expr_{q_0,q}}{\Graph}$ and $(n_j, n_i) \in \Apply{\expr_{q,q_i}}{\Graph}$ at line~\ref{line:while_start}. 
It follows that $(n_0,n_i) \in \Apply{\expr_{q_0,q} \compose \identity \compose \expr_{q,q_i}}{\Graph}$. Hence, we have $(n_0, n_i) \in \Apply{\expr_{q_0,q_i}'}{\Graph}$.
            
            \item There exists several $j$, $1\leq j\leq i-1$ with $q_j = q$. We split the run $(q_0, n_0) \Label_0 (q_1, n_1) \Label_1 \dots (q_{i-1}, n_{i-1}) \Label_{i-1}(q_i, n_i)$ at each $(q_j, n_j)$, $1 \leq j \leq i$, with $q_j = q$ resulting in runs $(q_1, n_1) \dots (q, n'_1)$, $(q, n'_1) \dots (q, n'_2)$, $\dots$, $(q, n'_{k-1}) \dots (q, n'_k)$, $(q, n'_k) \dots (q_i, n_i)$, with $1 \leq k$.
            
            At line~\ref{line:while_start} we had $M[q'] = \false$ for all $q' \in \{q_1, \dots, q_{i-1} \} \difference \{ q \}$. Hence, by the Invariant, we had $(n_1, n'_1) \in \Apply{\expr_{q_1, q}}{\Graph}$, $(n'_l, n'_{l+1}) \in \Apply{\expr_{q, q}}{\Graph}$, for $1 \leq l < k$, and $(n'_k, n_i) \in \Apply{\expr_{q, q_i}}{\Graph}$. It follows that $(n'_1,n'_k)\in \Apply{\transitive{\expr_{q,q}}}{\Graph}$ and $(n_0,n_i) \in \Apply{\expr_{q_0,q} \compose \transitive{\expr_{q,q}} \compose \expr_{q,q_i}}{\Graph}$. Hence, we have $(n_0, n_i) \in \Apply{\expr_{q_0,q_i}'}{\Graph}$.
        \end{enumerate}
 In both cases we use Invariant~\ref{inv:increase_expr} to conclude that, when executing line~\ref{line:unmark_q}, we have $(n_0, n_i) \in \Apply{\expr_{q_0,q_i}}{\Graph}$.        
    \end{enumerate}
\end{enumerate}

As $v \neq w$, $v$ is the only initial state, and $w$ is the only final state, each accepting run of $\Automaton'$ performs at least one transition. Hence Invariants~\ref{inv:exprtocta} and~\ref{inv:ctatoexpr} imply that $\Automaton'$ and the resulting navigational expression $\expr_{v,w}$ are path-equivalent. Invariant~\ref{inv:exprinf} implies that the resulting navigational expression is, as required, in $\Lang(\Fragment)$.
\end{proof}

Notice that we did not only prove path-equivalence between classes of condition automata and classes of the navigational expressions on general labeled graphs, but also provided constructive algorithms to translate between these classes.

In the following, we work towards showing that condition automata, when used as queries on \emph{trees}, are not only closed under $\compose$, $\union$, and $\transop$, as Proposition~\ref{prop:cta_close_basics} showed, but also under $\intersect$ and $\difference$. We then use this closure result to remove $\intersect$ and $\difference$ from navigational expressions that are used to query trees. The standard approach to constructing the intersection of two finite automata is by making their cross-product. In a fairly straightforward manner, we can apply a similar cross-product construction to condition automata, given that they are $\identity$-transition free. Observe that the $\identity$-labeled transitions fulfill a similar role as empty-string-transitions in finite automata and, as such, can be removed, which we show next.

\begin{definition}
Let $\Automaton = (\States, \Alphabet, \Conditions, \Initials, \Finals, \Transitions, \StateConditions)$ be a condition automaton. The pair $(q, \{ q, q_1, \dots, q_i\})$, with $q, q_1, \dots, q_i \in \States$, is an \emph{identity pair} of $\Automaton$ if there exists a path $q\,\identity\,q'_1\dots\,\identity\, q'_j$ in $\Automaton$ with $\{q, q_1, \dots, q_i\} = \{q, q'_1, \dots, q'_j \}$.
\end{definition}

\begin{lemma}\label{lem:identity_free}
Let $\Fragment \in \{ \transop, \proop, \coproop \}$ and let $\Automaton$ be an $\Fragment$-free condition automaton. There exists an $\identity$-transition free and $\Fragment$-free condition automaton $\Automaton_\nIdentity$ that is path-equivalent to $\Automaton$ with respect to labeled graphs. The condition automaton $\Automaton_\nIdentity$ is acyclic whenever $\Automaton$ is acyclic.
\end{lemma}
\begin{proof}
Let $\Automaton = (\States, \Alphabet, \Conditions, \Initials, \Finals, \Transitions, \StateConditions)$ be an $\Fragment$-free condition automaton. We construct $\Automaton_{\nIdentity} = (\States_{\nIdentity}, \Alphabet, \Conditions, \Initials_{\nIdentity}, \Finals_{\nIdentity}, \Transitions_{\nIdentity}, \StateConditions_{\nIdentity})$ with
\begin{align*}
\States_\nIdentity        &= \{ (q, Q) \mid \text{$(q, Q)$ is an identity pair of $\Automaton$} \};\\
\Initials_{\nIdentity}    &= \{ (q, Q) \mid ((q, Q) \in \States_\nIdentity) \land (q \in \Initials) \};\\
\Finals_{\nIdentity}      &= \{ (q, Q) \mid ((q, Q) \in \States_\nIdentity) \land (Q \intersect \Finals \neq \emptyset) \};\\
\Transitions_{\nIdentity} &= \{ ((p, P), \Label, (q, Q)) \mid ((p, P) \in \States_\nIdentity) \land  (\Label \in \Alphabet) \land{}\\
                          &\hphantom{{}=\qquad} ((q, Q) \in \States_\nIdentity) \land{} (\exists p'\ (p' \in P) \land (p', \Label, q) \in \Transitions) \};\\
\StateConditions_{\nIdentity} &= \{ ((q, Q), \cond) \mid ((q, Q) \in \States_\nIdentity) \land (\exists q'\ (q' \in Q) \land (\cond \in \StateConditions(q'))) \}.
\end{align*}

Let $\Graph = (\Nodes, \Alphabet, \ELabels)$ be a graph. To prove that $\Automaton_\nIdentity$ is path-equivalent to $\Automaton$, we prove $\Apply{\Automaton_\nIdentity}{\Graph} = \Apply{\Automaton}{\Graph}$.

\begin{enumerate}
\item $\Apply{\Automaton_\nIdentity}{\Graph} \subseteq \Apply{\Automaton}{\Graph}$. Assume $(n_0, n_i) \in \Apply{\Automaton_\nIdentity}{\Graph}$. Hence, there exists a run $((q_0, Q_0), n_0)\dots((q_i, Q_i), n_i)$ of $\Automaton_\nIdentity$ on $\Graph$ with $(q_0, Q_0) \in \Initials_{\nIdentity}$ and $(q_i, Q_i) \in \Finals_{\nIdentity}$. By the definition of $\Finals_{\nIdentity}$, we have $Q_i \intersect \Finals \neq \emptyset$. Let $p \in Q_i \intersect \Finals$. For each $j$, $0 \leq j \leq i$, we shall prove that there exists a run $(q_j, n_j)\dots(p, n_i)$ of $\Automaton$ on $\Graph$ using induction on $i - j$.

The base case is $j = i$. Consider $((q_i, Q_i), n_i)$. For each $q \in Q_i$, $n_i$ satisfies $q$ by the definition of $\States_\nIdentity$ and $\StateConditions_\nIdentity$. Moreover, there exists a path $q_i\,\identity\, s_1\dots s_k\,\identity\, p$ in $\Automaton$ from $q_i$ to $p$ with $q_i, s_1, \dots, s_k \in Q_i$. Hence, $(q_i, n_i)\,\identity\,(s_1, n_i)\dots(s_k, n_i)\,\identity\,(p, n_i) $ is a run of $\Automaton$ on $\Graph$.

Now assume as induction hypothesis that, for all $j$, $0 < k \leq j \leq i$, there exists a run $(q_{j}, n_{j})\dots(p, n_i)$ of $\Automaton$ on $\Graph$. Consider the $k$-th step in the run, $((q_{k-1}, Q_{k-1}), n_{k-1}) \Label_{k-1} ((q_k, Q_k), n_k)$. By the definition of a run, we have $((q_{k-1}, Q_{k-1}), \Label_{k-1}, (q_{k}, Q_{k})) \in \Transitions_{\nIdentity}$, $n_{k-1}$ satisfies $(q_{k-1}, Q_{k-1})$, and $n_k$ satisfies $(q_k, Q_k)$. By the construction of $\States_\nIdentity$ and $\StateConditions_{\nIdentity}$, we have, for each $q \in Q_{k-1}$, $n_{k-1}$ satisfies $q$, and, for each $q \in Q_{k}$, $n_{k}$ satisfies $q$. By the construction of $\Transitions_{\nIdentity}$, there exists a state $p_{k-1} \in Q_{k-1}$ such that $(p_{k-1}, \Label, q_{k}) \in \Transitions$ and, by the construction of $\States_{\nIdentity}$, there exists a path $q_{k-1}\,\identity\, s_1\dots s_{i'}\,\identity\, p_{k-1}$ in $\Automaton$ from $q_{k-1}$ to $p_{k-1}$ with $s_1, \dots, s_{i'} \in Q_{k-1}$. Hence, $(q_{k-1}, n_{k-1})\,\identity\,(s_1, n_j)\dots(s_{i'}, n_j)\,\identity\,(p_{k-1}, n_j)\Label_j(q_k, n_k)$ is a run of $\Automaton$ on $\Graph$. Using the induction hypothesis on $k$, we conclude that there exists a run $(q_k, n_k)\dots(p, n_i)$ of $\Automaton$ on $\Graph$. We concatenate these runs to conclude that there exists a run $(q_{k-1}, n_{k-1})\dots(p, n_i)$ of $\Automaton$ on $\Graph$.

By the definition of $\Initials_{\nIdentity}$, we have $q_0 \in \Initials$. Hence, we conclude that a run $(q_0, n_0)\dots(p, n_i)$ of $\Automaton$ on $\Graph$ with $q_0 \in \Initials$ and $p \in \Finals$ exists, and, as a consequence, $(n_0, n_i) \in \Apply{\Automaton}{\Graph}$.

\item $\Apply{\Automaton_\nIdentity}{\Graph} \supseteq \Apply{\Automaton}{\Graph}$. Assume $(n_0, n_i) \in \Apply{\Automaton}{\Graph}$. Hence, there exists a run $(q_0, n_0)\dots(q_i, n_i)$ of $\Automaton$ on $\Graph$ with $q_0 \in \Initials$ and $q_i \in \Finals$. For each $j$, $0 \leq j \leq i$, we shall prove that there exists a run $((q_j, Q), n_j)\dots((p, P), n_i)$ of $\Automaton_\nIdentity$ on $\Graph$ with $q_i \in P$ using induction on $i - j$.

The base case is $j = i$. Observe that $(q_i, \{ q_i \})$ is an identity pair. Hence, $(q_i, \{ q_i \}) \in \States_\nIdentity$. By the construction of $\StateConditions_\nIdentity$, $n_i$ satisfies $(q_i, \{ q_i \})$. We conclude that, $(q_i, \{ q_i \})$ is a run of $\Automaton_\nIdentity$ on $\Graph$.

Now assume as induction hypothesis that, for all $k$, $0 < k \leq j \leq i$, there exists a run $((q_j, Q), n_j)\dots\allowbreak((p, P), n_i)$ of $\Automaton_\nIdentity$ on $\Graph$ with $q_i \in P$. Consider the $k$-th step in the run, $(q_{k-1}, n_{k-1}) \Label_{k-1}(q_{k}, n_{k})$. By the induction hypothesis, there exists a run $r = ((q_{k}, Q), n_{k})\dots((p, P), n_i)$ of $\Automaton_\nIdentity$ on $\Graph$ with $q_i \in P$.  We distinguish two cases:
\begin{enumerate}
\item $\Label_{k-1} \neq \identity$. By the definition of a run, we have $(n_{k-1}, n_k) \in \Label_{k-1}$. By the construction of $\States_\nIdentity$, we have $(q_{k-1}, \{ q_{k-1} \}) \in \States_\nIdentity$, by the construction of $\StateConditions_\nIdentity$, we have $n_{k-1}$ satisfies $(q_{k-1}, \{ q_{k-1} \})$, and, by the construction of $\Transitions_\nIdentity$, we have $((q_{k-1}, \{ q_{k-1} \}), \Label_{k-1}, (q_{k}, Q)) \in \Transitions_\nIdentity$. Hence, $((q_{k-1}, \{ q_{k-1} \}), n_{k-1})\Label_{k-1}((q_{k}, Q), n_{k})\dots((p, P), n_i)$ is a run of $\Automaton_\nIdentity$ on $\Graph$.

\item $\Label_{k-1} = \identity$. By the semantics of $\identity$, we have $n_{k-1} = n_{k}$. Let $Q = \{q_{k}, p_1, \dots, p_{i'} \}$. By the construction of $\States_\nIdentity$, there exists a path $q_{k}\,\identity\, p_1\dots p_{i'}$ in $\Automaton$ and, by the definition of a run, we have $(q_{k-1}, \identity, q_{k}) \in \Transitions$. Hence, we conclude that $q_{k-1}\,\identity\, q_{k}\,\identity\, p_1\dots p_{i'}$ is a path in $\Automaton$, $(q_{k-1}, Q \cup \{ q_{k-1} \})$ is an identity pair, and, by the construction of $\States_\nIdentity$, we have $(q_{k-1}, Q \cup \{ q_{k-1} \}) \in \States_\nIdentity$. We again distinguish two cases:
\begin{enumerate}
\item $r = ((q_{k}, Q), n_{k})$. In this case, we have $q_i \in Q$, and, hence, $q_i \in Q \cup \{ q_{k-1} \}$. We conclude that $((q_{k-1}, Q \cup \{ q_{k-1} \}), n_k)$ is a run of $\Automaton_\nIdentity$ with $q_i \in Q \cup \{ q_{k-1} \}$.
\item $r = ((q_k, Q), n_k)\Label'((q', Q'), n')\dots((p, P), n_i)$. In this case we have, by the definition of a run, $((q_k, Q), \Label', (q', Q'))\in \Transitions_\nIdentity$. By the definition of $\Transitions_\nIdentity$, we have $((q_k, Q), \Label', (q', Q'))\in \Transitions_\nIdentity$ if and only if there exists $q'' \in Q$ such that $(q'', \Label', q') \in \Transitions$. It follows that $q'' \in Q \cup \{ q_{k-1} \}$ and $((q_{k-1}, Q \cup \{ q_{k-1} \}), \Label', (q', Q'))\in \Transitions_\nIdentity$. So, $((q_{k-1}, Q \cup \{ q_{k-1} \}), n_k)\Label'((q', Q'), n')\dots((p, P), n_i)$ is a run of $\Automaton_\nIdentity$.
\end{enumerate}
\end{enumerate}
By the definition of $\Initials_{\nIdentity}$, we have $(q_0, Q) \in \Initials_\nIdentity$ and by the definition of $\Finals_\nIdentity$ and $q_i \in P$, we have $(p, P) \in \Finals_\nIdentity$. Hence, we conclude that the run $((q_0, Q), n_0)\dots((p, P), n_i)$ is a run of $\Automaton_\nIdentity$ on $\Graph$ with $(q_0, Q) \in \Initials_\nIdentity$ and $(p, P) \in \Finals_\nIdentity$, and, as a consequence, $(n_0, n_i) \in \Apply{\Automaton_\nIdentity}{\Graph}$.
\end{enumerate}

As we did not add new condition expressions to the set of condition expressions in the above constructions, it follows that $\Automaton_\nIdentity$ is $\Fragment$-free whenever $\Automaton$ is $\Fragment$-free. As each path in $\Automaton_\nIdentity$ can be translated to a path of at least equal length in $\Automaton$ using the same reasoning as in the proof of $\Apply{\Automaton_\nIdentity}{\Graph} \subseteq \Apply{\Automaton}{\Graph}$, it finally follows that $\Automaton_\nIdentity$ is acyclic whenever $\Automaton$ is acyclic.
\end{proof}

Hence, we may assume that condition automata are $\identity$-transition free.

\begin{example}
In Figure~\ref{fig:example_identity} two condition automata are shown. The condition automaton on the \emph{left} is a simple automaton with $\identity$-transitions. The $\identity$-transition free condition automaton on the \emph{right} is obtained by applying the construction of Lemma~\ref{lem:identity_free}. The main step in constructing the condition automaton on the \emph{right} is constructing the identity pairs (as these are the  states of the constructed condition automaton). Observe that the condition automaton on the \emph{left} has the following paths consisting of identity-transitions only:
\[ u,\quad u\,\identity\,v, \quad u\,\identity\,v\,\identity\,w, \quad v, \quad v\,\identity\,w,\quad w, \]
resulting in the identity pairs $(u, \{ u\})$, $(u, \{u,v\})$, $(u, \{u,v,w\})$, $(v, \{v \})$, $(v, \{v,w\})$, and $(w, \{w\})$, which are the states in the condition automaton on the \emph{right}.

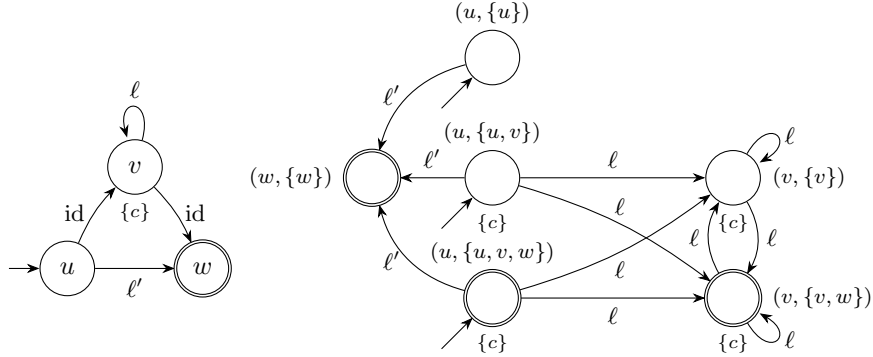
\begin{figure}[ht!]
\centering
    \begin{tikzpicture}[automata_style,scale=0.9]
        \node                  (s) at (-1, 0) {};
        \node[state]           (u) at (0, 0) {$u$};
        \node[state]           (v) at (1, 1.5) {$v$};
        \node[state,accepting] (w) at (2, 0) {$w$};  
        
        \node[conditions,below of=v] {$\{ \cond \}$};

        \begin{scope}[on background layer]
            \path[->] (s) edge (u)
                      (u) edge[bend left=10] node[label,left] {$\identity$} (v)
                      (v) edge[bend left=10] node[label,right] {$\identity$} (w)
                      (v) edge[in=105,out=75,loop] node[label] {$\Label$} (v)
                      (u) edge node[label,below] {$\Label'$} (w);
        \end{scope}
    \end{tikzpicture}
    \begin{tikzpicture}[automata_style,scale=0.8]
        \node                  (su)   at (1, 5) {};
        \node                  (suv)  at (1, 3) {};
        \node                  (suvw) at (1, 1) {};
        \node[state]           (u)    at (2, 6) {};
        \node[state]           (uv)   at (2, 4) {};
        \node[state,accepting] (uvw)  at (2, 2) {};
        \node[state]           (v)    at (6, 4) {};
        \node[state,accepting] (vw)   at (6, 2) {};
        \node[state,accepting] (w)    at (0, 4) {};

        \node[idset,above of=u]   {$(u, \{ u \})$};
        \node[idset,above of=uv]  {$(u, \{ u, v \})$};
        \node[idset,above of=uvw] {$(u, \{ u, v, w \})$};
        \node[idset,right=0.4cm] at (v)   {$(v, \{ v \})$};
        \node[idset,right=0.4cm] at (vw)  {$(v, \{v, w \})$};
        \node[idset,left=0.4cm] at (w)  {$(w, \{ w \})$};
        
        \node[conditions,below of=uv]  {$\{ \cond \}$};
        \node[conditions,below of=uvw] {$\{ \cond \}$};
        \node[conditions,below of=v]   {$\{ \cond \}$};
        \node[conditions,below of=vw]  {$\{ \cond \}$};
        
        \begin{scope}[on background layer]
            \path[->] (su) edge (u)
                      (suv) edge (uv)
                      (suvw) edge (uvw)
                      (u) edge[bend right] node[label,left] {$\Label'$} (w)
                      (uv) edge node[label] {$\Label'$} (w)
                      (uvw) edge[bend left] node[label,left] {$\Label'$} (w)
                      (uv) edge node[label] {$\Label$} (v)
                      (uv) edge[bend left=10] node[label] {$\Label$} (vw)
                      (uvw) edge[bend right=10] node[label,below] {$\Label$} (v)
                      (uvw) edge node[label,below] {$\Label$} (vw)
                      (v) edge[bend left] node[label,right] {$\Label$} (vw)
                      (vw) edge[bend left] node[label,left] {$\Label$} (v)
                      (v) edge[in=30,out=60,loop] node[label,right] {$\Label$} (v)
                      (vw) edge[in=330,out=300,loop] node[label,right] {$\Label$} (vw);
        \end{scope}
    \end{tikzpicture}
    \caption{Two path-equivalent condition automata, only the one on the \emph{right} is $\identity$-transition free.}\label{fig:example_identity}
\end{figure}
\end{example}

We now proceed with showing that condition automata, when used to query trees as opposed to general labeled graphs, are closed under intersection ($\intersect$). We already know from Example~\ref{exam:intro_loop_int} that on general graphs standard automata-techniques cannot be adapted to obtain closure results. On trees, however, the situation of Example~\ref{exam:intro_loop_int} cannot occur, as a directed path between two nodes in a tree is always unique. This observation is crucial in showing that the cross-product construction on condition automata works for querying trees. The lemma below formalizes this observation:

\begin{lemma}\label{lem:run_identity_free}
Let $\Automaton_1$ and $\Automaton_2$ be $\identity$-transition free condition automata and let $\Tree = (\Nodes, \Alphabet, \ELabels)$ be a tree. If there exists a run $r_1 = (p_1,n_1)\Label_1^1\dots\Label_{i_1}^1(q_1, n_{i_1 + 1})$ of $\Automaton_1$ on $\Tree$ and there exists a run $r_2 = (p_2, m_1)\Label_1^2\dots\Label_{i_2}^2\allowbreak(q_2, m_{i_2 + 1})$ of $\Automaton_2$ on $\Tree$ with $n_1 = m_1$ and $n_{i_1 + 1} = m_{i_2 +1}$, then $i_1 = i_2 = i$ and, for all $1 \leq j \leq i$, $\Label_j^1 = \Label_j^2$ and $n_j = m_j$.
\end{lemma}
\begin{proof}
Let $m = m_1 = m_2$ and $n = n_1 = n_2$. By the semantics of condition automata, the existence of run $r_1$ implies that $(m, n) \in \Apply{\Label_1^1\compose\dots\compose\Label_{i_1}^1}{\Tree}$ and the existence of run $r_2$ implies that $(m, n) \in \Apply{\Label_1^2\compose\dots\compose\Label_{i_2}^2}{\Tree}$. As $\Automaton_1$ and $\Automaton_2$ are $\identity$-transition free, each $\Label_{j_1}^1$ and each $\Label_{j_2}^2$, with $1 \leq {j_1} \leq i_1$ and $1 \leq {j_2} \leq i_2$, is an edge label. As there is only a single downward path from node $m$ to node $n$ in $\Tree$, the two runs must traverse the same path, and, hence, follow the same edge labels. We conclude that $i_1 = i_2 = i$ and, for all $1 \leq j \leq i$, $\Label_j^1 = \Label_j^2$.
\end{proof}

This allows us to prove the following:
\begin{proposition}\label{prop:closed_intersect}
Let $\Fragment \in \{ \transop, \proop, \coproop \}$ and let $\Automaton_1$ and $\Automaton_2$ be $\Fragment$-free condition automata. There exists an $\Fragment$-free condition automaton $\Automaton_\intersect$ such that, for every tree $\Tree$, we have $\Apply{\Automaton_\intersect}{\Tree} = \Apply{\Automaton_1}{\Tree} \intersect \Apply{\Automaton_2}{\Tree}$. The condition automaton $\Automaton_\intersect$ is acyclic whenever $\Automaton_1$ or $\Automaton_2$ is acyclic.
\end{proposition}
\begin{proof}
Let $\Automaton_1 = (\States_1, \Alphabet_1, \Conditions_1, \Initials_1, \Finals_1, \Transitions_1, \StateConditions_1)$ and $\Automaton_2 = (\States_2, \Alphabet_2, \Conditions_2, \Initials_2, \Finals_2, \Transitions_2, \StateConditions_2)$ be condition automata. By Lemma~\ref{lem:identity_free}, we assume that $\Automaton_1$ and $\Automaton_2$ are $\identity$-transition free. Without loss of generality, we may assume that $\States_1 \intersect \States_2 = \emptyset$.

We construct $\Automaton_\intersect = (\States_1 \times \States_2, \Alphabet_1 \union \Alphabet_2, \Conditions_1 \union \Conditions_2, \Initials_1 \times \Initials_2, \Finals_1 \times \Finals_2, \Transitions_\intersect, \StateConditions_\intersect)$ where
\begin{align*}
    \Transitions_\intersect     &= \{ ((p_1, q_1), \Label, (p_2, q_2)) \mid (p_1, \Label, p_2) \in \Transitions_1 \land (q_1, \Label, q_2) \in \Transitions_2 \}\text{ and}\\
    \StateConditions_\intersect &= \{ ((p, q), \cond) \mid (p, \cond) \in \StateConditions_1 \lor (q, \cond) \in \StateConditions_2  \}.
\end{align*}

Let $\Tree = (\Nodes, \Alphabet, \ELabels)$ be a tree and let $m, n \in \Nodes$ be a pair of nodes. We have $(m, n) \in \Apply{\Automaton_1}{\Tree} \intersect \Apply{\Automaton_2}{\Tree}$ if and only if there exists a run $(p_1,m)\Label_1^1\dots\Label_{i_1}^1(q_1, n)$ of $\Automaton_1$ on $\Tree$ with $p_1 \in \Initials_1$ and $q_1 \in \Finals_1$ and a run $(p_2, m)\Label_1^2\dots\Label_{i_2}^2(q_2, n)$ of $\Automaton_2$ on $\Tree$ with $p_2 \in \Initials_2$ and $q_2 \in \Finals_2$. Since $\Automaton_1$ and $\Automaton_2$ are $\identity$-transition free, by Lemma~\ref{lem:run_identity_free}, and by the construction of $\Automaton_\intersect$, these runs exist if and only if there exists a run $((p_1, p_2), m)\Label_1\dots\Label_i((q_1, q_2), n)$ of $\Automaton_\intersect$ on $\Tree$ with $(p_1, p_2) \in \Initials_1 \times \Initials_2$ and $(q_1, q_2) \in \Finals_1 \times \Finals_2$. Hence, we conclude $(m, n) \in \Apply{\Automaton_\intersect}{\Tree}$.

Observe that we did not add new condition expressions to the set of condition expressions in the proposed constructions. Hence, we conclude that $\Automaton_\intersect$ is $\Fragment$-free whenever $\Automaton_1$ and $\Automaton_2$ are $\Fragment$-free. We also observe that each run $r$ of $\Automaton_\intersect$ on $\Tree$ can be split into runs of $\Automaton_1$ and $\Automaton_2$ on tree $\Tree$ of the same length as $r$. Hence, there can only be loops in $\Automaton_\intersect$ if both $\Automaton_1$ and $\Automaton_2$ have loops and we conclude that $\Automaton_\intersect$ is acyclic whenever $\Automaton_1$ or $\Automaton_2$ is acyclic.
\end{proof}

Next, we show that condition automata, when used as tree queries, are also closed under difference ($\difference$). Usually, the difference of two finite automata $\Automaton_1$ and $\Automaton_2$ is constructed by first constructing the complement of $\Automaton_2$, and then constructing the intersection of $\Automaton_1$ with the resulting automaton. We cannot use such a complement construction for condition automata: the complement of a downward binary relation (represented by a condition automaton when evaluated on a tree) is not a downward binary relation. Observe, however, that it is not necessary to consider the full complement for this purpose. As the difference of two downward binary relations is itself a downward relation, we can restrict ourselves to the \emph{downward complement} of a binary relation.

\begin{definition}
Let $\Tree = (\Nodes, \Alphabet, \ELabels)$ be a tree. We define the \emph{downward complement} of a binary relation $\Relation \subseteq \Nodes \times \Nodes$, denoted by $\dcompl{\Relation}$, as \[\dcompl{\Relation} = \{ (m, n) \mid (m, n) \notin \Relation \land (m,n) \in \Apply{\transitivestar{\Edges}}{\Tree} \}.\]
\end{definition}

Indeed, if $\Automaton_1$ and $\Automaton_2$ are condition automata and $\Tree$ is a tree, then we have $\Apply{\Automaton_1}{\Tree} \difference \Apply{\Automaton_2}{\Tree} \equiv \Apply{\Automaton_1}{\Tree} \intersect \dcompl{\Apply{\Automaton_2}{\Tree}}$. Hence, we only need to show that condition automata are closed under downward complement. The construction of the downward complement uses deterministic condition automata:

\begin{definition}\label{def:dca}
The condition automaton $\Automaton = (\States, \Alphabet, \Conditions, \Initials, \Finals, \Transitions, \StateConditions)$ is \emph{deterministic} if it is $\identity$-transition free and if it satisfies the following condition: for every tree $\Tree = (\Nodes, \Alphabet, \ELabels)$ and for every pair of nodes $m, n$ with $m$ an ancestor of $n$, there exists exactly one run $(q, m)\Label\dots(p, n)$ of $\Automaton$ on $\Tree$ with $q \in \Initials$.
\end{definition}

We observe that if the condition automaton does not specify any conditions, then Definition~\ref{def:dca} reduces to the classical definition of a finite deterministic automaton. Moreover, the definition of a deterministic condition automaton relies on the automaton being evaluated on trees, as more general graphs can have several identically-labeled paths between pairs of nodes.

\begin{example}
The condition automaton in Figure~\ref{fig:example_automaton} is clearly not deterministic: there are already two different possible runs of length one starting at an initial state. In Figure~\ref{fig:example_dautomaton} we visualize a conditional automaton over $\Alphabet = \{ \Label_1, \Label_2 \}$ that is deterministic. This deterministic condition automaton accepts node pairs $(m, n)$, $m \neq n$, if $m$ satisfies $\project{2}{{\Label_1}^3}$ and if there exists a path from $m$ to $n$ whose labeling matches the regular expression $\Label_1 \transitivestar{\Label_2} \Label_1$. It also accepts node pairs $(n, n)$ if $n$ does not satisfy $\project{2}{{\Label_1}^3}$.
\begin{figure}[htb!]
    \centering
    \begin{tikzpicture}[automata_style]
        \node        (s1)  at (-1, 0) {};
        \node[state] (n1) at (0, 0) {$q_1$};
        \node[state] (n2) at (2, 0) {$q_2$};
        \node[state,accepting] (n3) at (4, 0) {$q_3$};  
        \node        (s4)  at (-1, 2) {};
        \node[state,accepting] (n4) at (0, 2) {$q_4$};
        \node[state] (n5) at (2, 2) {$q_5$};
        
        \node[conditions,below of=n1] {$\{ \project{2}{{\Label_1}^3} \}$};
        \node[conditions,below of=n2] {$\{ \}$};
        \node[conditions,below of=n3] {$\{ \}$};
        \node[conditions,above of=n4] {$\{ \coproject{2}{{\Label_1}^3} \}$};

        \begin{scope}[on background layer]
            \path[->] (s1) edge (n1)
                      (s4) edge (n4)
                      (n1) edge node[label] {$\Label_1$} (n2)
                      (n4) edge node[label] {$\Label_1, \Label_2$} (n5)
                      (n1) edge node[label,above left] {$\Label_2$} (n5)
                      (n2) edge[loop] node[label] {$\Label_2$} (n2)
                      (n2) edge node[label] {$\Label_1$} (n3)
                      (n3) edge node[label,above right] {$\Label_1, \Label_2$} (n5)
                      (n5) edge[loop] node[label] {$\Label_1,\Label_2$} (n5);
                      
        \end{scope}
    \end{tikzpicture}
    \caption{An example of a deterministic condition automaton.}\label{fig:example_dautomaton}
\end{figure}
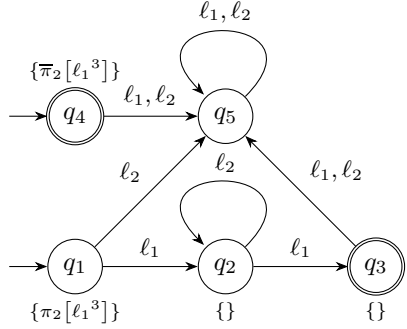
\end{example}

In the construction of deterministic condition automata we shall use the \emph{condition-complement} of a condition $\expr$, denoted by $\condc{\expr}$, and defined as follows:
$$\condc{\expr} =
        \begin{cases} \emptyset & \text{if $\expr = \identity$};\\
                      \identity & \text{if $\expr = \emptyset$};\\
                      \coproject{1}{\expr'} & \text{if $\expr = \project{1}{\expr'}$};\\
                      \coproject{2}{\expr'} & \text{if $\expr = \project{2}{\expr'}$};\\
                      \project{1}{\expr'} & \text{if $\expr = \coproject{1}{\expr'}$};\\
                      \project{2}{\expr'} & \text{if $\expr = \coproject{2}{\expr'}$}.\\
        \end{cases}$$

Observe that the condition complement of a projection expression is a coprojection expression, and vice-versa.  If $S$ is a set of conditions, then we use the notation $\condc{S}$ to denote the set $\{ \condc{\cond} \mid \cond \in S \}$.

\begin{lemma}
Let $\Fragment \in \{ \transop, \proop, \coproop \}$ and let $\Automaton$ be an $\Fragment$-free condition automaton. There exists a deterministic condition automaton $\Automaton_D$ that is path-equivalent to $\Automaton$ with respect to labeled trees. The condition automaton $\Automaton_D$ is $\{ \transop \}$-free if $\transop \notin \Fragment$ and $\{ \proop, \coproop \}$-free if $\proop, \coproop \notin \Fragment$.
\end{lemma}
\begin{proof}
Let $\Automaton = (\States, \Alphabet, \Conditions, \Initials, \Finals, \Transitions, \StateConditions)$ be a condition automaton. By Lemma~\ref{lem:identity_free}, we assume that $\Automaton$ is $\identity$-transition free. We construct $\Automaton_D = (\States_D, \Alphabet, \Conditions_D, \Initials_D, \Finals_D,\allowbreak \Transitions_D, \StateConditions_D)$, where $\States_D$, $\Initials_D$, and $\Transitions_D$ are constructed by Algorithm~\ref{alg:ctatodcta}, and \begin{align*}
        \Conditions_D     &= \Conditions \union \condc{\Conditions};\\
        \Finals_D         &= \{ (Q, V) \mid (Q, V) \in \States_D \land Q \intersect \Finals \neq \emptyset \};\\
        \StateConditions_D&= \{ ((Q, V), c) \mid (Q, V) \in \States_D \land (c \in V \lor c \in \condc{\Conditions \difference V})\}.
\end{align*}

\begin{algorithm}[htb!]
\caption{Translation to deterministic condition automaton}\label{alg:ctatodcta}
\begin{algorithmic}[1]
\STATE Let $\States_D$, $\Initials_D$, and $\VAR{new}$ be empty sets of states
\STATE Let $\Transitions_D$ be an empty transition relation
\FOR{$V \subseteq \Conditions$}
    \STATE $Q \GETS \{ q \mid q \in \Initials \land \StateConditions(q) \subseteq V \}$
    \STATE Add new state $(Q, V)$ to $\States_D$, $\Initials_D$, and $\VAR{new}$
\ENDFOR
\WHILE{$\VAR{new} \neq \emptyset$}
    \STATE Take and remove state $(Q, V)$ from $\VAR{new}$
    \FOR{$\Label \in \Alphabet$}
        \STATE $P \GETS \{ p \mid \exists q\ q \in Q \land (q, \Label, p) \in \Transitions \}$
        \FOR{$W \subseteq \Conditions$}
            \STATE $P' \GETS \{ p \mid p \in P \land \StateConditions(p) \subseteq W \}$
            \IF{$(P', W) \notin \States_D$}
                \STATE Add new state $(P', W)$ to $\States_D$ and $\VAR{new}$
            \ENDIF
            \STATE Add new transition $((Q, V), \Label, (P', W))$ to $\Transitions_D$
        \ENDFOR
    \ENDFOR
\ENDWHILE
\end{algorithmic}
\end{algorithm}

\newcommand{\EC}[1]{\zeta(#1)}

Let $\Tree = (\Nodes, \Alphabet, \ELabels)$ be a tree and let $m, n \in \Nodes$ be nodes. If $n' \in \Nodes$, then $\EC{n'}$ denotes the set $\{ \cond \mid \cond \in \Conditions \land (n', n') \in \Apply{\cond}{\Tree} \}$. Both determinism of $\Automaton_D$ and path-equivalence of $\Automaton_D$ and $\Automaton$ are guaranteed, as this construction satisfies the following properties:

\begin{enumerate}[ref=(\arabic*)]
\item \label{enum:const_d:unique_satisfy} \emph{There exists exactly one $V \subseteq \Conditions$ with $(n, n) \in \Apply{\cexpr{V \union \condc{\Conditions \difference V}}}{\Tree}$, and we have $V = \EC{n}$.}

By definition, $\EC{n}$ is the set of all conditions satisfied by $n$, hence, we can choose $V = \EC{n}$ and we have $(n, n) \in \Apply{\cexpr{V \union \condc{\Conditions \difference V}}}{\Tree}$. To show that no other choice for $V$ is possible, we consider any $V' \subseteq \Conditions$ with $V \neq V'$. We show that $(n, n) \notin \Apply{\cexpr{V' \union \condc{\Conditions \difference V'}}}{\Tree}$. As $V \neq V'$, there exists $\cond \in \Conditions$ such that  $\cond \in V \difference V'$ or $\cond \in V' \difference V$:

\begin{enumerate}
\item $\cond \in V \difference V'$. By $\cond \in V$ and $(n, n) \in \Apply{\cexpr{V \union \condc{\Conditions \difference V}}}{\Tree}$, we have $(n, n) \in \Apply{\cond}{\Tree}$ and we have $(n, n) \notin \Apply{\condc{\cond}}{\Tree}$. As $\cond \notin V'$, we have $\condc{\cond} \in V' \union \condc{\Conditions \difference V'}$. Hence, we conclude $(n, n) \notin \Apply{\cexpr{V' \union \condc{\Conditions \difference V'}}}{\Tree}$.
\item $\cond \in V' \difference V$. Since $\cond \notin V$, we have $(n, n) \notin \Apply{\cond}{\Tree}$. As $\cond \in V'$, we have $\cond \in V' \union \condc{\Conditions \difference V'}$. Hence, we conclude $(n, n) \notin \Apply{\cexpr{V' \union \condc{\Conditions \difference V'}}}{\Tree}$.
\end{enumerate}

\item \emph{There exists exactly one state $(P, V) \in \Initials_D$ such that $m$ satisfies $(P, V)$.}\label{enum:const_d:unique_initial}

By Property~\ref{enum:const_d:unique_satisfy}, we have $V = \EC{m}$. By the construction of $\States_D$, there exists exactly one set of states $Q \subseteq \Initials$ such that $(Q, V) \in \Initials_D$ and $\StateConditions((Q, V)) = V \union \condc{\Conditions \difference V}$.

\item \emph{Let $(P, V) \in \States_D$ be a state such that $m$ satisfies $(P, V)$. If there exists an edge label $\Label$ such that $(m, n) \in \Label$, then there exists exactly one transition $((P, V), \Label, (Q, W)) \in \Transitions$ such that $n$ satisfies $(Q, W)$.}\label{enum:const_d:unique_edge}

By Property~\ref{enum:const_d:unique_satisfy}, we have $W = \EC{n}$. By the construction of $\States_D$ and $\Transitions_D$, there is exactly one set of states $P' \subseteq \States$ such that $(P', W) \in \States_D$ and $((Q, V), \Label, (P', W)) \in \Transitions_D$. By the choice of $W$, $m$ must satisfy $(P', W)$.

\item \emph{Let $(P, V) \in \States_D$ be a state such that $m$ satisfies $(P, V)$. If there exists a directed path from $m$ to $n$, then there exists exactly one run $((P, V), m)\dots\allowbreak((Q, W), n)$ of $\Automaton_D$ on $\Tree$.}\label{enum:const_d:unique_path}

Repeated application of Property~\ref{enum:const_d:unique_edge}.

\item \emph{If $(p, n)\Label(p', n')$ is a run of $\Automaton$ on $\Tree$ then, for each $(P, V)$ with $p \in P$, there exists exactly one transition $((P, V), \Label, (P', V'))$ such that $n'$ satisfies $(P', V')$. For this transition, we have $p' \in P'$.}\label{enum:const_d:d_to_n_edge}

By Property~\ref{enum:const_d:unique_satisfy}, we have $V = \EC{n}$ and $V' = \EC{n'}$. By the construction of $\Transitions_D$, there is exactly one set of states $P' \subseteq \States$ such that $((P, V), \Label, (P', V')) \in \Transitions_D$. Observe that we must have $\StateConditions(p') \subseteq V'$, hence $p' \in P'$.

\item \emph{If there exists a run $(q_1, m) \dots (q_i, n)$ of $\Automaton$ on $\Tree$ with $q_1 \in \Initials$, then there exists a run $((Q_1, V_1), m) \dots ((Q_i, V_i), n)$ of $\Automaton_D$ on $\Tree$ with $(Q_1, V_1) \in \Initials_D$, and, for all $j$, $1 \leq j \leq i$, $q_j \in Q_j$.}\label{enum:const_d:d_to_n}

By induction. The base case is provided by Property~\ref{enum:const_d:unique_initial}, and the induction steps are provided by repeated applications of Property~\ref{enum:const_d:d_to_n_edge}.

\item \emph{If there exists a run $((P, V), m) \Label ((Q, W), n)$ of $\Automaton_D$ on $\Tree$ with $Q \neq \emptyset$, then, for all $q \in Q$, there exists a run $(p, m) \Label (q, n)$, with $p \in P$, of $\Automaton$ on $\Tree$.}\label{enum:const_d:n_to_d_edge}

By the construction of $\Transitions_D$, there must be $p \in P$ such that $(p, \Label, q) \in \Transitions$. By $p \in P$ and $q \in Q$ and by the definition of $\States_D$ and $\StateConditions_D$, we have $\StateConditions(p) \subseteq V \subseteq \StateConditions_D(P)$ and $\StateConditions(q) \subseteq W \subseteq \StateConditions_D(Q)$. Hence $m$ satisfies $p$ and $n$ satisfies $q$. Thus $(p, m) \Label (q, n)$ is a run of $\Automaton$ on $\Tree$.

\item \emph{If there exists a run $((Q_1, V_1), m) \dots ((Q_i, V_i), n)$ of $\Automaton_D$ on $\Tree$ with $Q_i \neq \emptyset$, then, for all $q_i \in Q_i$, there exists a run $(q_1, m) \dots (q_i, n)$ of $\Automaton$ on $\Tree$ with, for all $j$, $1 \leq j < i$, $q_j \in Q_j$.}\label{enum:const_d:n_to_d}

By induction. The base cases involve runs of the form $((Q_i, V_i), n)$ of $\Automaton_D$ on $\Tree$ with $Q_i \neq \emptyset$. We can choose any $q \in Q_i$ and, by the definition of $\States_D$ and $\StateConditions_D$, we have $\StateConditions(q) \subseteq V_i \subseteq \StateConditions_D((Q_i, V_i))$. Hence, $n$ satisfies $q_i$ and we conclude that $(q_i, n)$ is a run of $\Automaton$ on $\Tree$. The induction steps are provided by repeated applications of Property~\ref{enum:const_d:n_to_d_edge}.
\end{enumerate}

By~\ref{enum:const_d:unique_initial} and~\ref{enum:const_d:unique_path} we conclude that $\Automaton_D$ is a deterministic condition automaton. By~\ref{enum:const_d:d_to_n} and the construction of $\Initials_D$ and $\Finals_D$, $\Apply{\Automaton}{\Tree} \subseteq \Apply{\Automaton_D}{\Tree}$. By~\ref{enum:const_d:n_to_d} and the construction of $\Initials_D$ and $\Finals_D$, $\Apply{\Automaton_D}{\Tree} \subseteq \Apply{\Automaton}{\Tree}$. Hence, we conclude that $\Automaton$ and $\Automaton_D$ are path-equivalent. The construction of $\Conditions$ did not add any usage of $\transop$, and introduced $\coproop$ only when $\proop$ was present. Hence, the condition automata $\Automaton_D$ is $\{ \transop \}$-free if $\transop \notin \Fragment$ and $\{ \proop, \coproop \}$-free if $\proop, \coproop \notin \Fragment$.
\end{proof}

Based on the construction of the complement of a finite automata and the construction of deterministic condition automata, we can construct the downward complement of a condition automaton:

\begin{proposition}\label{prop:closed_dcompl}
Let $\Fragment \in \{ \transop, \proop, \coproop \}$ and let $\Automaton$ be an $\Fragment$-free condition automaton. There exists a condition automaton $\Automaton'$ such that, for every tree $\Tree$, we have $\Apply{\Automaton'}{\Tree} = \dcompl{\Apply{\Automaton}{\Tree}}$. The condition automaton $\Automaton'$ is $\{ \transop \}$-free if $\transop \notin \Fragment$ and $\{ \proop, \coproop \}$-free if $\proop, \coproop \notin \Fragment$.
\end{proposition}
\begin{proof}
Let $\Automaton_D = (\States_D, \Alphabet_D, \Conditions_D, \Initials_D, \Finals_D, \Transitions_D, \StateConditions_D)$ be the deterministic condition automaton equivalent to $\Automaton$. We construct $\Automaton' = (\States_D, \Alphabet_D, \Conditions_D, \Initials_D, \States_D \difference \Finals_D, \allowbreak\Transitions_D, \StateConditions_D)$. Besides those changes made by constructing a deterministic condition automaton, we did not add new condition expressions to the set of condition expressions in the proposed constructions. Hence, the condition automaton $\Automaton'$ is $\{ \transop \}$-free if $\transop \notin \Fragment$ and $\{ \proop, \coproop \}$-free if $\proop, \coproop \notin \Fragment$.
\end{proof}

We are now ready to conclude the following:

\begin{corollary}\label{cor:closed_difference}
Let $\Automaton_1$ and $\Automaton_2$ be condition automata. There exists a condition automaton $\Automaton_\difference$ such that, for every tree $\Tree$, we have $\Apply{\Automaton_\difference}{\Tree} = \Apply{\Automaton_1}{\Tree} \difference \Apply{\Automaton_2}{\Tree}$. The condition automaton $\Automaton_\difference$ is $\{ \transop \}$-free if $\transop \notin \Fragment$, $\{ \proop, \coproop \}$-free if $\proop, \coproop \notin \Fragment$, and acyclic whenever $\Automaton_1$ is acyclic.
\end{corollary}
\begin{proof}
Since $\Apply{\Automaton_1}{\Tree} \difference \Apply{\Automaton_2}{\Tree} = \Apply{\Automaton_1}{\Tree} \intersect \dcompl{\Apply{\Automaton_2}{\Tree}}$, we can apply Proposition~\ref{prop:closed_dcompl} and Proposition~\ref{prop:closed_intersect} to construct $\Automaton_\difference$.
\end{proof}

Proposition~\ref{prop:closed_intersect} and Corollary~\ref{cor:closed_difference} only remove intersection and difference at the highest level: these results ignore the expressions inside conditions. To fully remove intersection and difference, we use a bottom-up construction:

\begin{theorem}\label{thm:automaton_collapse}
Let $\Fragment \subseteq \{ \transop, \proop, \coproop, \intersect, \difference \}$. On labeled trees we have $\Lang(\Fragment) \path\LeqExpr \Lang(\BASE{\Fragment} \difference \{ \intersect, \difference \})$.
\end{theorem}
\begin{proof}
Given an expression $\expr$ in $\Lang(\Fragment)$, we construct the path-equivalent expression in $\Lang(\BASE{\Fragment} \difference \{ \intersect, \difference \})$ in a bottom-up fashion by constructing a condition automaton $\Automaton$  that is path-equivalent to $\expr$ and is appropriate, according to Table~\ref{tbl:expr_cta}, for the class $\Lang(\BASE{\Fragment} \difference \{ \intersect, \difference \})$. Using Proposition~\ref{prop:cta_to_ne}, the constructed condition automaton  $\Automaton$  can be translated to an expression in $\Lang(\BASE{\Fragment} \difference \{ \intersect, \difference \})$.

The base cases are expressions of the form $\emptyset$, $\identity$, and $\Label$ (for $\Label$ an edge label), for which we directly construct condition automata using Proposition~\ref{prop:ne_to_cta}. We use Proposition~\ref{prop:cta_close_basics} to deal with the operations $\compose$, $\union$, and $\transop$. We deal with expressions of the form $f(\expr)$, $f \in \{\proop[1], \proop[2], \coproop[1], \coproop[2] \}$ by translating the condition automaton path-equivalent to $\expr$ to an expression $\expr'$ (which is in $\Lang(\BASE{\Fragment} \difference \{ \intersect, \difference \})$) and then use Proposition~\ref{prop:ne_to_cta} to construct the condition automaton path-equivalent to $f(\expr)$. Finally, we use Proposition~\ref{prop:closed_intersect} and Corollary~\ref{cor:closed_difference} to deal with the operators $\intersect$ and $\difference$.
\end{proof}

Observe that Theorem~\ref{thm:automaton_collapse} does not strictly depend on the graph being a tree: indirectly, Theorem~\ref{thm:automaton_collapse} depends on Lemma~\ref{lem:run_identity_free}, which holds for all graphs in which each pair of nodes is connected by at most one directed path. Hence, the results of Theorem~\ref{thm:automaton_collapse} can be generalized to, for example, forests. In combination with Proposition~\ref{prop:unl_coproop}, we can also conclude the following:

\begin{corollary}\label{cor:prim_coproop}
Let $\Fragment_1, \Fragment_2 \subseteq \{ \transop, \proop, \intersect, \difference \}$. If $\coproop \in \BASE{\Fragment_1}$ and $\coproop \notin \BASE{\Fragment_2}$, then we have $\Lang(\Fragment_1) \bool\nLeqExpr \Lang(\Fragment_2)$ on unlabeled chains.
\end{corollary}

\subsubsection{Removing projections from expressions used on labeled chains}

The concept of condition automata to represent and manipulate navigational expressions can also be used to simplify boolean queries. Next, we use condition automata to simplify expressions evaluated on chains that use projection ($\proop$), this by providing manipulation steps that reduce the total weight of the projections in an expression:

\begin{definition}
Let $\expr$ be an expression in $\Lang( \transop, \proop)$. We define the \emph{condition-depth} of $\expr$, denoted by $\depthcond(\expr)$, as \[\depthcond(\expr) = \begin{cases}
0&\text{if $\expr \in \{ \emptyset, \identity \}$};\\
0&\text{if $\expr = \Label$, with $\Label$ an edge label};\\
\depthcond(\expr')&\text{if $\expr = \transitive{\expr'}$};\\
\depthcond(\expr') + 1&\text{if $\expr \in \{ \project{1}{\expr'}, \project{2}{\expr'} \} $};\\
\max(\depthcond(\expr_1), \depthcond(\expr_2))&\text{if $\expr \in \{ \expr_1 \compose \expr_2, \expr_1 \union \expr_2 \}$}.
\end{cases}\]
We define the \emph{condition-depth} of a $\{ \coproop \}$-free condition automaton $\Automaton = (\States, \Alphabet, \Conditions, \Initials,\allowbreak \Finals, \Transitions, \StateConditions)$, denoted by $\depthcond(\Automaton)$, as $\depthcond(\Automaton) = \max \{ \depthcond(\cond) \mid \cond \in \Conditions \}$.

We define the \emph{condition-weight} of $\Automaton$, denoted by $\weightcond(\Automaton)$, as $$\weightcond(\Automaton) = \abs{\{ \cond \mid \cond \in \Conditions \land \depthcond(\cond) = \depthcond(\Automaton) \}}.$$
\end{definition}

We can now prove the following technical lemma:

\begin{lemma}\label{lem:cta_bool_pro_decrease}
Let $\Automaton$ be a $\{ \coproop \}$-free and $\identity$-transition free condition automaton. If $\depthcond(\Automaton) > 0$, then there exists a $\{ \coproop \}$-free and $\identity$-transition free condition automaton $\Automaton_{\proop}$ such that
\begin{enumerate}
\item for every labeled chain $\Chain$, we have $\Apply{\Automaton}{\Chain} = \emptyset$ if and only if $\Apply{\Automaton_{\proop}}{\Chain} = \emptyset$; and 
\item $\depthcond(\Automaton) > \depthcond(\Automaton_{\proop})$ or $\depthcond(\Automaton) = \depthcond(\Automaton_{\proop}) \land \weightcond(\Automaton) > \weightcond(\Automaton_{\proop})$.
\end{enumerate}
The condition automaton $\Automaton_{\proop}$ is acyclic and $\{ \transop \}$-free whenever $\Automaton$ is acyclic and $\{ \transop \}$-free.
\end{lemma}
\begin{proof}
Let $\Automaton = (\States, \Alphabet, \Conditions, \Initials, \Finals, \Transitions, \StateConditions)$ be a $\{ \coproop \}$-free and $\identity$-transition free condition automaton. Choose a condition $\cond \in \Conditions$ with $\depthcond(\Automaton) = \depthcond(\cond)$. Let $\Automaton' = (\States', \Alphabet', \Conditions', \Initials', \Finals', \Transitions', \StateConditions')$  be a $\{ \coproop \}$-free and $\identity$-transition free condition automaton equivalent to $\expr'$. If we construct $\Automaton'$ in the canonical way, we have $\depthcond(\Automaton') = \depthcond(\expr')$. Let $\notstate \notin \States \union \States'$ be a fresh state. We define the power set of set $S$, denoted by $\PowerSet{S}$, as $\PowerSet{S} = \{ S' \mid S' \subseteq S \}$. In the following, we use the values $1$  and $2$ and the variable $i$, $i \in \{ 1, 2 \}$, to indicate that definitions depend on the type $i$ of the condition $\cond = \project{i}{\expr'}$. We define $\Automaton_{\proop} = (\States_{\proop}, \Alphabet_{\proop}, \Conditions_{\proop}, \Initials_{\proop}, \Finals_{\proop}, \Transitions_{\proop}, \StateConditions_{\proop})$, for $\cond = \project{i}{\expr'}$, as follows:
\begin{align}
\label{eq:states_cond}
\States_\cond               ={} & \{ q \mid \cond \in \StateConditions(q) \};
\displaybreak[0]\\
\label{eq:states_lnotcond}
\States_{\lnot\cond}        ={} & \States \difference \States_\cond;
\displaybreak[0]\\
\label{eq:states_not1}
\States_{\lnot1}            ={} & \{ (q, Q) \mid q \in \States_\cond \land Q \subseteq \States' \land Q \intersect \Initials' = \emptyset \};
\displaybreak[0]\\
\label{eq:states_not2} 
\States_{\lnot2}            ={} & \{ (q, Q) \mid q \in \States_\cond \land Q \subseteq \States' \land Q \intersect \Finals' = \emptyset \};
\displaybreak[0]\\
\label{eq:states_pro}
\States_{\proop}                ={} & \left(\States \times \PowerSet{\States'}\right) \difference \States_{\lnot i} \union \{ \notstate \} \times \left(\PowerSet{\States'} \difference \emptyset\right);
\displaybreak[0]\\
\label{eq:labels_pro}
\Alphabet_{\proop}                ={} & \Alphabet;
\displaybreak[0]\\
\label{eq:conditions_pro}
\Conditions_{\proop}            ={} & (\Conditions \difference \{ \cond \}) \union \Conditions';
\displaybreak[0]\\
\label{eq:initials_1}
\Initials_1                 ={} & \{ (q, \{ q' \}) \mid q \in \States_{\cond} \intersect \Initials \land q' \in \Initials' \};
\displaybreak[0]\\
\label{eq:initials_2}
\Initials_2                 ={} &      \{ (q, Q) \mid q \in \States_{\lnot\cond} \intersect \Initials \land  \emptyset \subset Q \subseteq \Initials' \}\notag\\
                                &\union \{ (q, Q) \mid q \in \States_{\cond} \intersect \Initials \land \emptyset \subset Q \subseteq \Initials' \intersect \Finals' \}\notag\\
                                &\union \{ (\notstate, Q) \mid \emptyset \subset Q \subseteq \Initials'  \};
\displaybreak[0]\\
\label{eq:initials_pro}
\Initials_{\proop}              ={} & \{ (q, \emptyset) \mid q \in \States_{\lnot\cond} \intersect \Initials \} \union \Initials_i;
\displaybreak[0]\\
\label{eq:finals_1}
\Finals_1                   ={} &      \{ (q, Q) \mid q \in \States_{\lnot\cond} \intersect \Finals \land  \emptyset \subset Q \subseteq \Finals' \}\notag\\
                                &\union \{ (q, Q) \mid q \in \States_{\cond} \intersect \Finals \land  \emptyset \subset Q \subseteq \Finals' \intersect \Initials' \}\notag\\
                                &\union \{ (\notstate, Q) \mid  \emptyset \subset Q \subseteq \Finals' \};
\displaybreak[0]\\
\label{eq:finals_2}
\Finals_2                   ={} & \{ (q, \{ q' \}) \mid q \in \States_{\cond} \intersect \Finals \land q' \in \Finals' \};
\displaybreak[0]\\
\label{eq:finals_pro}
\Finals_{\proop}                ={} & \{ (q, \emptyset) \mid q \in \States_{\lnot\cond} \intersect \Finals \} \union \Finals_i;
\displaybreak[0]\\
\label{eq:trans_power_set}
\Transitions_{\PowerSet{\States'}}
                            ={} & \{ (P, \Label, Q) \mid P \subseteq \States' \land \Label \in \Alphabet_{\proop} \land Q \subseteq \States' \land{}\notag\\
            &\hphantom{{}=\qquad}(\forall p\ p \notin P \lor (\exists q\ q \in Q \land (p, \Label, q) \in \Transitions')) \land{}\notag\\
            &\hphantom{{}=\qquad}(\forall q\ q \notin Q \lor (\exists p\ p \in P \land (p, \Label, q) \in \Transitions')) \};
\displaybreak[0]\\
\label{eq:trans_1b}
\Transitions_{1,\text{b}} ={} & \{ ((p, P \union P'), \Label, (q, Q)) \mid{}\notag\\
        &\hphantom{{}=\qquad}(p, P \union P') \in \States_{\proop} \land P' \subseteq \Finals' \land (q, Q) \in \States_{\proop} \land{}\notag\\
        &\hphantom{{}=\qquad}((p, \Label, q) \in \Transitions  \lor ((p = \notstate \lor p \in \Finals) \land q = \notstate)) \land{}\notag\\
        &\hphantom{{}=\qquad} (P, \Label, Q) \in \Transitions_{\PowerSet{\States'}} \};
\displaybreak[0]\\
\label{eq:trans_2b}
\Transitions_{2,\text{b}} ={} &  \{ ((p, P), \Label, (q, Q \union Q')) \mid{}\notag\\
            &\hphantom{{}=\qquad} (p, P) \in \States_{\proop} \land Q' \subseteq \Initials' \land (q, Q \union Q') \in \States_{\proop} \land{}\notag\\
            &\hphantom{{}=\qquad}((p, \Label, q) \in \Transitions  \lor (p = \notstate \land (q = \notstate \lor q \in \Initials))) \land{}\notag\\
            &\hphantom{{}=\qquad}(P, \Label, Q) \in \Transitions_{\PowerSet{\States'}} \};
\displaybreak[0]\\
\label{eq:trans_1cond}
\Transitions_{1,\cond} ={} & \{ ((p, P \union P'), \Label, (q, Q \union \{ q' \}))\mid{}\notag\\
            &\hphantom{{}=\qquad}(p, P \union P') \in \States_{\proop} \land (q, Q \union \{ q' \}) \in \States_{\proop} \land{}\notag\\
            &\hphantom{{}=\qquad} (p, \Label, q) \in \Transitions \land (P, \Label, Q) \in \Transitions_{\PowerSet{\States'}} \land{}\notag\\
            &\hphantom{{}=\qquad}q \in \States_\cond \land q' \in \Initials'  \land P' \subseteq \Finals' \};
\displaybreak[0]\\
\label{eq:trans_2cond}
\Transitions_{2,\cond} ={} & \{ ((p, P \union \{ p' \}), \Label, (q, Q \union Q'))\mid{}\notag\\
            &\hphantom{{}=\qquad} (p, P \union \{ p' \}) \in \States_{\proop} \land (q, Q \union Q') \in \States_{\proop} \land{}\notag\\
            &\hphantom{{}=\qquad}(p, \Label, q) \in \Transitions \land (P, \Label, Q) \in \Transitions_{\PowerSet{\States'}}\land{}\notag\\
            &\hphantom{{}=\qquad} p \in \States_\cond \land p' \in \Finals'  \land  Q' \subseteq \Initials \};
\displaybreak[0]\\
\label{eq:trans_pro}
\Transitions_{\proop}           ={} & \Transitions_{i,\text{b}} \union \Transitions_{i,\cond};
\displaybreak[0]\\
\label{eq:stateconditions_pro}
\StateConditions_{\proop}       ={} & \{ ((q, Q), \cond') \mid (q, Q) \in \States_{\proop} \land{}\notag\\
    &\hphantom{{}=\qquad} (\cond'  \in \StateConditions(q) \lor (\exists q'\ q' \in Q \land \cond' \in \StateConditions'(q'))) \}.
\end{align}

We shall prove that $\Automaton_{\proop}$ satisfies the necessary properties.
\begin{enumerate}

\item \emph{We have $d > \depthcond(\Automaton_{\proop})$ or $d = \depthcond(\Automaton_{\proop}) \land \weightcond(\Automaton) > \weightcond(\Automaton_{\proop})$.}

Observe that $\depthcond(\Automaton') < \depthcond(\Automaton)$. Hence, the property follows directly from~\eqref{eq:conditions_pro}, the definition of $\depthcond(\DOT)$, and the definition of $\weightcond(\DOT)$.

\item \emph{Let $\Chain = (\Nodes, \Alphabet, \ELabels)$ be a labeled chain and let $\cond = \project{1}{\expr'}$. If $(m, n) \in \Apply{\Automaton}{\Chain}$, then there exists $v \in \Nodes$ such that $(m, v) \in \Apply{\Automaton_{\proop}}{\Chain}$.}\label{proof:cta_ctapro_1}

\item \emph{Let $\Chain = (\Nodes, \Alphabet, \ELabels)$ be a labeled chain and let $\cond = \project{2}{\expr'}$. If $(m, n) \in \Apply{\Automaton}{\Chain}$, then there exists $v \in \Nodes$ such that $(v, n) \in \Apply{\Automaton_{\proop}}{\Chain}$.}\label{proof:cta_ctapro_2}

We only prove Property~\ref{proof:cta_ctapro_1}, Property~\ref{proof:cta_ctapro_2} is similar. We show that a single run in $\Automaton$ is simulated by a single run in $\Automaton_{\proop}$ that, at the same time, also simulates the runs for $\Automaton'$ starting at each state $q \in \States$ with $\cond \in \StateConditions(q)$.

Let $(q_1, n_1)\dots(q_i, n_i)$  be an $\identity$-transition free run with $q_1 \in \Initials$ and $q_i \in \Finals$ proving $(n_1, n_i) \in \Apply{\Automaton}{\Chain}$. Now consider a state $q_j$, with $1 \leq j \leq i$, such that $\cond \in \StateConditions(q_j)$. Observe that, by~\eqref{eq:states_cond}, we have $\cond \in \StateConditions(q_j)$ if and only if $q_j \in \States_\cond$. As $n_j$ satisfies $q_j$, we must have $(n_j, n_j) \in \Apply{\cond}{\Chain}$. Hence, by the semantics of $\project{1}{\DOT}$, there must exist an $\identity$-transition free run $(p_{j}, n_j)\dots(p_j', m_j)$ of $\Automaton'$ on $\Chain$ with $p_j \in \Initials'$ and $p_j' \in \Finals'$ proving that a node $m_j$ exists such that $(n_j, m_j) \in \Apply{\Automaton'}{\Chain}$.

For each state $q_j$ with $q_j \in \States_\cond$ we choose such a run $(p_{j}, n_j)\dots(p_j', m_j)$ of $\Automaton'$ on $\Chain$ with $p_j \in \Initials'$ and $p_j' \in \Finals$. Let $d$ be the maximum distance between, on the one hand, node $n_1$, and, on the other hand, node $n_i$ and the nodes $m_j$ in these runs. For each  $1 \leq k \leq d$, we define \begin{align*}R_k = \{ s \mid{}& \text{$(s, n_k)$ is part of a run $(p_{j}, n_j)\dots(p_j', m_j)$}\\&\text{with $1 \leq j \leq i$ and $q_j \in \States_\cond$} \},\end{align*} we define $r_k = q_k$ if $k \leq i$ and $r_k = \notstate$ if $k > i$. Let $n_1\Label_1\dots n_d$ be the directed path from node $n_1$ to the node at distance $d$ in chain $\Chain$. We prove that $((r_1, R_1), n_1)\Label_1\dots((r_d, R_d), n_d)$ is a run of $\Automaton_{\proop}$ on $\Chain$ with $(r_1, R_1) \in \Initials_{\proop}$ and $(r_d, R_d) \in \Finals_{\proop}$:

\begin{enumerate}
\item \emph{For $k$, $1 \leq k \leq d$, we have $(r_k, R_k) \in \States_{\proop}$.}

If $1 \leq k \leq i$, then $(r_k, R_k) \in \States \times \PowerSet{\States'}$. If $r_k \in \States_\cond$, we have $(p_k, n_k) \in \Initials'$ and, hence, $p_k \in R_k$. By~\eqref{eq:states_not1}, we have $(r_k, R_k) \notin \States_{\lnot1}$, and, hence, $(r_k, R_k) \in \States_{\proop}$. For $k > i$, we have $(r_k, R_k) \in \{ \notstate \} \times \PowerSet{\States'}$. By the definition of $R_k$, we must have $R_k \neq \emptyset$. Hence, we conclude, $(r_k, R_k) \in \States_{\proop}$.

\item \emph{We have $(r_1, R_1) \in \Initials_{\proop}$.}

By construction, we have $r_1 = q_1$ and $q_1 \in \Initials$. If $r_1 \notin \States_\cond$, then $R_1 = \emptyset$ and, by~\eqref{eq:initials_pro}, $(r_1, R_1) \in \Initials_{\proop}$. If $r_1 \in \States_\cond$, then $R_1 = \{ p_1 \}$, with $p_1 \in \Initials'$, and, by~\eqref{eq:initials_1}, $(r_1, R_1) \in \Initials_{\proop}$. 

\item \emph{For $k$, $1 \leq k \leq d$, $n_k$ satisfies $(r_k, R_k)$.}

If $1 \leq k \leq i$, then $r_k = q_k$ and $n_k$ satisfies $q_k$, hence, we also have $(n_k, n_k) \in \Apply{\cexpr{\StateConditions(q_k) \setminus \{ \cond \}}}{\Chain}$. By construction of $R_k$, we have, for each $s \in R_k$, $n_k$ satisfies $s$. Hence, by~\eqref{eq:stateconditions_pro}, $n_k$ satisfies $(r_k, R_k)$.

\item \emph{For all $1 \leq k < d$, $((r_k, R_k), \Label_k, (r_{k+1}, R_{k+1})) \in \Transitions_{\proop}$.}

Construct sets $P$ and $Q$ in the following way:
 \begin{align*}
P &= \{ p \mid p \in R_k \land (\exists q\ q \in R_{k+1} \land (p, \Label_k, q) \in \Transitions') \}\\
Q &= \{ q \mid q \in R_{k+1} \land (\exists p\ p \in R_k \land (p, \Label_k, q) \in \Transitions') \}
\end{align*}
Let $P' = R_k \difference P$ and $Q' = R_{k+1} \difference Q$. Conceptually, $P$ contains all states of runs of $\Automaton_{\proop}$ with a successor state in $Q$ (with respect to $\Transitions_{\proop}$), and, likewise, all states of $Q$ have a predecessor state in $P$. As such, $P'$ contains all states from $R_k$ for which no successor state is in $R_{k+1}$. By~\eqref{eq:trans_1b} and~\eqref{eq:trans_1cond}, such transition is only allowed if all these states are final states. Set $Q'$ contains all states from $R_{k+1}$ for which no predecessor state is in $R_k$. Again, by~\eqref{eq:trans_1b} and \eqref{eq:trans_1cond}, such transition is only allowed if $Q'$ contains at most a single state, which must be an initial state. In the following, we prove that these restrictions on $P'$ and $Q'$ hold. Let $s \in P'$ be a state. By the construction of $P$, there does not exists $s' \in R_{k+1}$ such that $(s, \Label_k, s') \in \Transitions'$. Hence, $s$ can only be a state used at the end of a run $(p_{j}, n_j)\dots(p_j', m_j)$ with $1 \leq j \leq i$ and $q_j \in \States_\cond$, and, we must have $s = p_j'$, for some $1 \leq j \leq i$. We conclude that $P' \subseteq \Finals'$. Let $s_1, s_2 \in Q'$ be states. By the construction of $Q$, there does not exist a state $s' \in R_{k}$ such that $(s', \Label_k, s_1) \in \Transitions'$ or $(s', \Label_k, s_2) \in \Transitions'$. Hence, $s_1$ and $s_2$ can only be states used at the begin of a run $(p_{j}, n_j)\dots(p_j', m_j)$ with $1 \leq j \leq i$ and $q_j \in \States_\cond$, and, hence, we must have $s_1 = p_{k+1}$ and $s_2 = p_{k+1}$ with $1 \leq k < i$. We conclude that $\abs{Q'}$ contains at most a single state, which must be an initial state.

If $r_{k+1} = \notstate$, then, by construction of $R_{k+1}$, we have, for each $s \in R_{k+1}$ and every $1\leq j\leq i$ with $q_j \in \States_\cond$, $(s, n_{k+1}) \neq (p_j, n_j)$, Hence, we conclude that $s$ is not at the begin of any run $(p_{j}, n_j)\dots(p_j', m_j)$. As such there exists an $s'$ and a run $(p_{j}, n_j)\dots(p_j', m_j)$ with $q_j \in \States_\cond$, containing $(s', n_k)\Label_{k}(s, n_{k+1})$. Hence, we have $s' \in P$, $s \notin Q'$, and $Q' = \emptyset$. If $r_k \in \States_\cond$, then, by construction, we have $p_k \in R_k$ and $p_k \in \Initials'$. If $1 \leq k < i$, then $r_k = q_k$, $r_{k+1} = q_{k+1}$, and $(q_k, \Label, q_{k+1}) \in \Transitions$. We use~\eqref{eq:trans_1b} when  $P' = \emptyset$ and $Q' = \emptyset$, or when $P' \neq \emptyset$ and $Q' = \emptyset$, and we use~\eqref{eq:trans_1cond} when $P' = \emptyset$ and $Q' \neq \emptyset$, or when $P' \neq \emptyset$ and $Q' \neq \emptyset$ to conclude $((r_k, R_k), \Label_k, (r_{k+1}, R_{k+1})) \in \Transitions_{\proop}$. If $k = i$, then $r_k = q_k$, $r_{k+1} = \notstate$, and $q_k \in \Finals$. We have $Q' = \emptyset$ and we use~\eqref{eq:trans_1b} to conclude $((r_k, R_k), \Label_k, (r_{k+1}, R_{k+1})) \in \Transitions_{\proop}$. If $k > i$, then $r_k = \notstate = r_{k+1}$. We have $Q' = \emptyset$ and we use~\eqref{eq:trans_1b} to conclude $((r_k, R_k), \Label_k, (r_{k+1}, R_{k+1})) \in \Transitions_{\proop}$.

\item \emph{We have $(r_d, R_d) \in \Finals_{\proop}$.}

If $r_d \neq \notstate$ then $d = i$ and $q_i \in \Finals$. If $R_d = \emptyset$ then, by~\eqref{eq:finals_pro}, we have $(r_d, R_d) \in \Finals_{\proop}$. If $R_d \neq \emptyset$, then, as $n_i = n_d$ is the node with maximum node distance $d$ to $n_1$ used in any of the runs $(p_{j}, n_j)\dots(p_j', m_j)$ with $1 \leq j \leq i$ and $q_j \in \States_\cond$, we must also have, for each $s \in R_d$, $s \in \Finals'$. If $q_i \in \States_\cond$, then, by construction, $p_i \in R_d$ with $p_i \in \Initials'$. By~\eqref{eq:finals_1} we conclude that, in all these cases, we have $(r_d, R_d) \in \Finals_{\proop}$.
\end{enumerate}

We conclude that $((r_1, R_1), n_1)\Label_1\dots((r_d, R_d), n_d)$ is a run of $\Automaton_{\proop}$ on $\Chain$ with $(r_1, R_1) \in \Initials_{\proop}$ and $(r_d, R_d) \in \Finals_{\proop}$, and, hence $(n_1, n_d) \in \Apply{\Automaton_{\proop}}{\Chain}$.

\item \emph{Let $\Chain = (\Nodes, \Alphabet, \ELabels)$ be a labeled chain and let $\cond = \project{1}{\expr'}$. If $(m, n) \in \Apply{\Automaton_{\proop}}{\Chain}$, then there exists a $v \in \Nodes$ such that $(m, v) \in \Apply{\Automaton}{\Chain}$.}\label{proof:ctapro_cta_1}

\item \emph{Let $\Chain = (\Nodes, \Alphabet, \ELabels)$ be a labeled chain and let $\cond = \project{2}{\expr'}$. If $(m, n) \in \Apply{\Automaton_{\proop}}{\Chain}$, then there exists a $v \in \Nodes$ such that $(v, n) \in \Apply{\Automaton}{\Chain}$.}\label{proof:ctapro_cta_2}

We only prove Property~\ref{proof:ctapro_cta_1}, Property~\ref{proof:ctapro_cta_2} is similar. We show that a single run in $\Automaton_{\proop}$ simulates a single run in $\Automaton$ and, at the same time, also simulates the runs for $\Automaton'$ starting at each state $q \in \States$ with $\cond \in \StateConditions(q)$.

Let $((q_1, Q_1), n_1)\Label_1\dots((q_i, Q_i), n_i)$ be a run of $\Automaton_{\proop}$ on $\Chain$ with $(q_1, Q_1) \in \Initials_{\proop}$ and $(q_i, Q_i) \in \Finals_{\proop}$ be the $\identity$-transition free run proving $(n_1, n_i) \in \Apply{\Chain}{\Automaton_{\proop}}$. Choose $j$ such that $1 \leq j \leq i$, $q_j \neq \notstate$, and $j = i$ or $q_{j+1} = \notstate$. We prove that $(q_1, n_1)\Label_1\dots(q_j, n_j)$ is a run of $\Automaton$ on $\Chain$ with $q_1 \in \Initials$ and $q_j \in \Finals$:
\begin{enumerate}
\item \emph{We have $q_1 \in \Initials$.}

By~\eqref{eq:initials_1} and~\eqref{eq:initials_pro}, we have $(q_1, Q_1) \in \Initials_{\proop}$ only if $q_1 \in \Initials$.
\item \emph{For all $k$, $1 \leq k \leq j$, $n_k$ satisfies $q_k$.}

We have $n_k$ satisfies $(q_k, Q_k)$. By~\eqref{eq:stateconditions_pro}, node $n_k$ satisfies all conditions in $\StateConditions(q_k) \setminus \{ \cond \}$. Hence, if $q_k \notin \States_\cond$, then $n_k$ satisfies $q_k$. If $q_k \in \States_\cond$, then, by~\eqref{eq:states_not1}, there exists a $p_1 \in Q_k$ such that $p_1 \in \Initials'$.

We prove that there are states $p_1 \in Q_k, \dots, p_d \in Q_{k+d}$ such that $(p_1, n_k)\Label_k\dots(p_d, n_{k+d})$ is a run of $\Automaton'$ on $\Chain$ with $p_1 \in \Initials'$ and $p_d \in \Finals'$. We do so by induction on the length of the run. The base case is $(p_1, n_k)$ and, as $p_1 \in Q_k$, this case is already proven. Assume we have a run $(p_1, n_k)\Label_k\dots(p_e, n_{k+e})$ of $\Automaton'$ on $\Chain$ with $1 \leq e < d$ and $p_e \notin \Finals'$. We prove that we can extend this run to a run of length $e + 1$. Observe that~\eqref{eq:trans_pro} depends on~\eqref{eq:trans_power_set}, via~\eqref{eq:trans_1b} and~\eqref{eq:trans_1cond}. If $q_{e+1} \neq \notstate$ and $p_e \notin \Finals$, then~\eqref{eq:trans_1b} or \eqref{eq:trans_1cond} applies, hence, by~\eqref{eq:trans_power_set}, there must be a $p_{e+1} \in Q_{k+e+1}$ such that $(p_e, \Label_{k+e}, p_{e+1}) \in \Transitions'$. If $q_{e+1} = \notstate$ and $p_e \notin \Finals$, then~\eqref{eq:trans_1b} applies, hence, by~\eqref{eq:trans_power_set}, there must be a $p_{e+1}$ such that $(p_e, \Label_{k+e}, p_{e+1}) \in \Transitions'$.

We observe that this construction will terminate, as the original run has a finite length $i$. Hence, at some point we encounter a state $p_d \in \Finals'$. We conclude $(n_k, n_{k+d}) \in \Apply{\Automaton'}{\Chain}$, and, by the semantics of $\project{1}{\DOT}$, we conclude $(n_k, n_k) \in \Apply{\cond}{\Chain}$. Hence, also in this case, we conclude $n_k$ satisfies all conditions in $\StateConditions(q_k)$, hence $n_k$ satisfies $q_k$.

\item \emph{For all $1 \leq k < j$, $(q_k, \Label_k, q_{k+1}) \in \Transitions$.}

We have $q_k \neq \notstate$ and $q_{k+1} \neq \notstate$, and, by the definition of a run, we have $((q_k, Q_k), \Label_k, (q_{k+1}, Q_{k+1})) \in \Transitions_{\proop}$. This transition follows from~\eqref{eq:trans_1b} or \eqref{eq:trans_1cond}. When $q_k \neq \notstate$ and $q_{k+1} \neq \notstate$, then each of~\eqref{eq:trans_1b} and~\eqref{eq:trans_1cond} guarantees that $(q_k, \Label_k, q_{k+1}) \in \Transitions$.

\item \emph{We have $q_j \in \Finals$.}

We distinguish three cases. If $i = j$ and $Q_j = \emptyset$, then, by~\eqref{eq:finals_pro}, $(q_j, Q_j) \in \Finals_{\proop}$ implies $q_j \in \Finals$. If $i = j$ and $Q_j \neq \emptyset$, then, by~\eqref{eq:finals_1},  $(q_j, Q_j) \in \Finals_{\proop}$ implies $q_j \in \Finals$ and $Q_j \subseteq \Finals'$. If $i \neq j$, then we must have $q_{j+1} = \notstate$. Hence, by~\eqref{eq:trans_1b},  $((q_j, Q_j), \Label_j, (q_{j+1}, Q_{j+1})) \in \Transitions_{\proop}$ implies $q_j \in \Finals$.
\end{enumerate}

Hence, we conclude $(n_1, n_j) \in \Apply{\Automaton}{\Chain}$.

\item \emph{$\Automaton_{\proop}$ is a $\{ \coproop \}$-free and $\identity$-transition free condition automaton. The condition automata $\Automaton_{\proop}$ is acyclic whenever $\Automaton$ is acyclic and $\{ \transop \}$-free.}

By~\eqref{eq:conditions_pro} and by~\eqref{eq:trans_pro} we immediately conclude that $\Automaton_{\proop}$ is a $\{ \coproop \}$-free and $\identity$-transition free condition automaton and $\Automaton_{\proop}$ is $\{\transop\}$-free whenever $\Automaton$ is $\{ \transop \}$-free. If $\Automaton$ is acyclic and $\{ \transop \}$-free, then we can use the proofs of Property~\ref{proof:ctapro_cta_1} or Property~\ref{proof:ctapro_cta_2} to translate every run of $\Automaton_{\proop}$ into runs of $\Automaton$ and $\Automaton'$. Using this translation, a run of $\Automaton_{\proop}$ can only be unbounded in length if runs in  $\Automaton$ or in $\Automaton'$ can be unbounded in length. Hence, $\Automaton_{\proop}$ must be acyclic whenever $\Automaton$ and $\Automaton'$ are acyclic.\qedhere
\end{enumerate}
\end{proof}

Lemma~\ref{lem:cta_bool_pro_decrease} only removes a single $\proop$-condition. To fully remove $\proop$-conditions, we repeat these removal steps until no $\proop$-conditions are left. This leads to the following result:

\begin{theorem}\label{thm:chain_proop_collapse}
Let $\Fragment \subseteq \{ \transop, \proop \}$. On labeled chains we have $\Lang(\Fragment) \bool\LeqExpr \Lang(\Fragment \difference \{ \proop \})$. 
\end{theorem}
\begin{proof}
We use Proposition~\ref{prop:ne_to_cta} to translate a navigational expression to a condition automaton $\Automaton$, then we repeatedly apply Lemma~\ref{lem:cta_bool_pro_decrease} to remove conditions, and, finally, we use Proposition~\ref{prop:cta_to_ne} to translate the resulting $\{\proop\}$-free automaton back to a navigational expression in $\Lang(\Fragment \difference \{ \proop \})$. Observe that only a finite number of condition removal steps on $\Automaton$ can be made, as Lemma~\ref{lem:cta_bool_pro_decrease} guarantees that either $\depthcond(\Automaton)$ strictly decreases or else $\depthcond(\Automaton)$ does not change and $\weightcond(\Automaton)$ strictly decreases.
\end{proof}

We observed that Theorem~\ref{thm:automaton_collapse} does not strictly depend on the graph being a tree. A similar observation holds for Theorem~\ref{thm:chain_proop_collapse}: for boolean queries, we can remove a $\proop$-condition whenever the condition checks a part of the graph that does not branch. This is the case for $\proop_2$-conditions on trees, as trees do not have branching in the direction from a node to its ancestors. For $\proop_1$, this observation does not hold, as is illustrated by the proof of Proposition~\ref{prop:proop_branching}.

\begin{proposition}\label{prop:proop_tree_partial_collapse}
Let $\Fragment \subseteq \{ \transop \}$. On labeled trees we have $\Lang(\Fragment \union \{ \proop_2 \}) \bool\LeqExpr \Lang(\Fragment)$, but $\Lang(\Fragment \union \{ \proop_1 \}) \bool\nLeqExpr \Lang(\Fragment)$.
\end{proposition}

\subsection{Results using first-order logic}\label{ss:fo}

For all $\Lang(\Fragment)$ with $\Fragment \subseteq \{ \diversity, \convop, \proop, \coproop, \intersect, \difference \}$, it is straightforward to show that every expression in $\Lang(\Fragment)$ can be expressed by a first-order logic formula over the structure $(\Nodes; \ELabels)$. Moreover, as $\Lang(\Fragment)$ is essentially a fragment of the calculus of relations, every expression can be expressed in FO[3], the language of first-order logic formulae using at most three variables~\cite{tarski,givant}. Exploring this relationship yields the following results involving $\transop$:

\begin{proposition}\label{prop:trans_trivial}
Let $\Fragment \subseteq \{ \diversity, \convop, \proop, \coproop, \intersect, \difference \}$. On unlabeled chains we have $\Lang(\transop)\path\nLeqExpr \Lang(\Fragment)$, and on labeled chains we have $\Lang(\transop) \bool\nLeqExpr \Lang(\Fragment)$. 
\end{proposition}
\begin{proof}
Using well-known results on the expressive power of first-order logic~\cite{libkin}, we conclude that no navigational expression in $\Lang(\Fragment)$ is path-equivalent to $\transitive{\Edges}$ and that no navigational expression is boolean-equivalent to $\Label_1 \compose \transitive{\Label_2 \compose \Label_2} \compose \Label_1$.
\end{proof}

\begin{proposition}\label{prop:chain_bool_tc_unl}
Let $\transop \notin \Fragment$. On unlabeled chains we have $\Lang(\Fragment \union  \{\transop \}) \bool\LeqExpr \Lang(\Fragment)$ if and only if $\coproop \notin \BASE{\Fragment}$.
\end{proposition}
\begin{proof}
For the cases with $\coproop \notin \BASE{\Fragment}$, we refer to Corollary~\ref{cor:collapse_lang}. If $\coproop \in \BASE{\Fragment}$, then we conclude that no navigational expression in $\Lang(\Fragment)$ is boolean-equivalent to $\smash{\coproject{2}{\Edges} \compose \transitive{\Edges \compose \Edges} \compose \coproject{1}{\Edges}}$, using well-known results on the expressive power of first-order logic~\cite{libkin}.
\end{proof}

\section{Related Work}\label{sec:related}
Tree query languages have been widely studied, especially in the setting of the XML data model using XPath-like query languages. For an overview, we refer to Benedikt et al.~\cite{xpath_leashed}. Due to the large body of work on querying of tree-based data models, we only point to related work that studies similar expressiveness problems.

Benedikt et al.~\cite{struct_xpath} studied the expressive power of the XPath fragments with and without the \texttt{parent} axis, with and without \texttt{ancestor} and \texttt{descendant} axes, and with and without qualifiers (which are $\proop_1$-conditions). Furthermore, they studied closure properties of these XPath fragments under intersection and complement. As such, the work by Benedikt et al.~answered similar expressiveness questions as our work does. The Core XPath fragments studied by Benedikt et al.~do, however, not include non-monotone operators such as $\coproop$ and $\difference$ and allow only for a very restricted form of transitive closure, required to define the \texttt{ancestor} and \texttt{descendant} axes. Hence, queries such as $\transitive{\Label \compose \Label}$ and $\Label_1 \compose \transitive{\Label_2 \compose \Label_2} \compose \Label_1$, used in Proposition~\ref{prop:trans_trivial}, are not expressible in these XPath fragments. 

When accounting for the difference between the node-labeled tree model used by Benedikt et al.~\cite{struct_xpath} and the edge-labeled tree model used here, and when restricting ourselves to the downward fragments as studied here, we see that all relevant XPath fragments of Benedikt et al., fragment $\mathscr{X}_{r, []}$ and its fragments,  are strictly less expressive than the navigational query language $\Lang(\transop, \proop_1)$. Furthermore,  we observe that $\mathscr{X}$ is path-equivalent to $\Lang()$ and that $\mathscr{X}_{[]}$ is path-equivalent to $\Lang(\proop_1)$. As such, our work extends some of the results of Benedikt et al.~to languages that have a more general form of transitive closure.

Conditional XPath, Regular XPath, and Regular XPath$^\approx$~\cite{condxpath,nav_xpath,xpath_tc,nav_xpath_calcalg} are studied with respect to a sibling-ordered node-labeled tree data model. The choice of a sibling-ordered tree data model makes these studies incomparable with our work: on sibling-ordered trees, Conditional XPath is equivalent to FO[3], and FO[3] is equivalent to general first-order logic~\cite{condxpath}. This result does not extend to our tree data model: on our tree data model, FO[3] cannot express simple first-order counting queries such as \begin{multline*}\exists n \exists c_1\exists c_2\exists c_3\exists c_4\ \Edges(n, c_1) \land \Edges(n, c_2) \land \Edges(n, c_3) \land \Edges(n, c_4) \land{}\\(c_1 \neq c_2) \land (c_1 \neq c_3) \land (c_1 \neq c_4) \land{}\\(c_2 \neq c_3) \land (c_2 \neq c_4) \land (c_3 \neq c_4),\end{multline*}
which is true on all trees that have a node with at least four distinct children. Although $\Lang(\transop, \proop, \coproop,  \intersect, \difference)$ is not a fragment of FO[3], due to the inclusion of the transitive closure operator, a straightforward brute-force argument shows that not even $\Lang(\transop, \proop, \coproop, \intersect, \difference)$ can express these kinds of counting queries. With an ordered \texttt{sibling} axis, as present in the sibling-ordered tree data model, the above counting query is boolean-equivalent to $\texttt{sibling}\compose \texttt{sibling}\compose\texttt{sibling}$.

Due to these differences in the tree data models used, the closure properties under intersection and complementation for Conditional XPath and Regular XPath$^\approx$ cannot readily be translated to closure properties for the navigational query languages we study. Moreover, even if the closure properties for Conditional XPath and Regular XPath$^\approx$  could be translated to our setting, then these results would only cover a single fragment.

Lastly, the XPath algebra of Gyssens et al.~\cite{xpathalgebrajr}, when restricted to the downward fragment, corresponds to the navigational query language $\Lang(\proop, \intersect, \difference)$. This work studied the expressiveness of various XPath algebra fragments with respect to a given tree, whereas we study the expressive power with respect to the class of labeled and unlabeled trees and chains. The positive algebra of Wu et al.~\cite{pospathtree}, when restricted to the downward fragment, corresponds to the navigational query language $\Lang(\proop, \intersect)$. The expressivity results in this work are dependent on the availability of a \texttt{parent}-axis (or a converse operator), and, thus, are not directly relevant for the study of the downward-only fragments.

There has been some work on the expressive power of variations of the regular path queries and nested regular path queries~\cite{rpq}, which are equivalent to fragments of the navigational query languages. Furthermore, on graphs the navigational query languages (both labeled and unlabeled) have already been studied in full detail~\cite{graph_amai,graph_navjr,graphjournal_tc,graph_icdt}. We refer to Figure~\ref{fig:graph_results} for a summary of the details of the relative expressive power of the languages we study when used to query graphs.

\begin{figure}[htb!]
    \begin{tikzpicture}[hasse_style,yscale=0.75]
        \node[dot] (n1) at (0, 0) {};
        
        \node[dot] (b1)  at (0, 0) {};
        \node[dot] (b2)  at (3, 0) {} edge[<-] (b1);
        \node[dot] (b3)  at (6, 0) {} edge[<-] (b2);
        \node[dot] (m1)  at (0, 2) {} edge[<-] (b1);
        \node[dot] (m2)  at (3, 2) {} edge[<-] (m1) edge[<-] (b2);
        \node[dot] (m3)  at (6, 2) {} edge[<-] (m2) edge[<-] (b3);
        \node[dot] (t1)  at (0, 4) {} edge[<-] (m1);
        \node[dot] (t2)  at (3, 4) {} edge[<-] (t1) edge[<-] (m2);
        \node[dot] (t3)  at (6, 4) {} edge[<-] (t2) edge[<-] (m3);
        \node[dot] (bb1) at (1, 0.5) {} edge[<-] (b1);
        \node[dot] (bb2) at (4, 0.5) {} edge[<-] (bb1) edge[<-] (b2);
        \node[dot] (bb3) at (7, 0.5) {} edge[<-] (bb2) edge[<-] (b3);
        \node[dot] (bm1) at (1, 2.5) {} edge[<-] (bb1) edge[<-] (m1);
        \node[dot] (bm2) at (4, 2.5) {} edge[<-] (bm1) edge[<-] (bb2) edge[<-] (m2);
        \node[dot] (bm3) at (7, 2.5) {} edge[<-] (bm2) edge[<-] (bb3) edge[<-] (m3);
        \node[dot] (bt1) at (1, 4.5) {} edge[<-] (bm1) edge[<-] (t1);
        \node[dot] (bt2) at (4, 4.5) {} edge[<-] (bt1) edge[<-] (bm2) edge[<-] (t2);
        \node[dot] (bt3) at (7, 4.5) {} edge[<-] (bt2) edge[<-] (bm3) edge[<-] (t3);

        \node (b1a) [below left,align=left] at (b1) {$\Lang()$};
        \node (b2a) [below left,align=left] at (b2) {$\Lang(\BASE{\intersect})$};
        \node (b3a) [below left,align=left] at (b3) {$\Lang(\BASE{\difference})$};
        
        \node (m1a) [below left,align=left] at (m1) {$\Lang(\BASE{\proop})$};
        \node (m2a) [below left,align=left] at (m2) {$\Lang(\BASE{\proop,\intersect})$};
        \node (m3a) [below left,align=left] at (m3) {$\Lang(\BASE{\proop,\difference})$};

        \node (t1a) [below left,align=left] at (t1) {$\Lang(\BASE{\coproop})$};
        \node (t2a) [below left,align=left] at (t2) {$\Lang(\BASE{\coproop,\intersect})$};
        \node (t3a) [below left,align=left] at (t3) {$\Lang(\BASE{\coproop,\difference})$};

        \node (bb1a) [above right,align=left] at (bb1) {$\Lang(\BASE{\transop})$};
        \node (bb2a) [above right,align=left] at (bb2) {$\Lang(\BASE{\transop,\intersect})$};
        \node (bb3a) [above right,align=left] at (bb3) {$\Lang(\BASE{\transop,\difference})$};
        
        \node (bm1a) [above right,align=left] at (bm1) {$\Lang(\BASE{\transop,\proop})$};
        \node (bm2a) [above right,align=left] at (bm2) {$\Lang(\BASE{\transop,\proop,\intersect})$};
        \node (bm3a) [above right,align=left] at (bm3) {$\Lang(\BASE{\transop,\proop,\difference})$};

        \node (bt1a) [above right,align=left] at (bt1) {$\Lang(\BASE{\transop,\coproop})$};
        \node (bt2a) [above right,align=left] at (bt2) {$\Lang(\BASE{\transop,\coproop,\intersect})$};
        \node (bt3a) [above right,align=left] at (bt3) {$\Lang(\BASE{\transop,\coproop,\difference})$};
    \end{tikzpicture}
    \caption[The full Hasse diagrams describing the relations between the expressive power of the various fragments of $\Lang(\transop, \coproop, \proop, \difference, \intersect)$ on graphs. This diagram describes the situation in all cases, except for boolean queries on unlabeled graphs~\cite{graph_navjr,graphjournal_tc}. For boolean queries on unbalebed graphs we have the collapse $\Lang(\transop) \bool\LeqExpr \Lang()$~\cite{graph_amai}.]{The full Hasse diagrams describing the relations between the expressive power of the various fragments of $\Lang(\transop, \coproop, \proop, \difference, \intersect)$ on graphs. This diagram describes the situation in all cases, except for boolean queries on unlabeled graphs~\cite{graph_navjr,graphjournal_tc}. For boolean queries on unbalebed graphs we have the collapse $\Lang(\transop) \bool\LeqExpr \Lang()$~\cite{graph_amai}. An edge $A$\tikz[baseline=-0.5ex]{\path (0,0) edge[->] (0.5,0);}$B$ indicates $A \LeqExpr B$ and $B\nLeqExpr A$.}\label{fig:graph_results}
\end{figure}
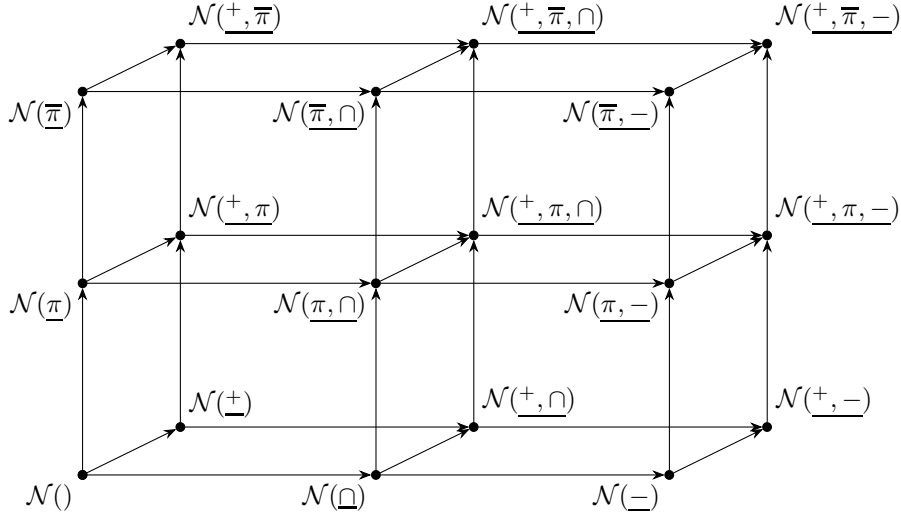

Observe that, on graphs, we have separation results in almost all cases, the only exception being the boolean equivalence of the fragments $\Lang(\transop)$ and $\Lang()$ on unlabeled graphs. These known separation results were all proven on general graphs. A major contribution of our work is strengthening several of these separation results to also cover much simpler classes of graphs (trees and chains). Moreover, we have shown that the navigational query languages behave, in many cases, very differently on trees and chains than they do on graphs, resulting in the major collapses in the expressive power of the navigational query languages that we have proven in this work.

\section{Conclusions and Directions for Future Work}\label{sec:conclusion}
This paper studies the expressive power of the downward navigational query languages on trees and chains, both in the labeled and in the unlabeled case. We are able to present the complete Hasse diagrams of relative expressiveness, visualized in Figure~\ref{fig:main_results}. In particular, our results show, for each fragment of the navigational query languages that we study, whether it is closed under difference and intersection when applied on trees. These results are proven using the concept of condition automata to represent and manipulate navigational expressions. We also use condition automata to show that, on labeled chains, projections do not provide additional expressive power for boolean queries. 

The next step in this line of research is to explore common non-downward operators, starting with node inequality via the diversity operator and the converse of the edge relation (which provides, among other things, the parent axis of XPath, and, in combination with transitive closure, provides the ancestor axis). Particularly challenging are the interactions between $\difference$ and $\diversity$. We conjecture, for example, that $\Lang(\transop, \difference, \diversity) \bool\nLeqExpr \Lang(\transop, \coproop, \intersect, \diversity) $, but this conjecture is still wide open, even on unlabeled chains.

Another direction is the study of languages with only one of the projections (or one of the coprojections) as Proposition~\ref{prop:proop_tree_partial_collapse} shows that in some cases adding only $\proop_1$ or only $\proop_2$ may affect the expressive power. Indeed, various XPath fragments and the nested RPQs only provide operators similar to $\proop_1$. Another interesting avenue of research is to explore the relation between the navigational expressions (and FO[3]) on restricted relational structures and FO[2], the language of first-order logic formulae using at most two variables~\cite{twovars}. Our results for $\Lang(\transop, \proop, \intersect)$ on unlabeled trees already hint at a collapse of FO[3] to FO[2] for boolean queries: the query $\Edges^k$ can easily be expressed in FO[2] algebras with semi-joins via $\Edges \ltimes (\dots \ltimes \Edges)$. A last avenue of research we wish to mention is to consider other semantics for query-equivalence and other tree data models, such as the root equivalence of Benedikt et al.~\cite{struct_xpath} and the ordered-sibling tree data model.

\bibliographystyle{plain}
\bibliography{casect}

\end{document}